\documentclass[11pt,runningheads]{llncs}

\usepackage[margin=1in]{geometry}

\usepackage[T1]{fontenc}
\usepackage{graphicx}
\usepackage{hyperref}

\usepackage{color}

\usepackage{booktabs} % For formal tables
\usepackage[linesnumbered,ruled,vlined]{algorithm2e}

\SetAlFnt{\small}
\SetAlCapFnt{\small}
\SetAlCapNameFnt{\small}
\SetAlCapHSkip{0pt}
\IncMargin{-\parindent}

\usepackage{iftex}
\usepackage{underscore} 
\usepackage{comment}
\usepackage[T1]{fontenc} 
\usepackage[english]{babel}
\usepackage{array}
\usepackage{multirow,bigdelim}
\usepackage{centernot}
\usepackage{musicography}
\usepackage{csquotes}

\usepackage{tikz}
\usetikzlibrary{arrows,chains,matrix,positioning,scopes}
\makeatletter
\tikzset{join/.code=\tikzset{after node path={%
\ifx\tikzchainprevious\pgfutil@empty\else(\tikzchainprevious)%
edge[every join]#1(\tikzchaincurrent)\fi}}}
\makeatother
\tikzset{>=stealth',every on chain/.append style={join},
         every join/.style={->}}
\tikzstyle{labeled}=[execute at begin node=$\scriptstyle,
   execute at end node=$]

\usepackage{booktabs} % For formal 

\usepackage{amsmath, amsfonts}
\usepackage{mathtools}
\usepackage{thm-restate}
\usepackage{color}
\usepackage{lscape}
\usepackage{soul}
\usepackage{bm, bbm}
\usepackage{enumerate}
\usepackage{tabto}

\newcommand{\Xcomment}[1]{{}}

\DeclareMathOperator*{\argmin}{arg\,min}
\DeclareMathOperator*{\argmax}{arg\,max}

\newcommand{\noaccents}[1]{#1}
\newcommand{\newagentvar}[3][\noaccents]{%
\expandafter\newcommand\expandafter{\csname #2\endcsname}{#1{#3}}%
\expandafter\newcommand\expandafter{\csname #2s\endcsname}{#1{\boldsymbol{#3}}}%
\expandafter\newcommand\expandafter{\csname #2smi\endcsname}[1][i]{#1{\boldsymbol{#3}}_{-##1}}%
\expandafter\newcommand\expandafter{\csname #2i\endcsname}[1][i]{#1{#3}_{##1}}%
\expandafter\newcommand\expandafter{\csname #2ith\endcsname}[1][i]{#1{#3}_{(##1)}}%
}

\newcommand{\newvecagentvar}[3][\noaccents]{%
\expandafter\newcommand\expandafter{\csname #2\endcsname}{#1{\boldsymbol{#3}}}%
\expandafter\newcommand\expandafter{\csname #2s\endcsname}{#1{\boldsymbol{#3}}}%
\expandafter\newcommand\expandafter{\csname #2smi\endcsname}[1][i]{#1{\boldsymbol{#3}}_{-##1}}%
\expandafter\newcommand\expandafter{\csname #2i\endcsname}[1][i]{#1{\boldsymbol{#3}}_{##1}}%
\expandafter\newcommand\expandafter{\csname #2ith\endcsname}[1][i]{#1{#3}_{(##1)}}%
}

\newagentvar{val}{v}
\newagentvar{bid}{b}
\newagentvar{dist}{F}
\newagentvar{alloc}{x}
\newagentvar{util}{u}
\newagentvar{pay}{p}
\newagentvar{ability}{a}
\newagentvar{effort}{e}

\newcommand{\yotam}[1]{{\color{red}{Yotam: #1}}}

\newcommand{\items}{\mathcal{M}} 
\newcommand{\agents}{\mathcal{N}} 
\newcommand{\sellable}{\mathcal{S}} 

\usepackage[]{color-edits}

\addauthor{UF}{red}

\newcommand{\ufe}[1]{{\UFedit{#1}}}

\addauthor{YG}{blue}

\begin{document}

\title{%To Sell or Not to Sell:  
Fair Allocation with Optional Selling}

% 1. List authors and use \inst{} to map them to the \institute list below.
% Use \and to separate multiple authors.
\author{Uriel Feige \and 
        Yotam Gafni %\inst{1}
        }

% 2. List the affiliations in the exact order of the \inst{} numbers.
% Use \and to separate different institutes. 
\institute{Weizmann Institute of Science \\
\email{\{uriel.feige, yotam.gafni\}@weizmann.ac.il}
}
\maketitle

\begin{abstract}
We consider fair allocation of indivisible goods in a setting in which agents have subjective valuation functions over the set of goods, and in addition, goods may be sold at given market prices.
In this setting, a fair allocation involves 
{deciding which goods to sell, how to allocate the unsold goods, and how to divide the money received from the sold goods.} We adapt to this setting the definitions of share-based fairness notions, such as the maximin share (MMS) and the truncated proportional share (TPS), and comparison-based fairness notions such as EF1 and EFX (which we adapt to SEF1 and SEFX). 
We show the following results when the utility of each agent is additive both over goods and over money.  With two agents, there are allocations that are simultaneously MMS and SEFX. With three agents, there are instances in which no allocation gives every agent more than $\frac{11}{12}$-MMS. With any number of agents, there are $\frac{2}{3}$-MMS allocations. There also are allocations that are simultaneously SEFX and $\frac{n}{2n-1}$-TPS. This latter ratio is best possible, even without the SEFX requirement.

\end{abstract}

%\def\authorrunning{}

%\settopmatter{printacmref=true}

\allowdisplaybreaks
\raggedbottom

 %\begin{document}

%\begin{titlepage}

% Optionally include a table of contents
%\enlargethispage{3\baselineskip}

\setcounter{tocdepth}{2} % adjust to 1 if desired
%\tableofcontents

%\end{titlepage}

\section{Introduction and Preliminaries}

The problem of Fair Division, popularized by the Cut \& Choose cake-cutting procedure \cite{steinhaus1948problem} has seen immense theoretical attention and progress in the last decade, with applications in diverse fields such as rent and land division \cite{gal2017fairest,segal2017fair}, course allocation \cite{budish2011combinatorial}, and household chores management \cite{igarashi2023kajibuntan}. However, one assumption that is absent in the core-model of indivisible items is that of an outside option, in the form of a market price. 
Consider, for example, a couple going through divorce with a single shared property. If we consider the house as an indivisible good, then there is no fair settlement: Only one of the partners can have it. However, a natural solution is to sell the house and split the proceeds equally. In formal fair allocation terms, we could say that in the purely indivisible setting, the Maximin Share (MMS) of each agent is $0$, but with selling, their maximin share is half of the house's market-price. Since selling only opens up more options, a definition of MMS for this setting always (weakly) outperforms the MMS guarantee in the purely indivisible setting. As we show, this superior guarantee is provably harder to approximate in some settings than the usual MMS without selling (see Table~\ref{tab:results}). Generally, a useful distinction in fair allocation problems is between \textit{share}-based and \textit{envy}-based notions, where the former are minimal utility guarantees each agent can expect, while the latter is a comparison-based guarantee for how agents view other's bundles in the final outcome. Since in the setting we introduce with selling, there is always the option of selling all items and sharing the proceeds equally, achieving an envy-free outcome is straight-forward. However, when the subjective valuations far supersedes the market price, this results in arbitrarily bad share guarantees. Thus, many of our results aim at a ``Best of Both Worlds'' guarantee: Optimal share approximation, together with a strong envy-free notion, adapted to our setting (SEFX).  

\subsection{Our Fairness Notions}

%\ufc{Define and motivate our versions of MMS and SEFX, SEFL. Say that our definitions are related to EFM and EFXM of two papers by Bei etal. This is part of our contribution. EF relaxations should not just be relaxations. They should be fair, and consistent with the intended use of selling so as to minimize envy. Hence, for example, if there is one item and it is sellable, our intention in envy based notions is that it must be sold, and any relaxation of EF that allows us not to sell it is not fair (it is a relaxation that is hard to justify).}

%\ygc{I didn't mention it's related to EFM or EFXM, because this appears in the related work, and also it is fundamentally different than their notion, that doesn't consider the selling counter-factual.}

%\ufc{Either define TPS here, or when it is first used, refer the reader to where the definition is}

In this section we introduce and motivate the fairness notions that we consider. They are adaptations of fairness notions that are commonly used in settings without sellable goods. The main share-based fairness notion that we will consider is the {\em maximin share} (MMS), and we shall also consider
the {\em truncated proportional share} (TPS). For envy-based fairness notions, we consider SEFX and SEF1, which are relaxations of envy-freeness adapted to our setting.

We consider a fair allocation problem with a set $\agents = \{1, 2, \ldots, n\}$ of agents and a set $\items = \{g_1, g_2, \ldots, g_m\}$ of indivisible goods. We may use the term ``item'' instead of ``good''. Each agent $i$ has a subjective valuation $v_i:2^{\items}\rightarrow \mathbb{R}_{\geq 0}$. 
We assume that valuations are additive ($v(S) = \sum_{g\in S} v(g)$) and non-negative. %\ufc{I think that here we should use standard definition for valuation (over all subsets)} 
In addition, there are market prices for the goods $p:\items \rightarrow \mathbb{R}_{\geq 0}$. Let $\sellable \subseteq \items$ be the subset of goods which are sold in a certain outcome. Given the sold goods $\sellable$, we get $\sum_{g \in \sellable} p(g)$ in sale proceeds,
%\ufc{we should say explicitly that the term sale proceeds refers to money} 
and are left with the unsold goods. We assume the sale proceeds are in the form of money, meaning they are divisible and have the same (additive) value for all agents. 
We can then decide a partition $((A_1, P_1), \ldots, (A_n, P_n))_{\sellable}$ that allocates the unsold goods and the money from the sale. For feasibility, the following condition must hold w.r.t. $(\sellable,A,P)$:

\begin{itemize}

\item All allocated kept goods are not sold: %\ufc{reformulate clearly. Perhaps use the notation $\sellable$ for the sold goods, rather than $\xi$}
\begin{equation}
\label{eq:kept_goods_not_sold}
    \sellable \cap (\cup_{i=1}^n A_i) = \emptyset,
\end{equation}

\item Distributed sale proceeds do not exceed the proceeds from sold goods:
\begin{equation}
\label{eq:sale_proceeds_feasibility}
\sum_{i=1}^n P_i \leq \sum_{g_j \in \sellable} p(g_j).\end{equation}
\end{itemize}

Moreover, if the allocation is a \textit{full} allocation (rather than a \textit{partial} allocation), it must hold that all goods are either allocated or sold, and the distributed sale proceeds equal the proceeds from sold goods:

\begin{equation}
\label{eq:all_non_sold_goods_allocated}
\sellable \cup \bigcup_{i=1}^n A_i = \items , \quad \sum_{i=1}^n P_i = \sum_{g_j \in \sellable} p(g_j).\end{equation}

We assume the agents utilities are \textit{additive} w.r.t. both their subjective valuations and the sale proceeds, i.e., given a bundle $B_j$ that consists of allocated kept goods $A_j$ and sale proceeds $P_j$, the agent utility is:

\begin{equation}
\label{eq:utility}
u_i(A_j,P_j) = \sum_{g\in A_j} v_i(g) + P_j.\end{equation}

%{We present here the definition for the MMS that applies when agents have linear marginal utility for money (sales proceeds). }

\begin{definition}
    \textit{Proportional Share (PS).}  $$PS(\items,v,p, n) = \frac{1}{n} \sum_{j\in \items} \max \{p(j), v(j)\}.$$
\end{definition}

Even in the setting without selling, the proportional share is not attainable or approximable (consider one item and two agents with low market price and very high subjective valuations). 

\begin{definition}
\label{def:MMS}
For agent $i \in \agents$, the maximin share (MMS) is the maximum value that agent $i$ can guarantee for themselves if they were to optimally sell the goods, and then partition the remaining goods and sale proceeds into $n$ bundles, where the agent receives the least valuable bundle according to their own valuation. Formally:

\begin{equation}
\text{MMS}_i(\items,v,p,n) = \max_{(\sellable,A,P)} \min_{j \in \{1, 2, \ldots, n\}} u_i(A_j, P_j),
\end{equation}

\noindent where the maximum is taken over all valid outcomes $(\sellable,A,P)$.
When the parameters $\items, v, p, n$ are clear from context, we use $MMS_i$ in short. 
We say that an outcome is $\rho$-MMS for some $0 \leq \rho \leq 1$ if each agent $i$ gets a utility of at least $\rho \cdot MMS_i$ in that outcome. 
\end{definition}

Observe that unlike the definition of MMS in settings without sellable goods, in Definition~\ref{def:MMS} the agent gets to optimize two aspects: both the decision of which goods to sell, and the decision of how to do the partition for the remaining goods (and sale proceeds). For the purpose of determining the MMS value, it will always benefit an agent to sell those goods $e$ for which $p(e) \ge v_i(e)$. Moreover, it may also benefit the agent to sell some goods $e$ with $v_i(e) > p(e)$. For example, this is necessary in order to get positive MMS value whenever $m < n$.

\textit{Truncated Proportional Share} (TPS) was previously introduced for the case of indivisible goods without selling \cite{bestBothWorlds}. We suggest the following extension to our case:

\begin{definition}
\label{def:tps_selling}
    \textit{Truncated Proportional Share With Selling (from here onwards, TPS):} With a set of agents $\agents$, set of goods $\items$, a valuation function $v$ and a market-price $p$, define $TPS(\items,v,p, n)$ as the largest value $t$ so that $$\frac{1}{n} \sum_{j\in \items} \max \{p(j), \min \{v(j), t\}\} = t.$$ 

\end{definition}

We note that if $p = 0$, the definition takes the same form as TPS for the setting without selling. We also show:

\begin{restatable}{lemma}{MMSTPSPS}
\label{lem:mms_tps_ps}
    $MMS(\items,v,p, n) \leq TPS(\items,v,p, n) \leq PS(\items, v,p,n)$. 
\end{restatable}

Before defining envy-based fairness notions, we introduce some notation. Given an additive valuation function $v_i$ and item price function $p$, 
%(recall that for simplicity we assume that all goods are sellable), 
we define a modified valuation function $\bar{v}_i$ as follows. For every good $g \in \items$, $\bar{v}_i(g) = \max \{p(g), v_i(g)\}$, and for every set $S \subset \items$, $\bar{v}_i(S) = \sum_{g\in S} \bar{v}_i(g)$.
%The fair allocation literature commonly distinguishes \textit{share} and \textit{envy-freeness} notions as the two main approaches towards achieving fairness. Consider the second approach. For the purpose of these definitions, for any good $g$ that is not sellable, we consider $p(g) = 0$. We let . 
The use of $\bar{v}_i$ is natural when we give the agent autonomy to decide if the unsold goods in their possession are kept or sold.

\begin{definition}
\label{def:EF}
    An allocation $(A_1, P_1), \ldots, (A_n, P_n)$ is {\em envy free} (EF) if for every agent $i$ and other agent $j$, {$\bar{v}_i(A_i) + P_i \ge \bar{v}_i(A_j) + P_j$.} %\ufc{what is $u_i$?} 
    {An allocation satisfies the following relaxations of EF if the above envy free condition holds whenever $P_j > 0$, and either the EF condition or the following conditions hold if $P_j = 0$.}
    \begin{enumerate}
        \item SEF1 (strong/sellable envy free up to one good). $\bar{v}_i(A_i) + P_i \ge \bar{v}_i(A_j \setminus \{g\}) + p(g)$ for some good $g \in A_j$.
        % EFL RELATED
        %\item $\epsilon$-SEFL (strong/sellable envy free up to one less preferred good). Either $|C_j| = 1$ and $\bar{v}_i(C_j) + \epsilon \ge p(C_j)$, or $u_i(C_i) + \epsilon \ge \bar{v}_i(C_j \setminus g) + p(g)$ for some item $g \in C_j$ satisfying $u_i(C_i) + \epsilon \ge \bar{v}_i(g)$. Since $\bar{v}_i(C_j) \ge p(C_j)$ by definition, the definition is equivalent to strengthening SEF1 with the additional requirement that either  $|C_j| = 1$, or $u_i(C_i) + \epsilon \ge \max_{g \in C_j} \bar{v}_i(g)$. 
         \item SEFX (strong/sellable envy free up to any good). $\bar{v}_i(A_i) + P_i \ge \bar{v}_i(A_j \setminus \{g\}) + p(g)$  for every good $g \in A_j$.
    \end{enumerate}
\end{definition}

Definition~\ref{def:EF} differs in several aspects from the standard definitions of envy-based fairness that are used when there are no sellable goods. To motivate these aspects, consider the simple setting of $n=2$ and $m=1$. Suppose that the single good $g$ has positive market price $p(g) > 0$. Then, if we wish the outcome to qualify as satisfying envy-based fairness notions, there seems to be only one acceptable outcome, which is to sell $g$ and give each agent $\frac{p(g)}{2}$. Consequently, we wish EF and any relaxation of EF (such as SEF1 and SEFX) to enforce this outcome, while a straight-forward application of EFX from the setting without selling would also accept allocating the good in full to one of the agents. Definition~\ref{def:EF} indeed does enforce this outcome (not just for EF, but also for SEFX and SEF1).

%The %strong/sellable versions of the 
%above definitions do not \textit{remove} goods (as is done in the corresponding EF1, EFX definitions), but \textit{sell} them. However, since non-sellable items are assumed to have price~0 for the purpose of the definition, if there are no sellable items, the %strong/sellable 
%definitions coincide with the EF1 or EFX definitions for the setting with no sellable items.

%We find this to be the natural definition for the setting with selling, out of the following reasoning: If $g$ was the last good remaining and was to be given either to $i$ of $j$, the commonly-used EF1 definition says that it is okay to give it to $j$. The SEF1 definition offers a third option (rather than allocating in full to one of the agents): that of selling the item. It is then not okay to give it to $j$ if selling it would give $i$ part of the money (one could think of it as a ``local MMS'' condition). 

%The above relaxations differ in what constitutes a legitimate complaint by $i$ against $j$. For SEF1, the complaint is that the minimum of $i$ and $j$ (according to $\bar{v}_i$) would increase if any one good $g \in C_j$ chosen by $i$ is sold. For SEFX, the complaint is that the minimum of $i$ and $j$ would increase if some good $g \in C_j$ chosen by $j$ is sold.

\subsection{Our Results}
\label{sec:technical_overview}

%\ufc{With the page limitations, makes sense to have a section of our results, and separate section of technical overview, and then discussion, and proofs in the appendix.}

All results in this section are stated for additive valuations in the setting of indivisible goods with optional selling.

We first present combined share \& envy notions results, where we combine tight share approximation guarantees with our strongest approximation notion of SEFX:

%\ufc{Note that EF is trivial, but may give very little value. So one theme is to instead give SEFX. but while guaranteeing high value. For the SEFX theorems, perhaps mention SEFX first, and MMS or TPS second.}

%With two agents, in the setting without selling, the MMS can be guaranteed by the Cut \& Choose procedure: One of the agents suggests a partition of the indivisible goods (``cut''), and the other agent chooses their preferred bundle in the partition (``choose''). However, once optional selling is introduced, this procedure fails. This is because the agent that cuts may require (as part of their partition) to sell a good that the other agent highly values, leaving the choosing agent no good choice. In Appendix~\ref{sec:two_agents}, we show that Cut \& Choose guarantees $\frac{1}{2}$-MMS, and provide examples that shows that the guarantee cannot be improved, no matter which agent cuts.
%Despite this hurdle, we design an allocation algorithm that gives the full MMS guarantee for two agents. Moreover, the resulting allocation is also SEFX.

\begin{restatable}{theorem}{TwoSEFX}
    \label{thm:2EFX}
    For two agents with additive valuations, there is an allocation that is simultaneously SEFX and MMS. 
    
     %\ufc{Consider removing this part from the theorem, and mention it later only in the proofs.} Moreover, if all goods have strictly positive value for all agents (in terms of $\bar{v}_i$), then \textit{every} MMS allocation is SEFX. 
\end{restatable}

%\begin{theorem}
%\label{thm:2agents}
%    Every allocation instance with $n = 2$ agents has an allocation that is simultaneuosly MMS and SEFX.
%\end{theorem}

In the setting without sellable items, a theorem analogous to Theorem~\ref{thm:2EFX} is proved via the Cut \& Choose allocation algorithm. In Appendix~\ref{sec:two_agents} we show that with sellable items, Cut \& Choose guarantees only $\frac{1}{2}$-MMS. To prove Theorem~\ref{thm:2EFX}, we introduce a {\em Cut} \& {\em Give} paradigm, which is a variation on Cut \& Choose. See more details in Section~\ref{sec:two_agents}.

%\ufc{The above replaces: Our allocation rule has an elegant reinterpretation of the Cut \& Choose algorithm, we call ``Cut \& Give'' (In Appendix~\ref{sec:two_agents}, we show that a straight-forward interpretation of Cut \& Choose guarantees only a tight $\frac{1}{2}$-MMS). 
%See more details in Section~\ref{sec:two_agents}.}
%The details of this construction are in Theorem~\ref{thm:2agents}. 

The other share notion we consider is the Truncated Proportional Share (TPS) \cite{bestBothWorlds}. In settings without sellable items, it was shown in~\cite{bestBothWorlds} that $\frac{n}{2n-1}$-TPS allocations always exist, and that $\frac{n}{2n-1}$ is best possible. 
Our main theorem concerning the TPS is the following.
%dominates MMS as a share notion for additive valuations, as it is always weakly higher.  For TPS we are able to obtain tight approximation results from the setting without selling , of $\frac{n}{2n-1}$, using an adjusted bag-filling algorithm. 

\begin{restatable}{theorem}{TPSAlloc}
\label{thm:TPSallocations}
    Every allocation instance has a $\frac{n}{2n-1}$-TPS allocation. Moreover, every allocation instance has an allocation that is both $\frac{n}{2n-1}$-TPS and SEFX.
\end{restatable}

%We note that the ratio $\frac{n}{2n-1}$ is best possible for this benchmark, as is already known from the setting without selling \cite{bestBothWorlds}, and this transfers to our setting. 

%The proof of Theorem~\ref{thm:TPSallocations} is based on adapting proof techniques that were previously used in settings without sellable items. \ygc{Remove? Unlike the proof of Theorem~\ref{thm:nMMS}, here the adaptation does not require ideas that substantially differ from ideas appearing in other proofs in our paper.}

For our other main results, we set aside achieving combined share and envy notions guarantees, and focus purely on MMS approximation. As in the case without sellable items, with $n=3$, MMS cannot be guaranteed. 
%For a general number of agents $n\geq 3$, our negative result of Proposition~\ref{prop:upper_bound} shows that the full MMS can not be guaranteed. 

\begin{restatable}{proposition}{UpperBound}
\label{thm:upper_bound}
    For $n=3$ agents with additive valuations over $m=8$ indivisible goods:
    \begin{enumerate}
        \item With sellable goods, for every $\epsilon$, there is an allocation instance in which no allocation gives every agent at least $\frac{11}{12} + \epsilon$ of her MMS.
        %\item Without sellable goods, there is always an allocation that gives every agent her full MMS (As shown by \cite{negativeMMSExample}). \ufc{Consider removing this part from the theorem, and only mention it in the paragraph afterwards.}
    \end{enumerate}
\end{restatable}

It is instructive to note that the approximation gap of $\frac{11}{12}$ in Proposition~\ref{thm:upper_bound} is wider than the currently known MMS approximation gap in the setting without selling, which is $\frac{39}{40}$. Moreover, as the result in \cite{negativeMMSExample} shows, for these parameters ($n=3, m = 8$) it is possible to guarantee the full MMS in the setting without selling, which proves a gap between the two settings. 

Given that full MMS cannot always be guaranteed, we consider approximations. The following result is %perhaps 
the most technically involved among the results in this paper.

\begin{restatable}{theorem}{nMMS}
\label{thm:nMMS}
    For every $n \ge 4$ and every allocation instance, there is a $\frac{2}{3}$-MMS allocation. For $n=3$, there is a $\frac{3}{4}$-MMS allocation.  
\end{restatable}

The proof of Theorem~\ref{thm:nMMS} is based on an {approach referred to as {``Lone Divider"~\cite{steinhaus1948problem,Kuhn1967,Dawson2001}} %\cite{fairEnough}, 
that guarantees $(\frac{2}{3} + \Omega(\frac{1}{n}))$-MMS in the setting {without selling~\cite{procaccia2014fair}}.} The algorithm works in rounds. In each round it allocates bundles that are acceptable (guarantee at least $\frac{2}{3}$-MMS) to those agents that receive them, but are unacceptable to each remaining agent $i$. To prove that eventually agent $i$ gets an acceptable bundle, one uses the fact that bundles allocated in previous rounds have only small value according to $v_i$, so sufficient value remains in the set of yet unallocated items, enabling the formation of bundles that are acceptable for $i$. However, with selling, it is not only that bundles were allocated in previous rounds, but also goods were sold. If goods of high value under $v_i$ were sold for a low price, then it may be that goods that remain towards later rounds do not have sufficient value for $i$ to get an acceptable bundle.

Overcoming this difficulty involves multiple new ideas: 

{
\begin{enumerate}
\item Handle items of large value in a preliminary phase. This is a standard component in other allocation algorithms, but not needed in the lone divider approach, in the setting without selling. 

\item  Choose the Lone Divider agent carefully (in our case, as the agent with highest MMS). Without sellable items, any agent can be the divider, but with sellable items, this is not true (as we show even for the $n=2$ case, see Example~\ref{ex:bad_cut_choose}).

\item Insist on a canonical partition, rather than any partition. In this partition, all bundles have convenient structure, except for at most one (referred to as a ``leftovers'' bundle).

\item In the matching phase of Lone Divider, we exclude the divider and the leftovers bundle, and add them to the matching when appropriate. 

\end{enumerate}

Points $3$ and $4$ are our major departure from the standard lone divider approach, and may turn out to be useful also in other settings.
}

{With the above ideas and an intricate proof, we recover nearly the same approximation ratio as without sellable items. Specifically, we get $\frac{2}{3}$-MMS approximation, without the additional $\Omega(\frac{1}{n})$ term. We provide a detailed overview of these ideas in Section~\ref{sec:general_n}, and the full proof is in Appendix~\ref{sec:mms_23_full_proof}.

For the important special case of $n=3$, we are able to fully recover the Lone Divider approximation ratio in the setting without selling, of $\frac{3}{4}$-MMS. To do so, we combine the principles developed for general $n$ (implementing them becomes somewhat simpler, as $n$ is so small) together with the $n=2$ full-MMS algorithm, and get an improved $\frac{3}{4}$-MMS approximation. The details appear in Appendix~\ref{sec:mms_23_full_proof}. }

%\ufc{For polytime computations, maybe degrade theorems to propositions, one for TPS, one for MMS. Then say what they imply about each previous theorem separately. In theorem~\ref{thm:2EFX}, do we get approximate MMS with exact EFX? In Theorem~\ref{thm:TPSallocations}, is the first part polytime? In the second part, do we get the TPS approximation with an $\epsilon$ relaxation for SEFX? Or with no relaxation? In Theorem~\ref{thm:23}, do we lose $\epsilon$?}

Beyond our existence results, we are also interested in polytime computation. Without sellable items, it is known that computing the MMS value of a single agent is weakly NP-hard for $n=2$ and strongly NP-hard for general $n$. However, a polynomial time approximation scheme (PTAS) for calculating MMS without selling is known, and we show one for our setting as well. We also show that when $n$ is constant, the problem has a pseudo-polynomial algorithm (for both results, see Appendix~\ref{sec:computation}). Our PTAS, in turn, means that our existence results for MMS approximation can be appropriately made into a computational result:

\begin{restatable}{proposition}{MMScomputational}
\label{prop:mms_comp_results}
An allocation as guaranteed by Theorem~\ref{thm:2EFX} can be computed in pseudo-polynomial time. 
%\ufc{remove: (that guarantees SEFX as well)}, the PTAS \ufc{FPTAS?} also guarantees $\epsilon$-SEFX. \ufc{why not SEFX?}
Each of our other MMS approximation existence results (Theorem~\ref{thm:nMMS}) that guarantees $\rho$-MMS approximation has a PTAS that guarantees $(\rho - \epsilon)$-MMS. 
\end{restatable}

For TPS, we show it can be computed in polynomial time, and moreover, the approximation algorithm that incorporates this calculated TPS value runs efficiently. Our proof for combining TPS with SEFX builds on an algorithm that combines TPS approximation with $\epsilon$-SEFX, where the potential argument it provides yields a fully polynomial time approximation scheme (FPTAS).

\begin{restatable}{proposition}{TPScomputational}
\label{prop:tps_comp_results}

An allocation that guarantees $\frac{n}{2n-1}$-TPS to each agent can be computed in polynomial time. An FPTAS exists for an allocation that is $\frac{n}{2n-1}$-TPS  and $\epsilon$-SEFX.   
\end{restatable}

\section{Related Work}
\label{sec:related_work}

%\ufc{Better to not start with this, but to end with saying that additional related work is discussed in Section~\ref{} in the appendix. Our work is interlinked to many fair allocation works. We focus on related work that conceptually relates to an environment with market prices in this section, and overview other related work that relates to our bounds and techniques in the appendix. \ufc{last sentence, grammar can be improved.} }

%\ygc{Instead of the above, in the end say "additional related work on ... is in the appendix" also, if there's space, do all the related work in main text. }

%To summarize, 
Table~\ref{tab:results} provides an overview of our main results concerning approximate MMS allocations for goods, and compares them with known results without sellable items.

\begin{table}
    \renewcommand{\arraystretch}{2}
    % Added an extra column 'c|' at the end
    \begin{tabular}{cc|c|c|c|c|}
        & \multicolumn{1}{c}{} 
        & \multicolumn{1}{c}{MMS $n=2$} 
        & \multicolumn{1}{c}{MMS $n=3$} 
        & \multicolumn{1}{c}{MMS $n\geq 4$} 
        & \multicolumn{1}{c}{TPS ($n\geq 2$)} \\\cline{3-6}
        \multirow{2}{*}{} & No Selling 
        & $1$  
        & $[\frac{11}{12}, \frac{39}{40}]$  
        & $[\frac{7}{9}, 1 - \frac{1}{n^4})$
        & $\frac{n}{2n-1}$  \\\cline{3-6}
        & With Selling 
        & $1$  
        & $[\frac{3}{4}, \frac{11}{12}]$  
        & $[\frac{2}{3}, 1 - \frac{1}{n^4})$  
        & $\frac{n}{2n-1}$  \\\cline{3-6}
    \end{tabular}
\caption{Share approximation results for our new setting of goods with selling, with a comparison to the current state of the art for indivisible goods (without selling). The full MMS guarantee with $n=2$, and the TPS guarantee, are tight, and we can guarantee that the allocations are SEFX. 
    %SEFL (Theorem~\ref{thm:halfEFX}, Lemma~\ref{lem:extend_sefl}). 
    The No Selling results are due, respectively, %\ufc{the EC format changes the order of references} 
    to \cite{budish2011combinatorial}, \cite{feigeNorkin},  \cite{negativeMMSExample}, \cite{huang2025fptas79approximationmaximinshare}, \cite{negativeMMSExample}, and \cite{bestBothWorlds}. The upper bound for $n\geq 4$ with selling, is not a new result, and follows simply because our model generalizes the setting without selling. Other than that, we show the With Selling results, respectively, in Theorem~\ref{thm:2EFX}), Theorem~\ref{thm:nMMS}, Proposition~\ref{thm:upper_bound}, Theorem~\ref{thm:nMMS}, Theorem~\ref{thm:TPSallocations}. }     
 %\ufc{for $n \ge 4$, the negative results are not as strong as those in the table}
    \label{tab:results}
\end{table}

%\ufc{In all cases of comparing with a different work, first define the model of the other work, and only then start comparing. Do not start by saying that models are similar, as this might be misunderstood to mean identical.}

 %\ufc{a sentence cannot start with [28]} 
A model with sellable additive goods and $n=2$ agents is introduced in \cite{procaccia2014}. In this model, the market price has a specific form that depends on the subjective valuations. Subjective valuations are also normalized so that their total sum is $1$. They focus on studying the \textit{price of fairness}, meaning the loss of welfare due to restricting attention to envy-free allocations.
Algorithm aspects of the \textit{price of fairness} problem in a model of sellable goods are considered by \cite{bilo2024}. The authors show the price of fairness is NP-hard to compute with $n=2$, give a PTAS for the case of $n=2$ and identical valuations, and a PTAS and pseudo-polynomial algorithm when agent subjective valuations do not vary too much (i.e., there is a constant bound between the maximal and minimal positive valuations). 

A setting in which some goods are divisible and some are indivisible is considered in \cite{bei2021fair}. In this setting, divisibility is fixed exogenously and does not depend on agents' preferences.  %divisible and indivisible goods.
A main difference from our model is that in our model the endogenous selling choice during allocation decides which items will be treated as divisible and which as indivisible. %\ufc{THink carefully what t say and rephrase} \ygc{Because now we dont have 'sellable' and 'unsellable': So we need a new positioning w.r.t. them} 
Another is that in our model, once a good is sold, all agents have exactly the same value for the resulting divisible good, which is the sales proceed, and this value is independent (and possibly different) from their subjective values for the unsold good. In contrast, for a divisible good in \cite{bei2021fair}, each agent has a subjective value for it, and moreover, this value is the respective fraction of the subjective value of the agent for the whole good.
The work adapts the notions of EF1 and EFX to the mixed setting, naming the new notions EFM and EFXM. For these new notions, if agent $j$ holds part of a divisible good, then agent $i$ must not envy $j$, and otherwise, the notion behaves like EF1 (or EFX, respectively).  Our definitions for SEF1 and SEFX have a similar flavor (though are technically different, as our setting is different).
%\ufc{removed: However, in \cite{bei2021fair}, the divisible/indivisible dichotomy is given exogenously, and is not part of the design space. }
Within the same mixed divisible and indivisible goods 
model, it can be shown \cite{beiMMS} that the MMS approximation of the worst mixed divisible and indivisible instance is not worse than the worst MMS approximation of an indivisible goods instance. %\ufc{I did not understand the last sentence, given that the next sentence mentions a ratio of 1/2.} 
A complementary result is an MMS allocation algorithm, which gives an approximation $\alpha$ ranging from the MMS approximation for indivisible goods and up to $1$, where $\alpha$ is monotonically increasing in the minimum (among all agents) of the utility of the divisible part relative to the indivisible part. 
%with $\alpha \in [\frac{1}{2}, 1]$, where $\alpha$ is monotonically increasing in the size \ufc{what does size mean? Each agent sees different values} of the divisible goods.  %\ufc{What does "cake" mean? How come the approximation ratio is as bad as 1/2? Either explain or remove the last sentence.} 

%We draw on their definition for envy-notions in the mixed setting (EFM) when defining our notion of SEFX (sellable envy-free up to any good). \cite{bei2021fair} also shows how to extend an indivisible-goods EFX allocation to an EFM allocation, given additional divisible goods, and we use this for our envy results. 
  The notion of  \textit{subjective divisibility} is introduced in \cite{bei2025SubjectiveDivisibility}. In this setting, all items are divisible. However, the valuations of agents are such the each agent views some of the items as indivisible, meaning that she herself gets no value from fractions of the item (though other agents might value fractions of that item), and other items as divisible, meaning that a $\beta$-fraction of the item gives a $\beta$-fraction of the value of the whole item, for all $0 \le \beta \le 1$. %some agents may consider a certain good indivisible, while others consider it divisible. If an agent does not view a good as divisible, they get no value from a partial allocation of it. 
 In the subjective divisibility setting with $n=2$ and $n=3$ agents, a tight $\frac{2}{3}$-MMS can be guaranteed, and with general $n$, they show a procedure to guarantee $\frac{1}{2}$-MMS. Compared to the mixed setting described above, {subjective divisibility} captures better some difficulties that arise also in our model (some agents want to divide a good, some do not, mirroring a dilemma regarding selling a good that appears in our model), but in some respects differs more from our model, in that sales proceeds can be enjoyed by all agents, whereas fractions of a good can be enjoyed only by agents that view it as divisible.  
 %\ufc{remove: is similar to our model in that the allocator has an endogenous choice of whether to divide the goods. It differs from our model in that for us, the divisible choice (market-price) is objective, and also pertinent to all agents. Our envy notions have the novelty that goods that are not sold (divided) as part of the allocation, but are sellable, are required to guarantee envy-freeness if they are sold. Thus, SEF1 is stronger than EFM, and SEFX is stronger than EFXM.} 

%\ufc{the above is an important section that needs to be written more clearly, without assuming that the reader is familiar with the previous model, and without assuming that the reader can infer that there is a difference between our SEFX and the previous definition. Our model is of indivisible items. Items can only be allocated as a whole, or sold as a whole. The previous models are of divisible items (every item is divisible) but the valuations of agents over items are of one of two types. Important to see exactly how they defined EF1 and EFX. As a side remark, their set of allowable valuations is disconnected and not convex, whereas our is connected and convex.} \ygc{I felt the distinctions here are a bit confusing, but anyway I explained and compared. }

 A model where both subjective valuations and market prices exist for the goods, but where goods are not sold is considered in \cite{fairDivisionMarketValues}. Rather, they aim for an {allocation} that is fair both towards market values, and towards the subjective valuations. I.e., the allocation should be fair both in ``objective'' terms, and in the agents' subjective perception. 
The type of questions and results that arise from this dualistic approach are different than our combined (``monistic'') approach. For example, one result is that it is impossible to achieve $\frac{1}{n}$-MMS for the subjective values and EF1 for the objective values, or vice-versa, where the MMS is computed using only the subjective (or only the objective) values. In our approach, the MMS is computed using both subjective and objective values, and we are able to guarantee a good approximation of it. %The only exception in our work is Appendix~\ref{sec:equal_proceeds}, where we impose a secondary fairness restriction, and indeed find bad approximation. \ufc{remove last sentence?}

Another setting that intertwines market values and subjective preferences is considered in \cite{dallaglio}. However, the subjective valuation take a specific form of ``correction'' to the market values, to allow for easy elicitation. The solution concept and discussion are tailored for their setting. 

Fair allocation when it is possible to make monetary transfers between the agents is considered in \cite{bogomolnaia2025fair}. Introducing this option results in different auction-like mechanisms. Our setting is more limited in that agents do not have ``outside money'' that can streamline the fair allocation. For example, consider a divorce setting: \textit{All} the partners' assets are being redistributed, and so, while they can sell existing goods and share the proceeds, they do not have extra-marital money to facilitate compensation payments. Similarly, inheritors / stakeholders in a liquidated company may have outside assets that are negligible compared to the value of assets of the company itself.

{%In the indivisible goods setting, MMS approximation has improved significantly beyond the $\frac{2}{3}$ approximation. 
In the setting of indivisible goods, the notions of {\em maximin share} (MMS)   and {\em envy-freeness up to one good} (EF1) are attributed to~\cite{budish2011combinatorial}. A line of works \cite{fairEnough,ghodsi2018,MMSImproved,MMSImprovement,MMSImproving,heidari2025improvedmaximinshareguarantee,huang2025fptas79approximationmaximinshare} develops allocations with increasingly better MMS approximations, going from $\frac{2}{3}$ to $\frac{7}{9}$.
  Envy-freeness up to \textit{any} good (EFX) was first introduced in \cite{caragiannis2019unreasonable}. EFX was shown to exist in several notable settings, such as with $n=3$ additive agents \cite{efxExistsThree}, lexiocographic preferences \cite{unifiedGoodsChores},
and agents with identical valuations \cite{plautRoughgarden}. }

\section{Overview of Proofs}

\subsection{MMS + SEFX for \texorpdfstring{$n=2$}{2} Agents}
\label{sec:two_agents}

In this section we restate and prove Theorem~\ref{thm:2EFX}. Our proof is based on a Cut \& Give protocol, which is a variation on the well known Cut \& Choose protocol. We remark that for the giver in the protocol, we get the stronger EF property, and not just SEFX, and for the cutter, we get EF in interesting special cases. 

\TwoSEFX*

\begin{proof}
    We may assume that in her own MMS partition, each agent sells at most one good (see Lemma~\ref{lem:number_of_split_goods}). We refer to the agent whose sold good has the lower price (or that does not sell a good, and breaking ties arbitrarily) as the {\em Cutter}, and to the other agent as the {\em Giver}. We refer to the goods that they intend to sell in their MMS partitions as $g_C$ and $g_G$, respectively. If one of the agents does not sell a good, or if $g_C = g_G$, then the standard Cut \& Choose protocol proves the theorem. So, we assume that $g_C \not= g_G$, with $p(g_C) \le p(g_G)$.

    Let $(A_1, P_1), (A_2, P_2)$ be the MMS partition of Cutter, with $P_1 + P_2 = p(G_C)$, and assume without loss of generality that $g_G \in A_1$. Being her MMS partition, getting any of these two parts also ensures for Cutter the SEFX property. (In fact, if $g_C$ is non-empty, meaning that Cutter really intends to sell a good, she gets EF, due to Lemma~\ref{lem:split_equal_mms}.) 
    
    If at least one of the two parts is acceptable to Giver (gives her at least her MMS value), then Giver takes the part that she prefers, ensuring both MMS and EF for Giver. Cutter gets the other part.

    If no part is acceptable for Giver, then she gives $(A_2,P_2)$ to Cutter. However, Giver does not take $(A_1,P_1)$ for herself (and for this reason, we refer to the protocol as Cut \& Give, not Cut \& Choose). Instead, she sells $g_G$ and cancels the sale of $g_C$, and keeps for herself $(A'_1, P'_1)$, with $A'_1 = A_1 \cup \{g_C\} \setminus \{g_G\}$ and $P'_1 = p(g_G) - P_2$. Importantly, $P'_1 \ge 0$, because $p(g_G) \ge p(g_C)$. Giver gets at least her MMS value, because the total value of the two parts $(A'_1, P'_1)$ and $(A_2, P_2)$ is at least twice her MMS (as they contain exactly the same items and sum of payments as the join of the two parts of her MMS partition), and she did not take the part that has value smaller than her MMS. For the same reason, the EF property is satisfied as well.      
\end{proof}

\begin{figure}[!htb]
\centering
\includegraphics[width=0.85\linewidth]{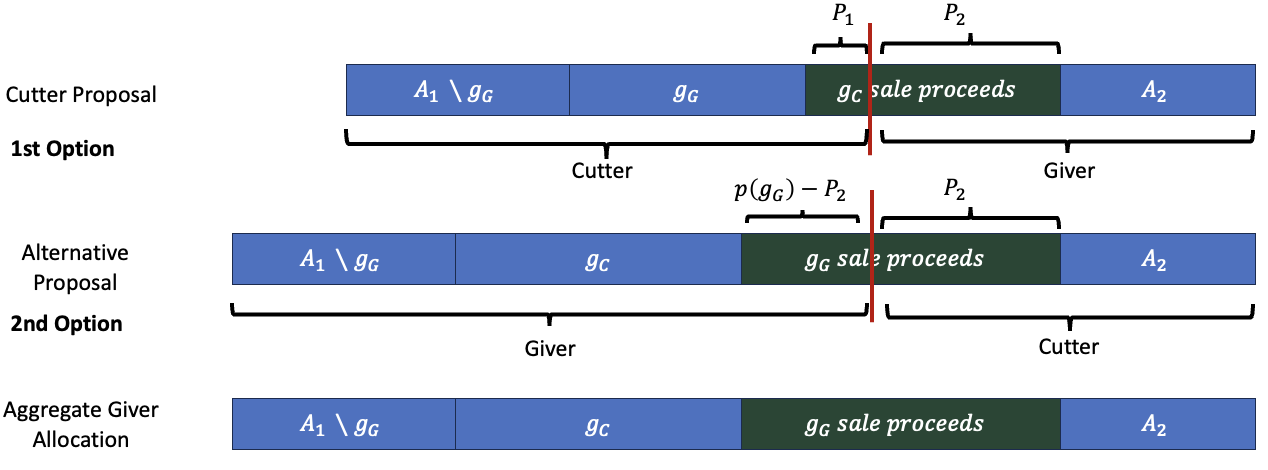}
\caption{A visual proof for the $n=2$ MMS existence result, which follows the outline of the proof of Theorem~\ref{thm:2EFX}.}
\label{fig:visual_proof_two_agents}
\end{figure}

\subsection{\texorpdfstring{$\frac{n}{2n-1}$}{n/(2n-1)}-TPS + SEFX for General \texorpdfstring{$n$}{n}}

The TPS (see Definition~\ref{def:tps_selling}) is efficiently computable, by an algorithm we describe in Appendix~\ref{app:tps}. 
We know that in the indivisible setting, it is impossible to generally achieve an approximation better than $\frac{n}{2n-1}$ to TPS. This immediately extends to our more general setting as well. We show how a combination of a ``moving knife'' and bag-filling procedures achieves the optimal $\frac{n}{2n-1}$-TPS approximation. The algorithm works as follows. First, if there is a good $g$ with a market price higher than $\frac{n}{2n-1}\cdot TPS_i$ for some agent $i$, we sell the good and let the agent $i$ with lowest MMS value take $\frac{n}{2n-1}\cdot TPS_i$ of the received money. We add the remaining sale proceeds from the good back as a virtual ``good'' (it is already sold, but for the purpose of the moving-knife loop, and later bag-filling, we treat it as one complete good). 
If we can satisfy another agent using the sale proceeds from at most one single good, we continue to do so (always prioritizing the agent requiring the least amount of sale proceeds), until all agents $i$ and goods $g$ satisfy $p(g) < \frac{n}{2n-1}TPS_i$. {Then, we allocate any item $g$ for which some agent $i$ has $v_i(g) \geq \frac{n}{2n-1}TPS_i$ to this agent.}
We then perform a bag-filling procedure into a bag $B$, where we add an arbitrary good at each step, %\ufc{do we need to first give single-item bundles?}
until some agent has a value of at least $\bar{v}_i(B) \geq \frac{n}{2n-1}TPS_i$, and then we give the bag to the agent, and repeat. For the full algorithm pseudo-code, see Algorithm~\ref{alg:APX-TPS} in the appendix. 

\begin{restatable}[]{lemma}{halfTPS}
\label{lem:halfTPS}
    In the setting with indivisible sellable goods and additive valuations, the algorithm described above for TPS approximation (with pseudo-code in Algorithm~\ref{alg:APX-TPS}) produces a partial %\ufc{somewhat strange to see "partial" here} 
    allocation that is  $\frac{n}{2n-1}$-TPS.
    %Algorithm~\ref{alg:APX-TPS}.  
\end{restatable}

%\ufc{Distinguish below between which written parts are full proofs, and which are only sketches for which a full proof appears in the appendix. The parts that are full proofs have to be written well.} \ygc{I added some more content around the bag-filling loop to Theorem 5.14's proof, but let's discuss how to make the proofs less sketchy}

In the appendix, we extend the algorithm to produce a (partial) allocation that is both a $\frac{n}{2n-1}$-TPS approximation and SEFX. We do so using the idea of ``stealing'' (with relevant precedents in \cite{akramiRathiMMS_EFX,feige2025residualmaximinshare}), where during the allocation process agents %that are already inactive
may steal a minimal subset of the bundle that is meant to be allocated {to some other agent}, if they are envious of it (in the SEFX sense), and release their current bundle. Implementing this process with sellable goods adds complexity {to the stealing process, and to its analysis. %because we may not always be able to allocate goods as a whole while satisfying SEFX. Solving this issue requires multiple caveats:
One modification that we make is not to implement SEFX directly, but instead implement $\epsilon$-SEFX, and to show that taking the limit as $\epsilon$ tends to~0 gives an SEFX allocation. Another principle that we use is that the stealing agent can sell only one item that was not sold in the bundle that she envied.
%\ufe{subdividing} at most one good when \ufe{stealing} the bundle, using , 
Another aspect that we take care of is money contained in the bundle released by the stealing agent, as this money came from a past selling of some good, and we might need to undo some of the past sales.}
%\ufc{Is it the same stealing as previous work, or a different version that involves also selling?} 

%\ufc{update statement of the next theorem}

\TPSAlloc*

%\begin{restatable}{theorem}{halfTPSSEFX}
%\label{thm:halfEFX}
%    In the setting with indivisible sellable goods and additive valuations, there is always a partial allocation that is $\epsilon$-SEFX and $\frac{n}{2n-1}$-TPS, given by Algorithm~\ref{alg:APX-TPS-EFX}. By Lemma~\ref{lem:limit} and Lemma~\ref{lem:extending_indivisible_EFX}, this implies a full SEFX and $\frac{n}{2n-1}$-TPS allocation. 
%\end{restatable}

\subsection{\texorpdfstring{$\frac{11}{12}$}{11/12}-MMS Gap for \texorpdfstring{$n=3$}{3 agents}}

Full MMS cannot be guaranteed with $n=3$ agents. Moreover, our result shows a quantitative gap between the MMS guarantees in the setting without sellable goods, and the one with sellable goods. 

\UpperBound*

\begin{proof}
 We say that an agent $i$ is {\em sell-seeking} if $v_i(e) \le p(e)$ for every good $e$. For a sell-seeking agent $i$, her MMS is equal to her proportional share. That is, $MMS_i = \frac{1}{n}\sum_{e\in \items} p(e)$. 

For $\epsilon > 0$, we say that an agent $i$ is {\em $\epsilon$-monetary} if $\sum_{e \in \items} p(e) \le \epsilon \cdot MMS_i$. If $\epsilon$ is sufficiently small, then $MMS_i$ is determined by $v_i$ alone, up to negligible terms. 

Suppose that there is one sell-seeking agent $i$, and all other agents are $\epsilon$-monetary. Then, up to $\epsilon$, a $\rho$-MMS allocation needs to give each of the $\epsilon$-monetary agents $\rho$-MMS through undivided goods, while the price of the sold goods needs to be at least $\rho \cdot MMS_i$. This is similar to the standard setting of $\rho$-MMS allocations for indivisible goods, with the following two changes:

\begin{enumerate}
    \item For the sell-seeking agent $i$, the valuation $v_i$ is replaced by the price function $p$.
    \item Agent $i$ needs to get at least $\rho$ times her proportional share, rather than $\rho$-times MMS.
\end{enumerate}

In \cite{feigeNorkin}, it was shown that for three agents, there is an allocation instance with \ufe{8} indivisible goods in which in every allocation, either at least one of the first two agents gets at most $\frac{11}{12}$ of her MMS, or the third agent gets at most $\frac{11}{12}$ of her proportional share. 
Concretely, the example in \cite{feigeNorkin} shows valuations $\tilde{v}_1, \tilde{v}_2, \tilde{v}_3$, so that either $\tilde{v}_1$ or $\tilde{v}_2$ get at most $\frac{11}{12}$ their MMS value, or $\tilde{v}_3$ get at most $\frac{11}{12}$ its proportional value. Moreover, in their instance, the MMS values for the two agents and proportional share for the third agent are all equal. 
Adjusting to our setting, let $p = \epsilon \cdot \tilde{v}_3, v_1 = \tilde{v}_1, v_2 = \tilde{v}_2, v_3 = 0$. Then, the MMS with selling of agent $3$ equals the proportional value of $\tilde{v}_3$, and the MMS with selling of agents $1,2$ equals (up to $\epsilon$) the MMS of $\tilde{v}_1, \tilde{v}_2$. We can assume any good sold goes fully to agent $3$, since it is negligible for the other agents. %\ufc{remove:  (since the proportional share of agent $3$ equals the MMS of agents $1$ and $2$,
%the sold good is worth at most $\epsilon \cdot \sum_{j\in \items} \tilde{v}_3(g_j) = \epsilon \cdot MMS(\items, v_i, p, 3)$ for $i\in \{1,2\}$.} 
Then, a partition that achieves a better than $\frac{11}{12}+\epsilon$ approximation for our constructed instance, implies a partition that violates the construction in \cite{feigeNorkin}.
\end{proof}

%\ufc{This should be discussed only in the results section, but not part of this proof.} In contrast, \cite{negativeMMSExample}
%show that in the setting without selling, with three agents and $8$ goods, a full MMS allocation always exists. Consequently, we have the following separation between allocation instances with and without sellable goods.

%\ufc{I would add the EFX for two agents already here.}

%\subsection{A \texorpdfstring{$\frac{11}{12}$}{11/12} MMS Gap for \texorpdfstring{$n\geq 3$}{3 agents}}
%\ufc{Need to say in the introduction that we use item and good interchangeably}

%All valuations in this section are assumed to be additive.

\subsection{A \texorpdfstring{$\frac{2}{3}$}{2/3}-MMS Algorithm for General \texorpdfstring{$n$}{n}}
\label{sec:general_n}

\nMMS*

For the general $n$ part of the theorem, we keep $\rho$ as a parameter (rather than fix it throughout to be $\frac{2}{3}$). We assume that $\frac{1}{2} < \rho < 1$.

We recall that the valuation function of agent $i$ is additive and denoted by $v_i$, the price of good $e$ is denoted by $p(e)$, and we use $\bar{v}_i$ to denote the additive valuation function satisfying $\bar{v}_i(e) = \max[v_i(e), p(e)]$ for every good $e \in \items$. 

As we show in Corollary~\ref{corr:value_readjust}, we may assume that each agent $i$ has an MMS partition in which each part has $\bar{v}_i$ value of exactly $MMS_i$. %For simplicity of the presentation, and without loss of generality, we may assume that the value of each part is exactly $MMS_i$.} \ygc{Maybe say a few words on why that's true? Adjutsing the valuations and so on.} \ufc{It is worth explaining this, but its better to explain it in a "preliminaries" section, as this is not an argument specific to this proof, but rather one that holds in general.} 
In this MMS partition, some goods are sold and the money received is split over more than one part. (A good that is sold but its money all belongs to the same part need not be sold at all, as we consider $\bar{v}_i$.) We refer to the set of goods that $i$ desires to sell (and split the money) as $\sellable_i$. 

%\subsection{Overview of proof}

Our allocation algorithm is inspired by a known algorithm \cite{procaccia2014fair,apxMMS2017} that was developed in the case in which goods are not sellable. However, to handle sellable goods, we introduce multiple new components to that algorithm.

Let us first provide a brief overview of the case without sellable goods. Fix a value of $\rho$ that is sufficiently small so that the algorithm does not get stuck, and conditioned on that, as high as possible.  A bundle $B$ is said to be {\em acceptable} to agent $i$ if $v_i(B) \ge \rho \cdot MMS_i$. 

An agent is said to be {\em active} if she did not yet receive any bundle. At a given step of the algorithm, $n'$ denotes the number of active agents. 

\begin{enumerate}
    \item Pick an arbitrary active agent $i$. Agent $i$ proposes a partition of the remaining goods into $n'$ bundles, each acceptable for $i$. We refer to such an $n'$-partition as acceptable.
    \item Employ a {\em matching procedure}. Hall's theorem for bipartite matching can be shown to imply that for some $1 \le k \le n'$, we can match $k$ active agents each with a bundle that is acceptable for her, with the additional  property that none of the matched bundles is acceptable to any of the remaining agents. Allocate the matched bundles to the respective active agents.
\end{enumerate}

A key aspect in the analysis is to show that for every agent, if $n - n'$ bundles not acceptable to her are removed, the remaining goods still have an $n'$-partition into acceptable bundles. In the setting without selling, it is known that for this property to hold, the best possible value of $\rho$ is $\frac{2}{3} + \Theta(\frac{1}{n})$.

Adapting the algorithm to the setting with sellable goods raises a major challenge. The fact that a bundle $B$ is not acceptable to agent $i$ no longer implies that by giving $B$ to another agent, the total value of the remaining goods decreases by at most $\rho \cdot MMS_i$. The problem is that $B$ might contain money received from selling a good $e$. If agent $i$ did not desire to sell $e$ (in particular, $e \not \in \sellable_i$), then in her original MMS partition agent $i$ got value of $v_i(e)$ from good $e$, but now that $e$ is sold, its value is $p(e)$. If $p(e) < v_i(e)$, the allocation of $B$ caused agent $i$ not only a loss of $\rho \cdot MMS_i$, but also an additional loss of $v_i(e) - p(e)$. Due to this additional loss, the analysis of the above algorithm breaks down.

We now explain some high level ideas that we use in order to cope with this additional loss. %\ufe{
To better appreciate our new contributions, it is convenient to note that the $\frac{2}{3}$ bound without sellable goods is a consequence of the following more general statement, that we shall refer to as the {\em approximate partition guarantee}. For every $\frac{1}{2} \le \rho \le \frac{2}{3}$, if $n - n'$ bundles each of value at most $2(1 - \rho) MMS_i$ (or more generally, of total value at most $2(1 - \rho)(n - n')MMS_i$) are removed, then the remaining goods have an $n'$-partition in which each part has value at least $\rho \cdot MMS_i$. 
%}

One principle that we use is to modify the algorithm in a way that allows us to upper bound this additional loss. This principle has three components:
\begin{itemize}
    \item The first component involves the structure of the $n'$-partition. Rather than allowing an agent $i$ to pick an arbitrary acceptable $n'$-partition, 
we require that for each part by itself, it suffices to add money from a single good that is sold in order to make the part acceptable (Lemma~\ref{lem:number_of_split_goods} shows that this property can be enforced). This limits the additional loss from each bundle to at most $v_i(e) - p(e)$ for some element $e$. 

\item The second component is to ensure that $v_i(e)$ is not too large. In the case with no sellable goods, we can always assume that $v_i(e) \le \rho \cdot MMS_i$, because we otherwise give $e$ to agent $i$, and the MMS of other agents does not decrease. With sellable goods, it is no longer true that by giving $e$ to $i$ the MMS of other agents does not decrease. For example, there may be an agent $i'$ who values all goods other than $e$ at $0$ and only derives value from selling $e$. Consequently, we employ a somewhat more complicated procedure so as to arrive to the situation in which $v_i(e) \le \rho \cdot MMS_i$ for every $i$ and $e$. 
%\ufe{
Let $i$ be the agent with smallest $MMS_i$ value for which there is a good $e$ satisfying $p(e) \ge \rho \cdot MMS_i$. If there are such $e$ and $i$, then sell $e$, and from the money received give $i$ only $\rho \cdot MMS_i$. (This is analogous to the moving-knife procedure 
used with divisible goods: giving an agent the smallest part of a good that is still acceptable to her.) The rest of the money that is received can be used by other agents. After all such cases are cleared, if there are such $e$ and $i$ so that $v_i(e) \ge \rho \cdot MMS_i$, let agent $i$ take good $e$ in full. 

\item The third component is to ensure that $p(e)$ is not too small. For this, we allow agent $i$ to sell a good $e$ only if $p(e) \geq (1 - \rho)MMS_i$. If $p(e) < (1 - \rho)MMS_i$, then good $e$ is required to be contained as a whole in one of the parts of the $n'$-partition (regardless of whether it was sold in $i$'s original MMS partition).
An aspect that helps absorb the loss incurred by failing to sell $e$ is that it suffices that each part of the $n'$-partition has value at least $\rho \cdot MMS_i$.
To ensure that $p(e) \ge (1 - \rho)MMS_j$ also for other active agents $j$, at each iteration we choose the partitioning agent $i$ to be the active agent with {\em highest} MMS value.
\end{itemize}

Using the above three components, we can upper bound the loss incurred to active agent $i$ due to the forced sale of goods when a bundle $B$ was given to agent $j$. The first component implies that there was only one sold good $e$. The second component implies that $v_i(e) \le \rho \cdot MMS_i$. The third component implies that $p(e) \ge (1-\rho) \cdot MMS_i$. Hence the loss is upper bounded by $v_i(e) - p(e) \le \rho \cdot MMS_i - (1 - \rho)MMS_i = (2\rho - 1)MMS_i$.
Another principle that we use in our analysis is the following.
We show that we can arrange the partition so that if a removed bundle $B$ causes a loss larger than $\rho \cdot MMS_i$, then $B$ has 
a particular simple structure. 
%holds for every  removed bundle $B$ that causes a loss larger than $\rho \cdot MMS_i$: 
Namely, in this case $B$ is composed of one good allocated in full and part of the money from selling an additional good. To see why this may be helpful, %\ufe{
observe that for $\rho = \frac{2}{3}$, the loss of value due to selling a single good  is at most $(\rho - (1-\rho))\cdot MMS_i = \frac{1}{3} \cdot MMS_i$, and so the combined loss incurred by allocating $B$ to a different agent is at most $MMS_i$, since in the matching provided by Hall's theorem, the active agents find $B$ unacceptable (i.e., its value is bounded by $\frac{2}{3} MMS_i$).
%} 
%\ygc{This sentence is a bit confusing}. 
Consider a related (but easier situation) in which $B$ is simply composed of two goods and has total value at most $MMS_i$. In this case, if both goods belong to the same bundle in the original MMS partition we can simply give up the whole bundle, and if they belong to different bundles, then the union of what remains from these bundles gives a new bundle of value at least $MMS_i$ (because valuations are additive). Hence in either case, the MMS value is not hurt by having some other agent receive $B$. In our setting, $B$ is composed of one good and some money rather than two goods, and the above local analysis (involving only two bundles) does not work. However, a more global type of analysis can be made to work.

To implement this structural property, we would like the $n'$-partition proposed by agent $i$ to be composed of only two types of parts. We refer to the first type as {\em pure} bundles. These bundles do not contain any money from goods that $i$ proposes to sell. For a pure bundle, if it is not acceptable to an agent $j$, then having it allocated to some other agent causes $j$ a loss of at most $\rho \cdot MMS_i$. We refer to the other type as {\em singleton} bundles. They contain a single good, and money from selling at most one other good, as explained in the second principle 
above. One thing that eases the analysis is allowing one special bundle we call the ``leftovers'' bundle to have no restriction on its type.  

% I revised the statement, and moved the example into a separate file, for future reference. 
We show that such partitions (that we refer to as {\em canonical}) do exist. To make use of canonical partitions, we need to ensure that in the matching procedure, we do not end with a partial matching in which some agent receives the leftovers bundle. We ensure this property by modifying the matching procedure. We drop from the bipartite graph $G$ both the leftovers bundle, and the agent $i$ who proposed the $n'$-partition.

We introduced above two main principles used in our proof. One is to make the additional loss due to sale of goods as small as feasibly possible (unfortunately, we cannot bring it down to~0). The other is that bundles that cause a loss larger than $\rho \cdot MMS_i$ must have a very simple structure. These principles provide a high level understanding of the proof, but in order to obtain a rigorous proof, there are many details to fill in. These are addressed in our full proof. For lack of space, it is presented in the appendix. See Appendix~\ref{sec:mms_23_full_proof}.

\subsection{A \texorpdfstring{$\frac{3}{4}$}{3/4}-MMS Algorithm for \texorpdfstring{$n=3$}{3} Agents}

We prove two lemmas that allow us to further optimize the $n=3$ case, and get a $\frac{3}{4}$ approximation. 

\begin{restatable}{lemma}{CanonicalThree}
\label{lem:canonical_three}
 Given an MMS partition for $n=3$, and where any good $g$ has $\bar{v}_i(g) < \frac{3}{4}\cdot MMS_i$, we can create a canonical form partition of acceptable bundles each of value at least $\frac{3}{4}\cdot MMS_i$. 
\end{restatable}

\begin{restatable}{lemma}{RemainderThreeFour}
\label{lem:remainder_three_four}
Suppose that $B$ is either a singleton bundle or a pure bundle, and $\bar{v}_i(B) \le \frac{3}{4}MMS_i$ holds for agent an $i$. Then if $B$ is allocated to some other agent, 
    agent $i$ can partition the remaining goods into two disjoint bundles $B'_1, B'_2$, each of value $\bar{v}_i(B'_j) \geq \frac{3}{4}MMS_i$.
\end{restatable}

%\ufc{discarded formulation: Given that a single singleton or pure bundle is allocated to an agent, 
 %   each agent $i$ of the agents not allocated a bundle can partition the remaining goods into two disjoint bundles $B'_1, B'_2$, each of value $\bar{v}_i(B'_j) \geq \frac{3}{4}MMS_i$.} 

The two lemmas can be applied through an algorithm that adapts the general $n$ algorithm (for $\frac{2}{3}$-MMS approximation). %\ufc{To shorten and clarify the section, do not repeat the description of the algorithm. Say that we basically run the same algorithm, but with different $\rho$. Then, state the two lemmas. Finally explain that it cannot happen that two agents receive a bundle in the first stage, and in the second stage, we use our previous result for two agents.} 
The general structure is as follows: (i) If there is a good that an agent values above $\frac{3}{4}$-MMS, allocate it (or part of its sale proceeds) through a moving-knife procedure. If indeed there is such a good, then we have a single agent that is allocated a bundle they accept, and two remaining agents that do not accept the bundle already allocated. Lemma~\ref{lem:remainder_three_four} then ensures that the remaining two agents value the remaining goods as at least $\frac{3}{4}$-MMS. 
(ii) If there is no such good, the highest MMS agent creates a canonical partition (made of pure or singleton bundles, and at most one leftovers bundle) with value of at least $\frac{3}{4}$-MMS for each bundle (as Lemma~\ref{lem:canonical_three} ensures can be done). (iii) {We create a bi-partite graph between the two non-partitioning agents and the two non-leftover bundles in the partition, where an edge between an agent and a bundle exists if the agent finds the bundle acceptable (i.e., values it at least as $\frac{3}{4}$-MMS). If there is a perfect matching, we assign according to the matching, and assign the leftover bundle to the partitioning agent (denote her agent $1$). If there is a bundle in the bi-partite graph that no agent wants, we assign it to the partitioning agent, and apply Lemma~\ref{lem:remainder_three_four}. The only remaining case is that one of the non-partitioning agents (denote her agent $2$) accepts both bundles, and the other (denote her agent $3$) finds both unacceptable. If any of the two bundles is pure, then agent $3$'s loss from the two bundles being allocated to agent $1$ and $2$ (which is possible, as they both find both bundles acceptable) is at most $2 \cdot \frac{3}{4}$-MMS by the bundles themselves (since they are unacceptable to agent $3$), and another at most $\frac{3}{4}$-MMS from $v_3(e) - p(e)$ due to the sale of at most a single good (as there is at most one non-pure bundle among the two). All in all, this implies that the leftover bundle is valued at least as $\frac{3}{4}$-MMS by agent $3$, and can be allocated to her. Then, the only remaining subcase is where both non-leftover bundles are singleton bundles. In this case, it is w.l.o.g. to assume that agent $1$ values them exactly at $\frac{3}{4}$-MMS (as any superfluous sale proceeds can be moved to the leftover bundles). Then, we can allocate one of the singleton bundles to agent $2$, and apply Lemma~\ref{lem:remainder_three_four} to agent $1$ and $3$.} 

In all events that we invoke Lemma~\ref{lem:remainder_three_four}, we thereafter use our $n=2$ algorithm (Theorem~\ref{thm:2EFX}) to output an allocation that gives each of these agents at least $\frac{3}{4}$-MMS.

%We run one iteration of the matching procedure, where each agent accepts a bundle if it values it at least $\frac{3}{4}$-MMS. Following the matching procedure, we
 %either arrive at a perfect matching, or a partial matching. In any case, the maximal matching is not empty, since the partitioning agent accepts all bundles. Moreover, since $n=3$, if it is a partial matching, {then the violating set (due to Hall's theorem) can not contain all bundles, as this implies there is an agent that finds all bundles unacceptable. Thus, the violating set matches exactly one bundle to one agent.} 
 %(iv) If at this point all agents are allocated a bundle, then we are done (all agents received a bundle they accept as having value at least $\frac{3}{4}\cdot MMS$). Otherwise, either through (i) or (iii), we have a single agent that is allocated a bundle they accept, and two remaining agents that do not accept the bundle already allocated. Lemma~\ref{lem:remainder_three_four} then ensures that the remaining two agents value the remaining goods as at least $\frac{3}{4}$-MMS. We then use 
 %We show that over the remaining goods, each agent $i$ of the two agents has an MMS for two players of $\frac{3}{4}MMS_i$ (i.e., each agent $i$ of the two agents can partition the remaining goods into two bundles, so that they each have value $\frac{3}{4}MMS_i$). We then use 
%em \ufc{what do the last 4 words mean?} \ygc{bundles? Remaining items?}

\section{Discussion}

In the appendices, we consider two natural extensions of our model.

\subsection{Equal Proceeds}

For our discussion so far, we have considered \textit{unrestricted} proceeds distribution, where the sale proceeds may be arbitrarily distributed among the agents. An interesting case is \textit{equal} proceeds distribution, where each agent must receive an equal monetary compensation: 

$$\forall i, P_i = \frac{1}{n}\sum_{g_j \in sellable} p(g_j).$$

%\ufc{Use $\sum_{g \in S} p(g)$}

%\ufc{Text below needs editing} 
Clearly, with this requirement, the MMS value of agents is weakly smaller than without this requirement. We show an example where only a $\frac{1}{n}$ approximation to this weaker MMS notion can be guaranteed. We also show an example where this weaker MMS notion guarantees only $\frac{1}{n}$ of the MMS notion with unrestricted proceeds. We complement the above negative results with an algorithm that guarantees $\frac{1}{n}$-TPS (as defined for the unrestricted proceeds) to each agent.  

\subsection{Chores with Outsourcing}

We consider allocation of chores (items of negative value that must be allocated). With chores, instead of the option of selling, agents have the option of outsourcing, where the agents collectively need to pay the market-price of outsourced chores (and then, split the payment in some way among the agents). Our results for MMS approximation for chores mirror in spirit our results for goods. Namely, we find that outsourcing adds difficulties to the design of allocation algorithms, and show approximation gaps that are wider than those known without outsourcing. We design allocation algorithms with constant approximation of the MMS, but not as good a constant as is known without outsourcing.

\begin{theorem}
\label{thm:chores}
    Every allocation instance with chores (and optional outsourcing) has a 2-MMS allocation. For every $\epsilon > 0$, there are allocation instances that do not have ${(\frac{19}{18} - \epsilon)}$-MMS allocations.
\end{theorem}

\subsection*{LLM Usage Disclosure}

When finalizing the manuscript, we were assisted by the AI feedback tool Refine.ink. We have also used the auto-formalizing system Aristotle \cite{aristotle} to proofread the paper and produce a Lean4 formalization of the manuscript, available at \href{https://github.com/yotam-gafni/FairAllocationWithOptionalSelling}{https://github.com/yotam-gafni/FairAllocationWithOptionalSelling}. 

\subsection*{Acknowledgments}

This research was supported in part by the Israel Science Foundation (grant No. 1122/22).

\bibliographystyle{splncs04}
\bibliography{ref.bib}

\appendix

%\ufc{Order the sections in the appendix in an order consistent with that of the main part. Move to after end{document} those parts of the appendix that are not needed for the submission. Maybe break the appendix into two main parts (each containing subsections). One with proofs for main results. The other with proofs for extensions. Then, we can check carefully the first part and less carefully the second part.} \ygc{I don't want to do the 'two sections' thing because it will mix REALLY different things in the same section. I kept to the "Extension:" convention for extensions...}

\section{Missing Proofs and Additional Results for Section~\ref{sec:two_agents}}
\label{app:missing_proofs_n=2}

We start by making the following useful distinction. 
Instead of talking about an allocation $A_i$ of kept goods and sale proceeds $P_i$, we convert them to: (i) An allocation $B_i$ of goods that belong in full to the agent, where the agent may decide whether to keep or sell them. For any such good $g$, the agent receives value $\bar{v}_i(g)$. (ii) Sale proceeds $P'_i$ out of goods that are \textit{sold} (a forced sale, without leaving a choice to the agents about it). For simplicity, we sometimes refer to the bundle of full items, together with sale proceeds from forced-sale goods, as a unified bundle. 

The following lemma shows that we can achieve the MMS value with an MMS output with at most $n-1$ \textit{sold} goods.

\begin{restatable}[]{lemma}{numSplit}
\label{lem:number_of_split_goods}
There always exists an MMS output with $n$ parts has at most $n-1$ \textit{forced-sale} goods.
\end{restatable}

\begin{proof}
    Let $\ell$ be the minimal amount of split goods $g_{i_1}, \ldots, g_{i_{\ell}}$ in an MMS output for the agent. Assume $\ell \geq n$.
It must be that 

\begin{equation}
\label{eq:split_goods_implication_general_n}
\max \{P^1, \ldots, P^n\} < \min_{1\leq j\leq \ell} \{p(g_{i_j})\},
\end{equation}
    or we could include the good that violates this condition in full in the respective allocation. 
    I.e., if w.l.o.g. we assume that good $i_1$ violates the condition w.r.t. $P^{1}$, we have $P_1 \geq p(g_{i_1}) = \sum_{i=1}^n P^1_i$, and so we can consider an MMS output where the same goods are sold, and has $\tilde{P}_1 = P_1, \ldots, \tilde{P}_n = P_n$, but has:

    $$\tilde{P}^1_1 = \sum_{i=1}^n P^1_i.$$

    This implies that each other sold good sale proceeds $i_2, \ldots, i_{\ell}$ has $\Delta_{i_j} = P_1^j - \tilde{P}_1^j \geq 0$ in sale proceeds to distribute among $P_2, \ldots, P_n$, but by sum conservation this can be done, resulting in an MMS output with the same value but with only $\ell - 1$ fractional good sale proceeds, contradicting $\ell$ minimality. 

\end{proof}

%\ufc{Do we really need this as a separate lemma/ If it is used only for the corollary, then incorporate the argument into the proof of the corollary, and change the corollary to a lemma}

It is sometimes convenient to consider that the MMS output has the same value for the agent across all bundles. We show this can be assumed. 

\begin{definition}
    We say an MMS output $B_1, \ldots, B_n$ is \textit{equi-valued} if $\bar{v}_i(B_j) = \bar{v}_i(B_{j'})$ for any $j, j'$. 
\end{definition}

%By the following lemma, we are justified in assuming it:

\begin{restatable}[Value Readjustment]{lemma}{ValueReadjust}
\label{corr:value_readjust}
For the purpose of guaranteeing MMS approximation, it is w.l.o.g. to assume every agent has an equi-valued MMS partition. %\ufc{ambiguous, and not true under the natural interpretation that all have the same MMS value}
\end{restatable}

\begin{proof}
    Consider agent $i$ and some MMS partition $B_1, \ldots, B_n$. Assume some bundle has $\bar{v}_i(B_j) > MMS_i$. This bundle cannot contain sale proceeds. If it does, we could redistribute them equally between $B_j$ and any bundle that has value of exactly $MMS_i$, and so increase the MMS value. We conclude that we can readjust $v_i$ (without changing $p$) so that all bundles have the same value in the MMS partition. Let the new subjective valuations be $\hat{v}_1, \ldots, \hat{v}_n$. It holds that $\forall 1\leq i\leq n, g_j\in \items, \hat{v}_i(g_j) \leq v_i(g_j),$ and $MMS(\items, v_i,p,n) = MMS(\items, \hat{v}_i,p,n)$. %\ufc{last notation undefined. use standard $MMS(\items, v, n)$ notation. Also, do not make several different claims inside one long mathematical expression} 
    If we can guarantee $\rho MMS(\items, \hat{v}_i,p,n)$ to every agent $i$, then we can also guarantee $\rho MMS(\items, v_i,p,n)$ to every agent $i$: Take the allocation that guarantees $\rho MMS(\items, \hat{v}_i,p,n)$. Then, it also guarantees $\rho MMS(\items, v_i,p,n)$ (as they are equal for each agent). Moreover, the value of the allocation is at least as high under $v_i$ than under $\hat{v}_i$. 
\end{proof}

It is technically convenient to attribute sale proceeds to goods. I.e., while Eq.~\ref{eq:sale_proceeds_feasibility} maintains that the sum of distributed sale proceeds does not exceed the sum of sold goods, we could introduce a more granular notation $P_i^j$, where $1\leq i\leq n$ denotes the agent, and $g_j$ with $1\leq j \leq m$ denotes the sold good. We require 
\begin{equation}  
\label{eq:granular_sale_proceeds_distributed}
\begin{split}
& \forall j \in \sellable,\; \; \; \;  \sum_{i=1}^n P_i^j \leq \xi(g_j) \cdot p(g_j), \\
& \forall i\in \agents, j \in \sellable, \; \; \; \; P_i^j \geq 0.
\end{split}
\end{equation}

Clearly, Eq.~\ref{eq:granular_sale_proceeds_distributed} implies Eq.~\ref{eq:sale_proceeds_feasibility} by summing over all goods $j$, but we can also always find a granular representation, e.g., by greedily attributing sale proceeds to the goods sold in ascending order of the agents.

\begin{lemma}
\label{lem:split_equal_mms}
With $n=2$ agents, if agent $i$ has a forced-sale good split between their two MMS output bundles, then both bundles have the same utility for the agent.  
\end{lemma}
\begin{proof}
Let the MMS partition be the unified bundles $C^1, C^2$, and assume towards contradiction $u_i(C^2) < u_i(C^1)$. Let 
$\delta = u_i(C^1) - u_i(C^2)$. 
There is a split good $g$ between the bundles, where $P^1 + P^2 = p(g)$, and $\min \{P^1, P^2\} > 0$, so that $P^1$ belongs to the unified bundle $C^1$, and $P^2$ to $C^2$. If instead we set the sale proceeds to be $\tilde{P}^1 = \max \{\frac{P^1}{2}, P^1 - \frac{\delta}{2} \}$ in $C^1$, and $p(g) - \tilde{P}^1$ in $C^2$, the MMS value would increase, in contradiction. 
\end{proof}

To highlight the contribution of our $n=2$ MMS existence result, we define a natural extension of the Cut \& Choose algorithm to our setting. We show an example where this extension of Cut \& Choose fails to produce an MMS allocation, regardless of the agent that cuts. 

\begin{definition}
    \textit{Cut \& Choose for $n=2$ agents.} Consider two agents $i$ and $j$. Let $C$ be an MMS partition of agent $i$. Allocate agent $j$ its preferred (unified) bundle among $C^1, C^2$, and allocate the other bundle to agent $i$. Then, we call this the Cut \& Choose allocation where agent $i$ is the ``cutting'' agent and agent $j$ is the ``choosing'' agent. 
\end{definition}

\begin{example}
\label{ex:bad_cut_choose}
    Consider $n=2$ agents and $m=5$ goods, with valuations and market price as given in Table~\ref{tab:cut_choose_valuations}. There is a left-right mirror symmetry between the agents, so we make the argument for the case Agent $1$ cuts and Agent $2$ chooses, but the other case is similar. The following is an MMS partition for agent $1$: The agent sells $g_5$, and allocates one bundle with $g_1, g_2$ and $\frac{3}{4}$ in sale proceeds, and another bundle with $g_3, g_4$ and $\frac{1}{4}$ in sale proceeds. When $x > \frac{3}{2}$, this MMS partition is unique. 
    In the perspective of Agent $2$, both bundles have a value of $\frac{x}{2} + \frac{5}{4}$. When $x \rightarrow \infty$, this yields a $\frac{1}{2}+\epsilon$ approximation to the MMS with arbitrarily small $\epsilon$. 
\end{example}

\begin{table}[ht]
    \centering
    \renewcommand{\arraystretch}{1.3} % Adds vertical padding for fractions
    \begin{tabular}{lcccccc}
        \toprule
        & $g_1$ & $g_2$ & $g_3$ & $g_4$ & $g_5$ & MMS \\
        \midrule
        Agent 1 & $x$ & $\frac{1}{2}$ & $\frac{x+1}{2}$ & $\frac{x+1}{2}$ & $0$ & $x + \frac{5}{4}$ \\
        Agent 2 & $0$ & $\frac{x+1}{2}$ & $\frac{x+1}{2}$ & $\frac{1}{2}$ & $x$ & $x + \frac{5}{4}$ \\
        \midrule
        Market price $p$ & $1$ & $0$ & $0$ & $0$ & $1$ & -- \\
        \bottomrule
    \end{tabular}
    \caption{Subjective valuations and market prices for Example~\ref{ex:bad_cut_choose}}
    \label{tab:cut_choose_valuations}
\end{table}

In fact, the $\frac{1}{2}$ approximation of our example is tight. The example shows a case where no matter which agent cuts, Cut \& Choose does not guarantee more than $\frac{1}{2}$ the MMS, whereas the following lemma holds in any case and regardless of the choice of the cutting agent. 
%As a warm-up, \ufc{Is this warm up used later? If not, then a different argument should be given as to why we present this, or remove it to the appendix} we show that Cut \& Choose does offer a non-trivial MMS approximation, regardless of the choice of cutter and chooser. \ufc{If we keep the next lemma, we should show that it is tight.}

\begin{lemma}
    With $n=2$, Cut \& Choose gives a $\frac{1}{2}$-MMS approximation. 
\end{lemma}
\begin{proof}
    W.l.o.g. assume agent $1$ is the cutter. If their MMS partition does not force the sale of any good, then the chooser gets at least their proportional value (and so, also their MMS) by choosing their preferred bundle, and agent $1$ gets their MMS value from the remaining bundle. 

    If agent $1$'s MMS partition does force the sale of a good, by Lemma~\ref{lem:number_of_split_goods} we can assume it is a single good $g$. If agent $2$ also sells $g$ in their MMS partition, then the same argument as before holds. If they do not, then $g$ is contained in full in one of the MMS bundles of agent $2$, say $C^1$. We have 

    $$\bar{v}_2(\items \setminus \{g\}) \geq \bar{v}_2(C^2) \geq MMS_2. $$

    By choosing their preferred bundle among the partition by agent $1$, agent $2$ gets at least $\frac{1}{2}\bar{v}_2(\items \setminus \{g\})$, since all goods but $g$ are allocated in full and agent $2$ has the discretion of whether to sell them or keep them. Agent $1$ gets the remaining bundle and their MMS value. 
\end{proof}

\section{Full Proof of \texorpdfstring{$\frac{2}{3}$}{2/3}-MMS Approximation}
\label{sec:mms_23_full_proof}

We now present the full proof of the general $n$ part of Theorem~\ref{thm:nMMS}.
It uses notation (such as $S_i$ and $\bar{v}_i$) and terminology (such as acceptable bundles) that is introduced in the main text overview.

\begin{proof}
In the course of the allocation algorithm, we shall allocate bundles to some agents. We trace the total value of goods (and money received from sold goods) that remains for those agents that are still active (did not yet receive a bundle).  For each allocated bundle $B$, we shall make sure that the loss for any active agent $i$ is of one of four types. %\ygc{Should clarify that while the loss types has some connection with the canonical bundle types, it's not the same thing, and also it's not important that the loss types cover all types of possible bundles and so on, it only serves as kind of a useful macro during the proof.} \ufc{The loss types to cover all possible bundles that might appear in the allocation algorithm. But it would be good to explain somewhere, perhaps in the overview, that there is some sort of amortization going on.}

%\ygc{Also, it's sort of an amortized loss: We associate the loss of the sale of a certain good with the bundle $B$, even though it's not really part of the value of $B$. So, for example, with $\ell_4$, it's not really only that $\bar{v}_i(B) \leq MMS_i$ (it is clearly less than $\rho MMS_i$, because it was not acceptable to the agent), but that even in the amortized sense it is less than $MMS_i$. The amortization takes hold because if we look at the entire loss for the last iteration of the process, and associate parts of it with bundles in this way, then summing it all up we get the entire loss.}

\begin{definition}[Loss Types]
\label{def:loss_types}
\begin{itemize}
    \item {\em Loss of type $\ell_1$.} The bundle $B$ contains only one good $e$ and no money, and $e \not\in \sellable_i$.

    \item {\em Loss of type $\ell_2$.} $\bar{v}_i(B) \le \rho \cdot MMS_i$, and $B$ does not contain money received from selling any good $e$  with ${v}_i(e) > p(e)$, unless $e \in \sellable_i$ (in which case $i$ also intended to sell $e$)

    \item {\em Loss of type $\ell_3$.} The value of $B$ is derived from two goods, $e$ and $f$, each with $\bar{v}_i$ value at most $\rho \cdot MMS_i$. Both $e$ and $f$ are not goods in $\sellable_i$ (i.e., agent $i$ would not have chosen to sell them), but the allocation algorithm did sell $f$, and so %\ygc{I dont think the reader is ready at this point to understand what it means that $f$ was sold anyway, maybe explain this point a bit.} \ufc{propose to change to: but the allocation algorithm did sell $f$.} 
    $B$ contains $e$ and part of the money received from selling $f$. %\ufc{Check if $f$ should be $e$. Is the moreover something obvious, or is it something that we will need to prove?} Moreover, 
    The combined %amortized 
    loss ($B$ being allocated to another agent, together with forcing the selling of $f$) is at most $\bar{v}_i(B) + v_i(f) - p(f)
    \le \rho \cdot MMS_i + v_i(f)  - (1-\rho)MMS_i = v_i(f) + (2\rho - 1)MMS_i$.
    %Moreover, $\bar{v}_i(B) \le (3\rho - 1)MMS_i$. 
    (For $\rho = \frac{2}{3}$, this implies that the amortized loss is at most $v_i(f) + (2\rho - 1)MMS_i \leq (\frac{2}{3} + \frac{1}{3})MMS_i = MMS_i$.)
    
    \item {\em Loss of type $\ell_4$.} The value of $B$ is derived from two goods, $e$ and $f$, each with $\bar{v}_i$ value at most $\rho \cdot MMS_i$. We have that $e \in \sellable_i$, $f \not\in \sellable_i$, but $f$ was sold and $e$ was not. $B$ contains $e$ and part of the money received from selling $f$. The %total 
  combined loss (losing $B$ in combination of selling $f$ at a loss) is at most $v_i(f) + (2\rho - 1)MMS_i$.  %Moreover, $\bar{v}_i(B) \le \bar{v}_i(f) + (2\rho - 1)MMS_i$. 
    %(For $\rho = \frac{2}{3}$, this implies that $\bar{v}_i(B) \le MMS_i$.)

\end{itemize}
\end{definition}

The allocation algorithm starts with a preliminary phase whose goal in to %achieve the first component of the first principle introduced in the overview. That is, we wish to 
make sure that $\bar{v}_i(e) \le \rho \cdot MMS_i$ for every good $e$ and agent $i$.

The preliminary phase has two stages.

{\bf Stage 1.}

As long as there is an active agent $i$ and good $e$ with $p(e) \ge \rho \cdot MMS_i$, do the following. Let $j$ denote the active agent with lowest MMS value, and observe that necessarily, there is such a good $e$ with $p(e) \ge \rho \cdot MMS_j$.  Sell $e$, pay $\rho \cdot MMS_j$ to $j$, and remove $j$. 

%\ygc{Maybe give the intuition that this is like a moving knife procedure, just because $p$ is global we know that $j$ with the lowest MMS is the one to stop it}

We track the loss for the remaining active agents. For any remaining active agent $i$ for which $e \not\in \sellable_i$ the loss is of type $\ell_1$, whereas for any remaining active agent $i$ for which $e \in \sellable_i$ the loss is of type $\ell_2$. 

When Stage~1 is done, we move to Stage~2.

{\bf Stage 2.}

If there is any agent $i$ and good $e$ with $v_i(e) \ge \rho \cdot MMS_i$, give $e$ to $i$, and remove $i$. 

Again, for any remaining active agent $j$ for which $e \not\in \sellable_j$ the loss is of type $\ell_1$, whereas for any remaining active agent $j$ for which $e \in \sellable_j$ the loss is of type $\ell_2$ (because we have already exhausted Stage~1, and so $p(e) \leq \rho \cdot MMS_j$). 

The outcome of the preliminary phase is summarized in the following proposition.

\begin{proposition}
    After the preliminary phase,  $\bar{v}_i(e) < \rho \cdot MMS_i$ for every active agent $i$ and remaining good $e$.
\end{proposition}

We now reach the main part of our algorithm. At this point, there are $n'$ active agents, some goods may still remain, and some money $P'$ may remain (from goods sold in previous steps). 
%We assume without loss of generality that $P'$ is smaller than the price of the least expensive item $e$ that was previously sold. Otherwise, we can buy $e$ back and add $e$ to the pool of remaining items (reducing $P'$ to $P' - p(e)$), without affecting $\hat{v_i}$ of any of the active agents. 

Consider the active agent $i$ of highest MMS value. 
Find a partition (into $n'$ parts) with the following properties. We use the notation $\sellable_i$ for the set of goods that $i$ intends to sell, so that the money received from selling any single good is split over more than a single part. Let $P = \sum_{e \in \sellable_i} p(e)$. Note that $P'$ is not considered part of $P$. %\yge{In our partition, we do not force the sale of the goods in $\sellable_i$ outright, but rather have a specialized process for their sale, which is utilized in Proposition~\ref{pro:fewSingletons}.}
%\ufc{Not sure that we want this last sentence.}

%We use $\sellable_i$ in a \textit{lazy} fashion in our partition, i.e., we do not force its sale outright, but later, when a selection is made, sell the necessary goods to make sure the selected bundles are acceptable. 

\begin{enumerate}

%\item 
%\ygc{Worth emphasizing that this is a tentative/lazy suggestion and that only a part of it, if any, will actualize. Also, $S$ is a preliminary to the partition into $n'$ parts, it's not one of the parts.}

\item $p(e) \ge (1 - \rho)MMS_i$ for every $e \in \sellable_i$. 
    
\item Each part is of one of three types:

\begin{enumerate}
    \item Pure parts: the $\bar{v}_i$ value of a pure part is at least $\rho \cdot MMS_i$, and it contains no money from $P$. (It may contain money from $P'$.)

    \item Singleton parts: a singleton part contains a single good and some money from $P$, and its $\bar{v}_i$ value is at least $\rho \cdot MMS_i$. Importantly, for each singleton part, there is an associated good in $S_i$ such that all the money in the part comes from the sale of the associated good. The same good in $S_i$ can be associated with several singleton parts.  

    %\ygc{I think it doesn't matter. Because in the perspective of the later agent (that experiences this bundle as a loss), it can fall within $\ell_2, \ell_3$, or $\ell_4$ loss types. The bundle was unacceptable to the agent, so the condition for the bundle holds. Because of the preprocessing steps, the condition on $e$ and $f$ (the kept and sold items) also hold. And the goods sold have $p(f)\geq (1-\rho)MMS_i$ (we keep this property). So overall for the loss event, no issue should arise. BUT EVEN if there was an issue, during the process of Proposition 6, in the part where we have a virtual good and all the remaining proceeds do not exceed MMS, where we now allocate all the proceeds to the bundle, we could instead only allocate up to $\rho MMS$, and leave the remaining proceeds for the next agent (put them aside), instead.}

    \item Special part: a part that is neither pure nor singleton. Importantly, there is at most one part that is special in the partition.
    
\end{enumerate}

%\ufc{remove: 
\item For every good $e_1$ in a singleton part and for every good $e_2 \in \sellable_i$, we have that $\bar{v}_i(e_1) + p(e_2) > \rho \cdot MMS_i$.
%}

%\item $P$ suffices to make all singleton parts and the special part acceptable under $\bar{v}_i$. %\yotam{Does that mean we don't actually need $P'$? }
%\ufc{This was needed in the old version, but not now.} 

\end{enumerate}

We refer to partitions as above as {\em canonical}. We remark that we do not require canonical partitions to make use of all goods (it is fine if some goods are neither in any of the parts nor in $\sellable_i$).

Suppose that a canonical partition is found.
Consider the following $n' - 1$ by $n' - 1$ bipartite graph $G$. On one side, we have all active agents except for $i$. On the other side, we have all parts of the canonical partition (where each part includes the goods of that part and the money in that part, both that received from $P'$ and that received from $P$), except for the special part. (For uniformity of the presentation, if there is no special part, then remove some other arbitrary part.) We place an edge between an agent $j$ and a part $k$ if part $k$ is acceptable under $\bar{v}_j$. 

Given $G$, we proceed with a {\em matching procedure}.

\begin{enumerate}
    \item If $G$ has a perfect matching, we allocate parts according to the matching, and give the special part to agent $i$. All agents receive acceptable bundles and we are done with the allocation.
    \item If $G$ has no edge, we allocate one of the non-special parts to agent $i$ and remove agent $i$. If this non-special part was a singleton, then we sell a single good from $\sellable_i$ in order to produce the payment associated with the part, but do not sell the remaining goods from $\sellable_i$.
    \item If none of the above two cases hold, then by Hall's theorem, there is a subset of $t < n' - 1$ parts that can be matched to $t$ active agents, and except for agent $i$, no unmatched active agent finds any of the matched parts acceptable. Allocate the matched parts to the respective agents and remove these agents.

\begin{itemize}
    \item If the subgraph of $G$  that remains still contain an edge, then repeat step~3. After finitely-many repetitions the subgraph necessarily becomes empty. At that point we are at step~2 above, and then agent $i$ gets matched to one of the non-special parts.
\end{itemize}
    
    When step~3 ends, the payments associated with the matched singleton parts are generated by selling the fewest possible goods from $\sellable_i$. As the partition is canonical, the number of goods sold does not exceed the number of singleton parts that are allocated.
\end{enumerate}

At the end of the above matching procedure, at least agent $i$ (and perhaps also other agents) receives an acceptable bundle. Hence the number of active agents strictly decreases, meaning that the allocation algorithm eventually ends (of course, we still need to show that we can find canonical partitions at each iteration).

\begin{algorithm}[H]
    \SetAlgoLined
\DontPrintSemicolon
\KwIn{Valuations $v_1, \ldots, v_n$ over goods $\items$, market-price $p$, number of agents $n$}
\KwOut{An allocation $B$ with at least $\frac{2}{3}\cdot MMS$ value for each agent}

$ACTIVE = [n], AVAIL = \items, P' = 0$

\While{$\exists i\in ACTIVE, g\in AVAIL, p(g) \geq \rho MMS_i$} {
    Let $i$ be the lowest MMS agent in ACTIVE. 
    
    Sell $g$, allocate $\rho MMS_i$ in sale proceeds to agent $i$, add any remaining sale proceeds to $P'$. Remove $g$ from AVAIL, and $i$ from ACTIVE.  
}
\While{$\exists i\in ACTIVE, g\in AVAIL, v_i(g) \geq \rho MMS_i$} {
    Give $g$ to agent $i$. Remove $g$ from AVAIL, and $i$ from ACTIVE.  
}

%Let $n' = |ACTIVE|$. 

\While{$|ACTIVE| > 0$} {
    Let $i$ be the highest MMS agent in ACTIVE. 

    Let $B_1, \ldots, B_{|ACTIVE|}$ be a canonical partition of AVAIL, $P'$, where $B_1$ is the leftovers bundle. 

    Let $G$ be the bipartite graph between $B_2, \ldots, B_{|ACTIVE|}$ and the agents in $ACTIVE \setminus \{i\}$, with an edge whenever an agent finds a bundle acceptable. 

    %\While{There are vertices and edges in $G$} {

    Find a perfect matching, and if not, {a maximal matching from a violating set due to Hall's theorem}. Denote the matching $\mu$. Allocate according to the matching. Remove matched bundles from AVAIL and matched agents from ACTIVE, sell goods not contained in full in a bundle, and adjust $P'$ accordingly. 

    %Update $G$ to be the subgraph of $G$ excluding the agents and bundles in $\mu$. 
%}
{\If{a perfect matching was found} {
    Allocate $B_1$ to agent $i$, remove agent $i$ from ACTIVE. 
}
\Else{
    Allocate a bundle $B_j$ from the violating set that was not matched to an agent, to agent $i$. Remove $B_j$ from AVAIL and agent $i$ from ACTIVE. 
}
}

%\ygc{I made some edits to make the writing clearer}

}

 \caption{APX-MMS-2/3}
 \label{alg:APX-MMS-23}
\end{algorithm}

We now analyze the loss of an agent $k$ that is still active (was not matched). For every pure part that was allocated, the loss is of type $\ell_2$ (as the part was not acceptable under $\bar{v}_k$).
For the singleton parts, there are several cases to consider. Let $t$ denote the number of singleton parts (numbered from~1 to $t$) that were matched, and let $e_1, \ldots, e_t$ denote the $t$ goods contained in the $t$ singleton parts (the goods that are not sold). Let $t'$ denote the number of goods from $\sellable_k$ that were sold in order to generate the payments in the singleton parts, and denote these sold goods by $f_1, \ldots, f_{t'}$. Recall that necessarily, $t' \le t$. %\ygc{Is this rearrangements of the items actually sold really used in the proof? I remember we discussed it as maybe needed for the $\frac{3}{4}$ case but it doesn't seem necessary for the proof here and I didn't see it used.}

{We iterate over the singleton bundles to map them into loss events. Each singleton bundle is comprised of a pair $e_j, f_j$, where $e_j$ is kept, and some sale proceeds from $f_j$ are used. If $f_j \in \sellable_k$, we consider it a loss of type $\ell_2$. This is because the payment taken from $p(f_j)$ so as to make the singleton part acceptable is such that $\bar{v}_k(e_j) + p_j \le \rho \cdot MMS_k$. 

If $f_j \not \in \sellable_k$, we consider whether this is the first time that $f_j$ has appeared during our iteration over the singleton bundles. If it is, then note that 
$p(f_j) \ge (1 - \rho)MMS_i \ge (1 - \rho)MMS_k$ (here we use the fact that $i$, the agent who gets to propose the partition, has the highest MMS value). Hence the fact that $f_j$ was sold caused agent $k$ an amortized loss of at most $v_k(f_j) - (1 - \rho)MMS_k$. The total loss includes also the value of the singleton part $B_j$ containing $e_j$ that was allocated, which is at most  $\rho MMS_k$ (because $B_j$ was not acceptable for agent $k$). Hence the total loss is at most $v_k(f_j) + (2\rho - 1)MMS_k$.  If $e_j \not\in \sellable_k$, then we have a loss of type $\ell_3$. If $e_j \in \sellable_k$, then we have a loss of type $\ell_4$.

If it is not the first time that $f_j$ appears in a singleton bundle during our iteration, we consider it a type $\ell_2$ loss. This is since agent $k$ does not find the bundle acceptable, $e_j$ is not sold, and the selling loss from $f_j$ was already accounted for in the first singleton bundle that included sale proceeds from $f_j$. }

So far, we established the following proposition.

\begin{proposition}
    For every active agent $i$ that remains, each agent that was removed received a bundle that caused agent $i$ to suffer a loss of one of the four types $\ell_1$, $\ell_2$, $\ell_3$ or $\ell_4$.
\end{proposition}

%\ygc{Recall that at each iteration of the process, at least one agent becomes inactive after getting an acceptable bundle. Thus, as long as the process can continue, eventually all agents get an acceptable bundle. YOTAM: I initially didn't understand why there is no part of the proof that shows it's ok for the agent to suffer this loss of bundles of the above types, but then I realized that it ends up ok because the process can continue on and on and eventually the agent gets an acceptable bundle.}
To show that the allocation algorithm terminates with all agents receiving acceptable bundles, we need to set the value of $\rho$ to be such that a canonical partition exists in every step of the main part of the algorithm.

\begin{lemma}
\label{lem:canonical}
    For $\rho = \frac{2}{3}$, a canonical partition exists in every step of the main part of the algorithm.
\end{lemma}

\begin{proof}
Consider active agent $i$ when $n'$ agents remain. 
%We may assume without loss of generality that initially, there was no item $e \in \items$ with $v_i(e) \ge \rho \cdot MMS_i$ that $i$ did not intend to sell in her original MMS partition. This is because if there was such an item $e$, then  it would be allocated in stage~1 and~2, not to $i$ (because $i$ is still active), but to some other agent $j$. But then removing $j$ and $e$, the MMS of $i$ does not decrease.  

Let $B_1, \ldots, B_n$ denote the original $MMS_i$ partition of $\items$, and recall that $\sellable_i$ is the set of goods that $i$ intends to sell. With the goods and money that remains, we need to produce a $\rho$-$MMS_i$ partition $B'_1, \ldots, B'_{n'}$ that is canonical. 

Let $n_1, n_2, n_3, n_4$ (with $n_1 + n_2 + n_3 + n_4 = n - n'$) be the number of events of losses of types $\ell_1, \ell_2, \ell_3, \ell_4$, respectively. 

In losses of type $\ell_1$, a single good $e \not\in \sellable_i$ is sold. Remove the (at most) $n_1$ bundles that contain these goods. We remain with an $MMS_i$ partition into at least $n - n_1$ parts, and there are no events of $\ell_1$ losses.

In a loss event $j$ of type $\ell_3$, two goods $e_j$ and $f_j$ are lost, with $e_j, f_j \not \in \sellable_i$. In addition some money $r_j$ is received, so that the total net loss is at most $(3\rho - 1)MMS_i$. For our choice of $\rho = \frac{2}{3}$, this loss is at most $MMS_i$. If $e_j$ and $f_j$ belong to the same bundle in $B_1, \ldots, B_n$, simply remove that bundle. If they belong to different bundles, then after removing $e_j$ and $f_j$, merge what remains of these two bundles, and add to the newly created bundle the payment of $r_j$. %\yge{We lose two bundles, but gain the merged bundle of value at least $2 \frac{1 + \rho}{2} MMS_i - (3\rho - 1)MMS_i = (2 - 2\rho)MMS_i$, which equals $\rho MMS_i$ for $\rho = \frac{2}{3}$. } 
Handling all loss events of type $\ell_3$, we remain with an $MMS_i$ partition into at least $n - n_1 - n_3$ parts, and there are no events of $\ell_1$ or $\ell_3$ losses. 

As we may discard parts, we assume that exactly $n - n_1 - n_3 = n' + n_2 + n_4$ parts remain (call them $B_1, \ldots B_{n' + n_2 + n_4}$), and we still need to account for $n_2$ loss events of type $\ell_2$, and $n_4$ loss events of type $\ell_4$.

%Now we consider loss events of type $\ell_4$. With $\rho = \frac{2}{3}$, they correspond to losing a single item of value at most $\frac{2}{3} \cdot MMS_i$ that $i$ did not intend to sell (but was sold anyway), and additional monetary loss of at most $\frac{1}{3} \cdot MMS_k$. 

Let $L$ (for {\em large}) denote that set of all goods of value at least $(1 - \rho)MMS_i$ that are not in $\sellable_i$. Of them, let $L'$ denote those that still remain available (neither allocated nor sold) when $n'$ agents remain. At this stage of our analysis, they belong to the $n' + n_2 + n_4$ $MMS_i$-bundles that we have. We wish now to redistribute $L'$ only among $n'$ bundles, $B'_1, \ldots, B'_{n'}$, those that will form our canonical $\rho$-$MMS_i$ partition into $n'$ parts. (We are allowed to also discard some goods from $L'$, as canonical partitions need not contain all goods.)

Consider first the case that $|L'| \le n' + 1$. Observe that for $\rho = \frac{2}{3}$, no loss event causes a loss larger than $MMS_i$ (type $\ell_2$ loss is at most $\rho \cdot MMS_i$ and type $\ell_4$ loss is at most $(3\rho - 1)MMS_i$). Hence the total value that we have at our disposal is at least $n' \cdot MMS_i$. 

\begin{proposition}
\label{pro:fewSingletons}
    If $|L'| \le n' + 1$, $\rho = \frac{2}{3}$, and the total value that we have at our disposal is at least $n' \cdot MMS_i$, then a canonical $\rho$-$MMS_i$ partition into $n'$ parts exists.
\end{proposition}

\begin{proof}
The total value that we have at our disposal is composed of four disjoint components:

\begin{enumerate}
    \item Goods in $L'$ ({\em large} goods).
    \item Money $P'$ leftover from previously sold goods. 
    %Recall that we may assume that $P' \le \rho MMS_i$.
    \item Goods in $C'$ (here $C$ stands for {\em cheap}) that were not sold until now, and are neither in $\sellable_i$ nor in $L'$. This last condition implies that $\bar{v}_i(e) \le (1 - \rho)MMS_i$ for every $e \in C'$.
    \item Goods in $\sellable'_i$ that were not sold until now and are in $\sellable_i$. We move some of the goods of $\sellable'_i$ into some of the other components (in particular, $i$ gives up the intention to sell them), as follows.
    \begin{enumerate}
        \item Those goods $e\in \sellable'_i$ with $\bar{v}_i(e) \le (1 - \rho)MMS_i$ are moved into $C'$ (in terms of their value they fit in $C'$).
        \item If $|L'| < n'$, then $n' - |L|'$ goods $e\in \sellable'_i$ with $\bar{v}_i(e) > (1 - \rho)MMS_i$ (or fewer, if the number of available such goods is smaller)  are moved into $L'$ (in terms of their value they fit in $L'$).
    \end{enumerate}
We reinterpret the notation $\sellable'_i$ to include only those goods not moved out of $\sellable'_i$. Observe that now every good  $e \in \sellable'_i$ satisfies $\bar{v}_i(e) > (1 - \rho)MMS_i$, and moreover, if $\sellable'_i$ is non-empty then $|L'| \ge n'$.
\end{enumerate}

If $|L'| \le n'$, then put each good of $L'$ in a different part of the canonical partition. If $|L'| = n' + 1$, then put two goods of $L'$ in the same bundle $B'$, which can serve as one pure bundle of the canonical partition. In this latter case, we need to create only $n'-1$ additional parts. The total value that remains at our disposal for this is at least $n' \cdot MMS_i - \bar{v}_i(B') \ge (n' - 2\rho)MMS_i$. To unify the cases of $|L'| \le n'$ and $|L'| = n' + 1$, we rename $n' - 1$ as $n'$ in the latter case, and so in either case we are in the situation that every good of $L'$ is in a different part, and the total value at our disposal is at least $(n' + 1 - 2\rho)MMS_i$. %(as $\rho \ge \frac{1}{2}$).

For money from $P'$, do a {\em bag filling procedure}, pouring money into the same part until the value of the part reaches $\rho \cdot MMS_i$, and then moving to the next part. Each part that is made acceptable in this manner is pure and has value at most $\rho \cdot MMS_i$.

For goods in $C'$,  continue the bag filling procedure from the part in which it previously ended, putting them one by one into the same part until the value of the part reaches $\rho \cdot MMS_i$, and then moving to the next part. Each part that is made acceptable in this manner is pure and has value at most $MMS_i$.

%For items $e\in S'_i$ with $\bar{v}_i(e) \le (1 - \rho)MMS_i$, move them into $C'$, and continue the bag filling procedure from the part in which it previously ended. Hence, we may assume that $\bar{v}_i(e) > (1 - \rho)MMS_i$ for every $e \in S'$.

At this stage we already have some parts that are acceptable and pure, some parts that contain only a single good from $L'$,  and (at most) one part that is {\em special} (a part to which the bag filling procedure added goods but it has not become acceptable yet). %\ygc{The following is true, but unnecessary to state at this stage, because it doesn't change what we do next, and is repeated in more general form (and clearer) later: REMOVE:
Observe that no part can be empty. Having a part with no good from $L'$ implies that $\sellable'_i$ is empty, but then the total combined value of all goods is at most $n'-1$, whereas the total value at our disposal was higher, at least $n' + 1 - 2\rho$.

%\ufc{Yotam removed}
%As long as $S'_i$ is non-empty, we can transform empty parts into parts that contain a single item from $L'$, by moving an item from $S'_i$ (note that it has $\bar{v}_i$ value at least $(1 - \rho)MMS_i$) into $L'$ \yotam{I didn't understand this. Why can we do it, and why do we need to do it?}, and then into the empty part. It cannot be that we exhaust $S'_i$ while there still is an empty part, as this will contradict the fact that the total value that we have at our disposal is strictly larger than $(n' - 1) MMS_i$ \yotam{This expression seems off, previously we said $(n'+1-2\rho)MMS_i$, and also maybe we should redo the accounting given some parameter of how many bundles were already filled? }. Hence we may assume that there are no empty parts.
%\ygc{I removed a paragraph here, and added some words at the end, I think this and the last paragraph were overlapping, and this paragraph was not so clear.} \ufc{Okay. I augmented the explanation of $S'_i$ above so that the paragraph that you removed is no longer needed.}

{We now continue with the bundles that were not involved in the bag-filling, and are not acceptable in current form (where they contain a single item from $L'$). Let us denote the set of all such bundles $J$ and order them arbitrarily from $j = 1$ to $j = |J|$. 
We iterate over the remaining goods in $\sellable'_i$ in an arbitrary order. For every such good $e$, we add it to bundle $j$ in full if the current value of bundle $j$ together with $p(e)$ does not exceed $MMS_i$. If we indeed add $e$, then this creates an acceptable bundle (its value is at least $2(1 - \rho)MMS_i \ge \rho \cdot MMS_i$, where the inequality holds for $\rho = \frac{2}{3}$). This bundle is pure and can be incorporated in the canonical partition. Otherwise, we sell $e$, and add exactly the amount of sale proceeds so that the value of bundle $j$ is $\rho \cdot MMS_i$. We know that the remaining sale proceeds from $e$ are of value at least $(1-\rho)MMS_i$, or it could have been fully added to bundle $j$. We increase $j$ by $1$ and move to the next bundle, and we know that it can be made acceptable using the remaining proceeds from $e$, since the value of the good from $L'$ is at least $(1-\rho)MMS_i$, the remaining proceeds are at least $(1-\rho)MMS_i$, and together this gives $2(1-\rho)MMS_i = \rho MMS_i$ with $\rho = \frac{2}{3}$. We then repeat the logic we applied to bundle $j$ repeatedly (either add the remaining sale proceeds in full if it does not exceed $MMS_i$, or exactly the amount of sale proceeds to reach $\rho MMS_i$ otherwise), and we are guaranteed that all bundles processed in this way will be acceptable and not exceed $MMS_i$, until all sale proceeds from $e$ are used. We then move on to the next good in $\sellable'_i$, and proceed in the same fashion. If all bundles in $J$ are processed, then we add any remaining sale proceeds and goods from $\sellable'_i$ to the leftover bundle we marked during the bag-filling process.}

%If there is a single good from $\sellable'_i$ that can be added to a part that contains only a single good from $L'$ such that the total value of the part does not exceed $MMS_i$, do so. This creates an acceptable bundle (its value is at least $2(1 - \rho)MMS_i \ge \rho \cdot MMS_i$, where the inequality holds for $\rho = \frac{2}{3}$). This bundle is pure and can be incorporated in the canonical partition. \yge{If a good from $\sellable'_i$ can not be added in this fashion (but there is still some $e \in \sellable'_i$), then add }

%\ygc{It seems like a case is missing, which I think is you put an item from $L'$ as a full item and add sale proceeds from the $S'_i$ item? } \ufc{There is no such case, because we sell in a lazy fashion. I edited the next paragraph to reflect this.}

{We claim that when the process ends, we have a canonical partition. This is because we have three type of bundles: pure bundles (which can be part of the canonical partition), one leftovers bundle (this is allowed in a canonical partition), and bundles that contain only a single full good (this good is from $L'$) and sale proceeds from a single good. %For the single-good bundles, any single additional good from $\sellable'_i$ suffices in order to make them acceptable.
Importantly, the process where we add proceeds from $\sellable'_i$ to bundles in $J$
%Moreover, all of $\sellable'_i$ 
suffices in order to make all such bundles acceptable, and the leftover bundle as well. That is because we maintain that no bundle has value above $MMS_i$. Since the total value at our disposal is at least $(n' + 1 - 2\rho) MMS_i = (n' - 1 + \rho) MMS_i$ (the last inequality holds for $\rho = \frac{2}{3}$). I.e., we have enough value to complete the bundle processing creating the $n'-1$ bundles other than the leftover bundle using at most $(n'-1)MMS_i$ of value, and with at least $\rho \cdot MMS_i$ value left for the leftover bundle, making it acceptable. }

%when we get to form the last bundle after losing at most $(n'-1)MMS_i$ creating the first $n'-1$ bundles, we have a value of at least $\rho MMS_i$ left, and so also the last bundle is acceptable. Hence all the single-good bundles can be made into singleton bundles, completing the canonical partition.

%\ygc{I think one way to make this more understandable is as follows: We do the bag-filling with goods from $C'$. We stop at some point when the bag-filling doesn't fill an acceptable bundle. We declare this the leftovers bundle. Then, for all the bundles after it, we make it singleton using $S'_i$. If we allocate to a bundle with a good $e\in S'_i$, and the remaining sale proceeds from that good are less than $(1-\rho)MMS_i$ when we get to $\rho MMS_i$, then we keep them in this bundle (and it still has value at most $MMS_i$). Otherwise, we have sale proceeds of value at least $(1-\rho)MMS_i$ to use in the next bundle, that has a good from $L'$, so we can make it acceptable, and so on. The point is, we never need to use sale proceeds from two goods in $S'_i$ to complete a bundle with a good from $L'$. So they are all valid singleton bundles.}
\end{proof}

Given Proposition~\ref{pro:fewSingletons}, it remains to handle the case that $|L'| > n' +1$. 

%\ygc{Maybe it can help the reader (it would have helped me) to mention that anyway $|L'| \leq 3n'$ with $\rho = \frac{2}{3}$, just to see that the gap is not that huge. And actually we can assume that $|L'| < 2n'$ (you write it later, but it doesn't require the handling process to hold, right?}

We do this in a sequence of transformations that reduce it to the case of $|L'| \le n' + 1$. To present this sequence of transformations, we introduce some notation.

Recall that there are $n'$ active agents, and there are $n_2$ loss events of type $\ell_2$ and $n_4$ loss events of type $\ell_4$. We have a partition into $n' + n_2 + n_4$ bundles (each bundle may contain goods, and also money from goods previously sold, and money from goods that $i$ intends to sell), where each has $\bar{v}_i$ value exactly $MMS_i$. (If there were parts of higher value, we may reduce the values of some of the goods or give up some of the payments so that each part has value exactly $MMS_i$, see Corollary~\ref{corr:value_readjust}). In our transformations we shall create some acceptable pure bundles. Hence, the effective value of $n'$ will go down, because fewer additional bundles need to be created for the partition. We let $n_0$ denote the number of additional bundles that need to be created. Likewise, in the process we may discard goods that were involved in loss events, and consequently also $n_2$ and $n_4$ might change their values. We refer to their new values as $n'_2$ and $n'_4$. Finally, also the set $L'$ may shrink (as some of its goods may become part of pure bundles that we created). We refer to the set of goods remaining in $L'$ as $L_0$. Our goal is to reach a stage in which $|L_0| \le n_0 + 1$, while maintaining that the total value that we have at our disposal is at least $n_0 \cdot MMS_i$. This will allow us to apply Proposition~\ref{pro:fewSingletons}.

%two stages. In the first stage, we get rid of redundant loss events. This stage must end, as there are only finitely many loss events. In the second stage we show how to create a bundle in the canonical partition, while keeping the invariant that the total value for our disposal (after removing the loss events) is at least $n' \cdot MMS_k$. Repeating this creating process, either we create all $n'$ bundles of the canonical partition, or we reach a situation in which $|L'|$ is smaller than the $n'$ value that is in effect at that time (which is the original $n'$ minus the number of bundles that we already created for the canonical partition).

Recall that each loss event of type $\ell_4$ involves two goods, $e$ that is part of the allocated bundle, and $f$ that is sold and some of the money received is part of the allocated bundle, and the loss is at most $v_i(f) + (2\rho - 1)MMS_i$. Below, whenever we name a good $f_j$ (for some $j$), this signifies that the good participated as the $f$ good in a loss event of type $\ell_4$. 

%\ufc{Following your comment, I rewrote the next paragraph. Is it clearer now? If so, remove this comment and unclor the text that follows.}

{\bf Removing redundant loss events.} Suppose that the current partition has a bundle $B_j$ that contains two goods, $f_1$ and $f_2$, from loss events of type $\ell_4$. The two loss events of type $\ell_4$ contributed a loss of at most $v_i(f_1) + (2\rho - 1)MMS_i + v_i(f_2) + (2\rho - 1)MMS_i \le \bar{v}_i(B_j) + (4\rho - 2)MMS_i \le \bar{v}_i(B_j) + \rho \cdot MMS_i$ (the last inequality holds for our choice of $\rho = \frac{2}{3}$). As $B_j$ is part of the current partition we have that $\bar{v}_i(B_j) \le MMS_i$. Replace the two $\ell_4$ loss events by a loss event for the remaining $\rho \cdot MMS_i$, which qualifies as a loss event of type $\ell_2$, and remove $B_j$.  %\ygc{I didn't correct here yet, but this paragraph is ok: We take $\tilde{B}_j$ instead of $B_j$. The only issue in the argument might be that it has value greater than MMS, because it got the full good that was supposed to be sold. But this good is not related to $f_1, f_2$ (which by the definition of the loss event, were already sold). So, the part we dispose of does not need to contain this good, which we can keep: So the part we throw away is less than MMS.}
%\ygc{I think $\bar{v}_i(B_j)$ is not a loss event we defined. Maybe it's better to say we just assign $B_j$ without going through calling it a loss event?}

%As each payment was at most $(2\rho - 1)MMS_i$ \yotam{I didn't understand it. The $(2\rho - 1)MMS$ expression is the loss from forced selling, why is it a payment? OK Maybe I understand it in the following way: on top of $v_i(f)$, we lose at most $2\rho - 1$. So assuming all of $f$ payment was in the bundle, the loss from $e$ (which is in $S_i$, so it's a payment) is at most $2\rho - 1$. So if we were to exchange payments of $f$ elsewhere with payments of $e$, to make sure all of $f$ is in this bundle, then the payment from $e$ is at most $(2\rho - 1)$, and then this works?}, the sum of payments is no more than $\rho \cdot MMS_i$, for $\rho = \frac{2}{3}$, and so this is a loss event of type $\ell_2$. 
The number of bundles in the partition decreased by one, $n'_4$ decreased by two, whereas $n'_2$ increased by one.

As the transformation above reduces the number of bundles, it can be employed only finitely many times. Hence, we necessarily reach a stage in which no bundle $B_j$ contains two goods of type $f$ from loss events of type $\ell_4$.

%\ygc{I think it's a bit confusing to conclude \textbf{each} operation type in this manner (what's the end result of it). Because we don't know (maybe it's true, I didn't fully check) that the other handling types don't re-introduce things that we got rid off in this part. But what matters is that at least some of the parameters (overall loss events, $|L'|$, number of bundles. So, we can be sure that after repetitively applying these $3$ types of operations we will (i) finish, (ii) not have bundles and loss events of the types they handle.}

{\bf Handling pairs of goods from $L'$ in the same bundle.}
Suppose that the current partition has a bundle $B_j$ that contains two complete goods from $L'$. Then we create from it a pure bundle $B'$ in the canonical partition by including only these two goods, and discarding whatever is left from $B_j$ (either goods or money). We have that $\bar{v}_i(B') \ge 2(1 - \rho)MMS_i = \rho \cdot MMS_i$, where the last equality holds because $\rho = \frac{2}{3}$. Without $B'$, the total number of bundles decreased by one, $n_0$ decreased by one, $|L'_0|$ decreased by two, and $n'_2$ and $n'_4$ remain unchanged. %\ygc{Don't we also need to show $\bar{v}_i(B') \leq MMS_i$? But that holds because $\bar{v}_i(B') \leq \bar{v}_i(B_j) \leq MMS_i$} 
%\ygc{Also, I think the word ``removing'' is confusing. We are not really removing $B'$, it's just a pure bundle now, so we can ignore it. Also for the ``removing'' redundant loss types, it's more like we are ``rearranging'' them to have less loss events}

Necessarily, we reach a stage in which every bundle $B_j$ contains at most one good from $L'$.

{\bf Mixed $L'$ and $f$ bundles.} Suppose that the current partition has two bundles, call them $B_1$ and $B_2$, each containing one good ($f_1$ and $f_2$, respectively), where each of these goods is in a loss event of type $\ell_4$. In addition,  $B_1$ and $B_2$, each contain one good ($d_1$ and $d_2$, respectively) from $L'$. Each $B_j$ may possibly contain additional goods and money. Then create the acceptable pure bundle $B' = \{d_1, d_2\}$ and put it in the canonical partition, and discard the remaining parts of $B_1$ and $B_2$ (thus discarding also $f_1$ and $f_2$). The number of remaining bundles decreased by two, $n_0$ decreased by one, $n'_4$ decreased by~2 and $n'_2$ increased by~1 (the claims regarding $n'_4$ and $n'_2$ are as in {\em removing redundant loss events}). %\yotam{If I understand this reduction correctly, $e_1,e_2$ don't necessarily correspond to $f_1, f_2$ (not in the same loss event). So let's call them $e_3, e_4$ instead for a moment. Do we implicitly 'open up' loss events and replace their $f_3, f_4$ with $f_1,f_2$?}

Necessarily, we reach a stage in which at most one bundle $B_j$ contains both an $f$ good from an $\ell_4$ loss event and a good from $L'$.

At this stage we have $n'_2$ loss events of type $\ell_2$, $n'_4$ loss events of type $\ell_4$, we need to create $n_0$ additional bundles for the canonical partition, and we have $n_0 + n'_2 + n'_4$ bundles, each of value $MMS_i$. Of these, $n'_4$ bundles contain a good $f$ from an $\ell_4$ loss event, at most one of which contains a good from $L_0$. Moreover, all goods from $L_0$ are in separate bundles. Hence, $|L_0| \le n_0 + n'_2 + 1$. If $|L_0| \ge 2n_0$ then we can create $n_0$ pairs of goods from $L_0$. Each such pair makes an acceptable pure bundle, and so we are done. Hence we may assume that $|L_0| < 2n_0$. Likewise, we may assume that $|L_0| > n_0 + 1$, as otherwise we can apply Proposition~\ref{pro:fewSingletons}.

Create $t = |L_0| - n_0 - 1 \le n'_2$ acceptable pure bundles, each containing two goods from $L_0$.
Notice that this is possible as $\frac{1}{2}|L_0| \leq n_0 + 1$, and so $t = |L_0| - n_0 - 1 \leq \frac{1}{2}|L_0|$. We remark that each of the pure bundles that we create in this step might have value as high as $2\rho \cdot MMS_i > MMS_i$. Thus, it might appear as if the value that remains at our disposal drops below $MMS_i$ times the number of bundles that we still need to create. However, the next paragraph shows the latter does not happen.

Afterwards, the number of bundles still needed for the canonical partition is $\hat{n} = n_0 - t$, whereas the new size of $L_0$ equals to $\hat{L} = |L_0| - 2t$. However, we have $\hat{L} = \hat{n} + 1$, as $|L_0| - 2t = t + n_0 + 1 - 2t = n_0 - t + 1 = \hat{n} + 1$. 
 The total value left at our disposal, which we denote by $G$, is at least $G \ge (n_0 + n'_2 + n'_4)MMS_i - n'_2 \rho MMS_i - n'_4 \cdot MMS_i - 2t \rho MMS_i$ (since each pure bundle created has value at most $2\rho$). As $n'_2 \ge t$, we have that $G \ge \left(n_0 + (1 - 3\rho)t\right)MMS_i = (n_0 - t)MMS_i = \hat{n} MMS_i$ (for $\rho = \frac{2}{3}$). Thus, we need to construct $\hat{n}$ additional bundles in the canonical partition, have at most $\hat{n} + 1$ large goods, and at least $\hat{n} \cdot MMS_i$ total value for our disposal. Proposition~\ref{pro:fewSingletons} thus applies, and we obtain the canonical partition.

%\ufc{I tried to explain two paragraphs above why there is the need for the accounting. Is it clearer now?}
%\ygc{By the end of the proof I actually don't see a reason we need to do accounting of how the parameters change after each reduction. It only matters that each operation creates acceptable bundles, with value not exceeding MMS, and that there's only a limited amount of times it repeats.
%Also, I guess it matters that the number of bundles continue to equal $n_0 + n'_2 + n'_4. $
%The operations don't overlap so it's ok, and then we use the fact that all the reductions are completed to have $|L_0| \leq n_0 + n'_2 + 1$, so what's the accounting for? Anyway, whatever it is for, need to make it explicit. 
 
%If $|L'| > n'$ and there is no bundle $B_j$ as in the above paragraph, then necessarily $n_2 \ge 1$. Then taking any two items from $L'$ and creating a pure bundle out of them, the remaining value in their original bundles is $g \ge (2 - 2\rho)MMS_i$, where there is a loss event of type $\ell_2$ with loss $\ell \le \rho \cdot MMS_i$. Also, $g \ge \ell$ for $\rho = \frac{2}{3}$, and so the invariant is kept.
\end{proof}

Lemma~\ref{lem:canonical} established that when $\rho = \frac{2}{3}$, canonical partitions can be found whenever needed by the allocation algorithm. Eventually, every agent receives a bundle that is acceptable to her. This proves the general $n$ part of Theorem~\ref{thm:nMMS}.
\end{proof}

\begin{algorithm}[htbp]

\small

    \SetAlgoLined
\DontPrintSemicolon
\KwIn{Valuations $v$ over goods $\items$, market-price $p$, number of active agents $n'$, $B_1, \ldots, B_n$ an MMS partition according to $v$, Loss events $\ell_1, \ldots, \ell_4$ }
\KwOut{A canonical partition $B_1, \ldots, B_{n'}$  and $\forall 1\leq j\leq n', \bar{v}(B_j) \geq \frac{2}{3}\cdot MMS$}

\If{$|\ell_1| > 0$} {
    Arrange the bundles so that bundles corresponding to a loss event in $\ell_1$ have higher indices.

    return FIND-CANONICAL($v, p, n', B_1, \ldots, B_{n - |\ell_1|}, \emptyset, \ell_2, \ell_3, \ell_4)$. 
} \If{$|\ell_3| > 0$} {
        Let $e, f$ be the goods associated with $l_3$, and $r$ the payment associated with it. 
        
        \If{$\exists j, \{e,f\} \subseteq B_j$} {
            Remove $B_j$
        } \Else{
            Let $j_1, j_2$ be the bundles containing $e,f$ respectively. 
Let $B' = (B_{j_1} \cup B_{j_2} \setminus \{e,f\}) \cup r$. 

            Remove $B_{j_1}, B_{j_2}$, and add $B'$ instead. 
            
        }
        return FIND-CANONICAL($v, p, n', B, \ell_1, \ell_2, \ell_3 \setminus l_3, \ell_4$).
}

Let $L'$ be goods of value $\bar{v}(g) \geq (1-\rho)MMS$ not in $\sellable$.

\If{$|L'| \leq n' + 1$} {

    Let $P'$ be money from goods that do not appear in full in $B_1, \ldots, B_n$ (i.e., already sold in previous loss events). Let $\sellable'$ be the goods newly sold in the MMS partition.  Let $C'$ be goods not in $\sellable \cup L'$. 

    return UNSTRUCTURED-FIND-CANONICAL($v,p,n', L', P', C', \sellable')$.
}

\If{$\exists j, l_1 = (e_1, f_1)\in \ell_4, l_2 = (e_2, f_2) \in \ell_4, \{f_1, f_2\} \subseteq B_j $} {

    Remove $B_j$. Create a loss event %\yotam{
    $l = \{e_1, e_2\}$ of type $\ell_2$. % does that work..?}

    return FIND-CANONICAL($v, p, n', B, \ell_1, \ell_2 \cup \{l\}, \ell_3, \ell_4 \setminus \{l_1, l_2\}$)
}

\If{$\exists j, |L' \cap B_j| \geq 2$} {
    Let $B' = L' \cap B_j$.

    return $B'$, FIND-CANONICAL($v,p, n' - 1, B\setminus B_j, \ell_1, \ell_2, \ell_3, \ell_4$). 
}

\If{$\exists j_1, j_2, B_{j_1} \cap L' = \{d_1\},B_{j_2} \cap L' = \{d_2\}, l_1 = (e_1, f_1) \in \ell_4, f_1 \in B_{j_1}, l_2 = (e_2, f_2) \in \ell_4, f_2 \in B_{j_2}$} {
    Let $B' = \{d_1, d_2\}$. Create a loss event $l$ with %\yotam{
    $l = \{e_1, e_2\}$ of type $\ell_2$. % does that work..?}

    return $B'$, FIND-CANONICAL($v,p,n'-1,B\setminus (B_{j_1} \cup B_{j_2}), \ell_1, \ell_2 \cup \{l\}, \ell_3, \ell_4 \setminus \{l_1, l_2\}$. 
}

\If{$|L'| \geq 2n'$} {
    Create $A_1, \ldots, A_{n'}$ bundles with pairs of goods from $L'$. 

    return $A_1, \ldots, A_{n'}$.
}
\Else {
    Let $t = |L'| - n' - 1$. Create $A_1, \ldots, A_t$ bundles with pairs of goods from $L'$. Remove all goods in $\cup_{j=1}^t A_j$ from $L'$. 
    
    Let $P'$ be money from goods that do not appear in full in $B_1, \ldots, B_n$ (i.e., already sold in previous loss events). Let $\sellable'$ be the goods newly sold in the MMS partition. Let $C'$ be goods not in $\sellable \cup L'$. 

    return $A_1, \ldots, A_t$, UNSTRUCTURED-FIND-CANONICAL($v,p, n' - t, L', P', C', \sellable')$. %\yotam{here we can not pass a MMS partition because it's unstructured... so maybe the Proposition~4.5 logic ($|L'| \leq n' + 1$) needs to be extracted to another sub-routine, UNSTRUCTURED-FIND-CANONICAL. Loss events don't matter to this sub-routine so no need to pass them.}

}

 \caption{FIND-CANONICAL}
 \label{alg:FIND-CANONICAL}
\end{algorithm}

\begin{algorithm}[H]
    \SetAlgoLined
\DontPrintSemicolon
\KwIn{Valuations $v$ over goods $\items$, market-price $p$, number of active agents $n'$, $L'$ large goods, $P'$ sale proceeds, $C'$ small goods, $\sellable'$ unsold large goods}
\KwOut{A canonical partition $B_1, \ldots, B_{n'}$  and $\forall 1\leq j\leq n', \bar{v}(B_j) \geq \frac{2}{3}\cdot MMS$}

    Put a good of $L'$ in each bundle $A_1, \ldots, A_{n'}$. If a good remains, add it to $A_1$. 

    Bag-fill bundles using proceeds from $P'$ up to size $\rho MMS$. 

    Bag-fill bundles using goods in $C'$, with a threshold of $\rho MMS$.

    Bag-fill bundles that still contain a single good of $L'$, each with a single good from $\sellable'$, and without exceeding MMS. {If adding a good exceeds MMS, sell it, add exactly the amount of proceeds to reach $\rho MMS$, and treat the remaining proceeds as a virtual good.} 

    return $A_1, \ldots, A_{n'}$.

 \caption{UNSTRUCTURED-FIND-CANONICAL}
 \label{alg:UNSTRUCTURED-FIND-CANONICAL}
\end{algorithm}

%\ufc{The next section, about envy freeness, is just a draft. Most importantly, it is not clear whether Theorem~\ref{thm:23EFX} has a proof.}

\subsection{\texorpdfstring{Achieving $\frac{3}{4}$-MMS with $n=3$}{Achieving 3/4-MMS with 3 agents}}
\label{sec:3-4}

\CanonicalThree*

%\ygc{Lemma statement uses $v_i$, make sure throughout the proof we use the index $i$}

%\ufc{May require more thought to see if the proof of this lemma can be simplified.}

%\ygc{I was able to simplify the proof by taking advantage of the fact we can assume at most two sold goods in the MMS partition. This eliminates many of the cases to be considered, and also I managed to join the two remaining cases in the original proof into one. }

\begin{proof}
By Lemma~\ref{lem:number_of_split_goods}, we can consider an MMS partition that has at most $2$ goods that are sold. If the MMS partition has no sold goods, then all bundles are pure, and it is a canonical partition as is. We thus assume there is at least one sold good. 

If at most one bundle in the MMS output has more than $\frac{1}{4} \cdot MMS$ in sale proceeds, we can declare it as the leftovers bundle, remove any sale proceeds from the other two bundles, and declare them as pure bundles.  This gives a canonical partition where each bundle has value at least $\frac{3}{4}\cdot MMS$. Thus, there must be at least two such bundles, and overall $\frac{1}{2}\cdot MMS$ in sale proceeds. Since, as we showed, we can assume between one and two sold goods, we must have a good with at last $\frac{1}{4} \cdot MMS$ in sale proceeds. Denote it $F$. %\ygc{It might be that $F$ is the same as $f$ later... be more accurate with the notation} 

Next, we consider the case where we can create a singleton bundle with sale proceeds from $F$ and a full good $e$, where the bundle's value is exactly $\frac{3}{4} \cdot MMS_i$. In this case, we let $B_1$ be the singleton bundle. We do a bag-filling procedure to create $B_2$, with the threshold at $\frac{3}{4} \cdot MMS_i$. Since for each good $g$, $\bar{v}_i(g) \leq \frac{3}{4} \cdot MMS_i$, $B_2$ has value at most $1.5 \cdot MMS_i$. Overall, the remaining goods have value of at least $\left( 3 - (\frac{3}{4} + 1.5) \right)\cdot MMS_i = \frac{3}{4}\cdot MMS_i$. This is a partition with a value of at least $\frac{3}{4} \cdot MMS_i$ for each bundle, and it is canonical, as it is made of a singleton, pure, and a single leftovers bundle. 

If we cannot create such a singleton bundle, this implies that all goods $g \in \items \setminus \{F\}$ have $\bar{v}_i(g) < \frac{1}{2} \cdot MMS_i$, or otherwise, we could create a singleton bundle made of $g$ and sale proceeds from $F$. 

We consider the following construction: We use \textit{greedy} bag-filling using goods in $\items \setminus \{F\}$ to create $B_1, B_2$ with a threshold of $\frac{3}{4}\cdot MMS_i$, and put all remaining goods (including $F$) in $B_3$. By greedy bag-filling we mean adding goods in order of their $\bar{v}_i$ value (larger goods earlier), and moving to the next bundle once the current bundle has a value of at least $\frac{3}{4}\cdot MMS_i$. This is a canonical partition as $B_1, B_2$ are pure, and $B_3$ is a single leftovers bundle. So, it remains to establish that $\bar{v}_i(B_3) \ge \frac{3}{4} MMS_i$.

Since we established that for all goods $g \in \items \setminus \{F\}$ have $\bar{v}_i(g) < \frac{1}{2}$, it must be that $B_1$ contains at least two goods. First, consider if it contains exactly two goods. Then, its total $\bar{v}_i$ value is at most $MMS_i$. The value of $B_2$ is at most $(\frac{3}{4} + \frac{1}{2})MMS_i = \frac{5}{4}MMS_i$, since the last good added to the bag-filling is of value at most $\frac{1}{2}MMS_i$, and the value of the bundle before adding it is below $\frac{3}{4}MMS_i$. Overall, we conclude that $B_3$ has value of at least $(3 - (1 + \frac{5}{4}) \cdot MMS_i = \frac{3}{4} MMS_i$.

Lastly, consider if $B_1$ contains at least $3$ goods. Then, the two larger goods in $B_1$ have value less than $\frac{3}{4} \cdot MMS_i$, and so the second largest good must have value of less than $\frac{3}{8} \cdot MMS_i$. All subsequently added goods (both to $B_1$ and $B_2$) have value of at most $\frac{3}{8} \cdot MMS_i$. We thus have that both $B_1$ and $B_2$ have a value of at most $(\frac{3}{4}+\frac{3}{8}) \cdot MMS_i = \frac{9}{8} \cdot MMS_i$, which guarantees a value of at least $(3 - 2 \cdot \frac{9}{8}) \cdot MMS_i = \frac{3}{4} \cdot MMS_i$ for $B_3$. 
\end{proof}

\RemainderThreeFour*

\begin{proof}
Consider an MMS output for agent $i$ be $B_1, B_2, B_3$, each with value $MMS_i$. A bundle given away to some agent by the algorithm in stages (i)-(iii), results in a loss event for agent $i$ of one of the four types specified in Definition~\ref{def:loss_types}. 
We specify the construction of $B'_1, B'_2$ depending on the loss event type.
Notice that for the purpose of this lemma, our construction does not need to be canonical. 

\begin{itemize}
    \item $\ell_1$. With this loss event, a single good is eliminated, which the agent does not sell (but rather keeps) in the MMS partition. It thus must belong to a single bundle, say $B_1$. 
    Then, we can set $B'_1 = B_2, B'_2 = B_3$ to serve as the two disjoint bundles required by the lemma.
    
    \item $\ell_2$. With this loss event, the value lost is at most $\frac{3}{4}MMS_i$. 

    We do greedy bag-filling into bundles $B'_1, B'_2$ with a threshold of $\frac{3}{4}$-MMS, where if the good that passes the threshold is meant to be sold in the MMS partition, we only include sale proceeds to arrive at a value of $\frac{3}{4}$-MMS. 

    Consider if $B'_1$ contains a single good. Then its value is at most MMS (if it is a good meant to be kept, it has at most a value of MMS, and if it is meant to be sold, then we only take $\frac{3}{4}$-MMS of sale proceeds). We then have a remaining value of $1.5$-MMS (after removing the $\ell_2$ loss event and $B'_1$), and so $B'_2$ must be successfully filled. 

    If $B'_1$ contains at least two goods, then the first added good has a value of at most $\frac{3}{4}$-MMS, and since this is greedy bag-filling, this means that any subsequently added good has a value of at most $\frac{3}{4}$-MMS as well. Overall, this means $B'_1$ has a value of at most $(\frac{3}{4} + \frac{3}{4}) \cdot MMS = 1.5 \cdot MMS$, and we have a remaining value of at least $\frac{3}{4}$-MMS (after removing the $\ell_2$ loss event and $B'_1$), and so $B'_2$ must be successfully filled. 

    \item $\ell_3$ or $\ell_4$. With either of these loss types, we have two goods $e, f$, where $f$
    is not sold as part of the MMS partition, and a maximal combined loss of $v_i(f) + (2\rho - 1)MMS_i \leq (\frac{3}{4} + \frac{1}{2})MMS_i = \frac{5}{4} MMS_i$. We make two observations: (i) $f$ is in a single bundle of $B_1, B_2, B_3$ (since it is not sold in the MMS partition), say $B_1$. We can remove $f$ in full from $B_1$, and add back unused sale proceeds (as $f$ is used in the loss event for its sale proceeds) as we wish to the other bundles. By the combined loss bound, we can ensure that one of the two other bundles has a loss of at most $\frac{1}{4} \cdot MMS$ by removing $e$ (in full if it is kept in the MMS partition, or its various sale proceeds if it is sold), and adding back unused sale proceeds of $f$ as necessary. This is because the combined loss bound guarantees other than $v_i(f)$ we lose at most $\frac{1}{2} \cdot MMS_i$. W.l.o.g., this means that the revised bundle $B_2$ (denote it $B'_2$ after the changes) has a value of at least $\frac{3}{4} \cdot MMS$. (ii) Gather all remaining goods and sale proceeds in a bundle $B'_1$. In the loss event, at most $\frac{5}{4}MMS_i$ were removed, and $B'_2$ has at most value of $MMS_i$ (as it was created out of $B_2$, where we only add sale proceeds from $f$ if necessary to complete its value to $\frac{3}{4}MMS_i$, but anyway not exceeding the original value of $MMS_i$). This leaves at least $(3 - (\frac{5}{4} + 1)) MMS_i = \frac{3}{4}MMS_i$ in $B'_1$, and so this is an acceptable partition. 
%\ufc{Agent $i$ intends to sell $e$, so $e$ is not in any single bundle of the $MMS_i$ partition. A cleaner argument is as follows. The total loss is at most $v_i(f) + \frac{1}{2} \cdot MMS_i$. As one of the three MMS bundles lost $f$, one of the remaining two lost at most $\frac{1}{4} \cdot MMS_i$. It can serve as one bundle in the $\frac{3}{4} \cdot MMS_i$ partition, and the union of what remains of the other two bundles can serve as the other bundle (because the total value left for the union is at least $\frac{3}{4} \cdot MMS_i$).}
%\ygc{I think there's some subtlety about the fact that the $v_i(f)$ in the loss event bound doesn't necessarily come from $f$, but this can be resolved by using sale proceeds of $f$ that are not part of the loss event, anyway I wrote it down above. }
\end{itemize}
\end{proof}

Given that each agent $i$ of the two remaining agents can partition the remaining items to two disjoint bundles with value $\frac{3}{4} \cdot MMS_i$, we can apply Theorem~\ref{thm:2EFX} with the remaining items and two remaining agents, and this guarantees each of them $\frac{3}{4}MMS_i$.  

%\ufc{Finish the proof. When two agents remain, the MMS for each agent is at least $\frac{3}{4}$-MMS. We can use Theorem~\ref{thm:2agents} to conclude that ...}

%\ygc{Write the opening paragraph here as a proof. }

\section{Existence Results for Envy-Free Notions}
\label{sec:ef}

%(Related definition:  EFM in [Bei, X.; Li, Z.; Liu, J.; Liu, S.; and Lu, X. Fair Division of Mixed Divisible and Indivisible Goods. Artifcial Intelligence, 293: 103436.])

We follow our envy notions definitions of Definition~\ref{def:EF}.
We distinguish between settings in which all goods are sellable, and those in which some goods are sellable and some are not. 

If all goods are sellable, then an EF allocation exists: sell all goods and divide the money equally. Even if not all goods are sellable, in essence, the EFX existence and SEFX existence situation are similar:%, following a result of \cite{bei2021fair}. \ufc{Ambiguous wording. Did \cite{bei2021fair} prove this lemma in our setting? Clarify.}

\begin{lemma}
\label{lem:extending_indivisible_EFX}
    If for some number of agents $n$ and additive valuations, an EFX allocation always exists, then the same is true for SEFX, and vice versa. The same holds for EF1 and SEF1. 
\end{lemma}
\begin{proof}
    We prove the lemma for EFX and SEFX, and the proof for EF1 and SEF1 is the same. 
    
    The SEFX $\implies$ EFX implication is straight-forward as remarked above. In the other direction, \cite{bei2021fair} show that for a mix of divisible and indivisible goods, if an EFX allocation exists for the indivisible goods, it can be extended to an EFX allocation of divisible and indivisible goods, by only adding divisible cake to the existing allocations. Thus, if an EFX allocation is guaranteed to exist for some $n$, we can find an EFX allocation for all the goods that are not sellable, sell all the sellable goods, and extend the EFX allocation to an EFX allocation of divisible and indivisible goods using the sale proceeds from the sellable goods.
Any indivisible good $g$ allocated in this setting has $p(g) = 0$ (since it is not sellable), and so the SEFX condition is the same as that of an EFX allocation of divisible and indivisible goods. 
\end{proof}

However, an allocation that only satisfies envy-freeness properties might have very little value for the agents. Hence, we seek allocations that satisfy a relaxed version of EF, but also guarantee good value to each agent.

As a useful technical step, we now define $\epsilon$ notions of envy, but our interest is in the case when $\epsilon = 0$. We consider the more general definition since later in the discussion, it helps to consider the $\epsilon$ notions (with the end-goal of proving existence for the $\epsilon=0$ case). 

\begin{definition}
\label{def:epsilon_EF}
    An allocation $(A_1, P_1), \ldots, (A_n, P_n)$ is {\em envy free} (EF) if for every agent $i$ and other agent $j$, {$\bar{v}_i(A_i) + P_i \ge \bar{v}_i(A_j) + P_j$.} %\ufc{what is $u_i$?} 
    {An allocation satisfies the following relaxations of EF if the above envy free condition holds whenever $P_j > 0$, and either the EF condition or the following conditions hold if $P_j = 0$.}
    \begin{enumerate}
        \item $\epsilon$-SEF1 (strong/sellable envy free up to one good). $\bar{v}_i(A_i) + P_i +  \epsilon \ge \bar{v}_i(A_j \setminus \{g\}) + p(g)$ for some good $g \in A_j$.
        % EFL RELATED
        %\item $\epsilon$-SEFL (strong/sellable envy free up to one less preferred good). Either $|C_j| = 1$ and $\bar{v}_i(C_j) + \epsilon \ge p(C_j)$, or $u_i(C_i) + \epsilon \ge \bar{v}_i(C_j \setminus g) + p(g)$ for some item $g \in C_j$ satisfying $u_i(C_i) + \epsilon \ge \bar{v}_i(g)$. Since $\bar{v}_i(C_j) \ge p(C_j)$ by definition, the definition is equivalent to strengthening SEF1 with the additional requirement that either  $|C_j| = 1$, or $u_i(C_i) + \epsilon \ge \max_{g \in C_j} \bar{v}_i(g)$. 
         \item $\epsilon$-SEFX (strong/sellable envy free up to any good). $\bar{v}(A_i) + P_i + \epsilon \ge \bar{v}_i(A_j \setminus \{g\}) + p(g)$  for every good $g \in A_j$.
    \end{enumerate}
\end{definition}

The following lemma is technically useful in simplifying some of the proofs. Due to the existence of divisible sales-proceeds, envy might be by an infinitesimally small amount.  By considering $\epsilon$-SEFX rather than SEFX, iterative algorithms that we design make finite progress in each step rather than infinitesimally small progress, allowing us to conclude that they terminate in finite time. %\ufc{remove: (Later, in our later iterative algorithms, we only guarantee $\epsilon$-SEFX at each step, as it allows us to ensure $\epsilon$ progress at each step).}  

\begin{lemma}
    \label{lem:limit}
    For a given allocation instance with sellable goods and a given value of $\rho$ and individual share values $s_i$, suppose that for every $\epsilon > 0$ there is an allocation (a partial allocation, respectively) that is $(\rho - \epsilon) \cdot s_i$ and $\epsilon$-SEFX. Then there also is an allocation (partial allocation, respectively) that is $\rho \cdot s_i$ and SEFX. %The above also holds if SEFX is replaced by SEFL.
\end{lemma}

\begin{proof}
    We prove the lemma for partial allocations and SEFX. %The proof for the SEFL case is similar.

%\ufc{This is a sketch. Should possibly be rewritten for the full version.}
In any partial allocation, each good is either allocated to a single agent, or sold, or not allocated at all. Hence, there are $(n+2)^m$ classes of partial allocations, conditioned on the decision for each good. Within a class, one need only distribute the money so as to minimize $\epsilon$. Refer to the minimum attainable value of $\epsilon$ within a class $c$ as $\epsilon_c$. Minimizing $\epsilon_c$ (subject to $\epsilon_c \ge 0$) is a solution to a linear program (as both the share guarantee and SEFX constraints are linear in the sale proceeds and $\epsilon_c$) over a closed and bounded set. As such, the optimal value of $\epsilon_c$ is either~0 or bounded away from~0. Consequently, the optimal value of $\epsilon$ (which is the minimum over $\epsilon_c$ for finitely many indices $c$) is either~0 or bounded away from~0. Under the terms of the lemma, it must be~0.
\end{proof}

%The following proposition is based on known results from the setting with no sellable items~\cite{}. These results hold also for general monotone valuations. We state our proposition for additive valuations with sellable items, as this is the focus of this paper.

%\begin{proposition}
 %   \label{pro:fullEFL}
 %   Consider a setting with indivisible sellable items and additive valuation. Then every partial allocation $A = A_1, \ldots, A_n$ can be extended in polynomial time to a full allocation $A' = A'_1, \ldots, A'_n$ in which $v_i(A'_i) \ge v_i(A_i)$ for every agent $i$, with the following additional properties.

  %  \begin{enumerate}
  %  \item If $A$ is SEFL, then $A'$ is SEFL.
  %      \item If $A$ is SEFX and all items not allocated by $A$ are sellable, then $A'$ is SEFX.
  %  \end{enumerate}
%\end{proposition}

%\begin{proof}
 %   \ufc{Fill in the proof. For second part: sell all remaining items and pour the money into agents that no one envies, doing cycle eliminations when needed (including if $i$ envies $j$ but $j$ is indifferent to $i$). For the first part, handle first the items that are not sellable. Note that to get SEFL and not just SEF1, one may need tricks similar the case with no sellable items~\cite{feigeRMMS}.}
%\end{proof}

\section{Algorithms and Missing Proofs for TPS Approximation}
\label{app:tps}

Let $f(t) = \frac{1}{n}\sum_{j\in \items} \max \{p(j), \min \{v(j), t\}\}$.

The following claims is useful:

\begin{claim}
\label{clm:tps_median_value}
    If for some $t$, $f(t) \geq t$, then $TPS(\items,v,p, n) \geq t$. 
\end{claim}
\begin{proof}
    If $f(t) = t$, then $TPS(\items,v,p, n) \geq t$ holds by the definition of TPS as the maximal $t$ value satisfying this equality. If $f(t) \neq t$, notice that $g(x) = f(x) - x$ is a continuous function with $g(t) > 0$. 

    By Lemma~\ref{lem:mms_tps_ps}, $g(t') < 0$ for any $t' > PS(\items,v,p, n)$.  Thus, there must be some $t'$ so that $t' > t$ and $g(t') < 0$. By the median value theorem, there is some $t < t'' < t'$ with $g(t'') = 0$, and so $TPS(\items, v,p, n) \geq t'' > t$. 
\end{proof}

\MMSTPSPS*

\begin{proof}

($TPS \leq PS$)

    Let $t' > PS(\items,v,p, n)$. Then, 
    
    $$f(t') = \frac{1}{n}\sum_{j\in \items} \max \{p(j), \min \{v(j), t'\}\} \leq \frac{1}{n}\sum_{j\in \items} \max \{p(j), v(j)\} < t'.$$

    Since the TPS is the maximal value so that the above holds as equality, the TPS must be at most the proportional share. 

($MMS \leq TPS$)

    Consider an MMS partition $C_1, \ldots, C_n$. Let $\sellable_C$ be the goods sold in the partition, and $V_C$ the goods kept. Every good kept $j$ contributes at most $$\min \{v(j), MMS\}$$ to their bundle (and they are included in a single bundle), and every sold good contributes at most $p(j)$ to all bundles. If we sum over the $n$ MMS bundles, we get $n \cdot MMS \leq \sum_{j\in \sellable_C} p(j) + \sum_{j\in V} \min \{v(j), MMS\}$, and thus (dividing by $n$) we get:
    
    \[
    \begin{split}
    MMS & \leq \frac{1}{n} \left( \sum_{j\in \sellable_C} p(j) + \sum_{j\in V_C} \min \{v(j), MMS\} \right) \\
    & \leq \frac{1}{n}  \sum_{j\in \sellable_C} \max \{p(j),\min \{v(j), MMS\} \}  \end{split}
\]

It then follows that $TPS(\items,v,p, n) \geq MMS$ by Claim~\ref{clm:tps_median_value}. 
\end{proof}

%\ufc{I do not understand the purpose of next claim. Isn't this simply part of the definition of the TPS? Maybe you meant to show that that there is a threshold value for $t$, with all values above it bad and all values below it good?}
%\ygc{It's not a part of the definition of the TPS because the definition requires equality. It could be that $f(t) > t$ for some $t$ but then there is not higher $t'$ so that $f(t') = t'$, and the claim is meant to show there is always such $t'$.}

%\ygc{Add some illustrating drawing of the above claim}

%We can now show the following:

%\ufc{Seems like a statement that can be made, with proofs moved to the appendix}

%\ufc{I think that the earlier comment of mine above (now commented out) was misunderstood. I meant the the above lemma and claim should not be in the main part of the paper. For the lemma below, refer the reader to the appendix, or provide a very short proof.}
%\ygc{It seems like Lemma 5.8 and Claim 5.9 lead to Lemma 5.10, but they don't (they're more here for Claim 5.11), and also there's no proof, so why did we need them..? Make the order and writing less confusing. Also, why repeat the PS claim if it's in Lemma 5.8}

The TPS is efficiently computable, by the following simple algorithm.

%Let $f(t) = \frac{1}{n}\sum_{j\in M} \max\{p(j), \min\{v(j), t\}\}$. 
Consider the set of (at most $2m+1$) points $S = \{0\} \cup \bigcup_{j\in \items} \{v(j), p(j) \}$, and sort the points in decreasing order $s_1, \ldots, s_{|S|}$. 
Find the first point $k$ so that $f(s_k) \geq s_k$. 

\begin{claim}
\label{clm:s_k}
    If $k=1$, then $TPS(\items,v,p, n) = PS(\items,v,p, n)$. If $k>1$, then
    $s_k \leq TPS(\items,v,p, n) < s_{k-1}$. 
\end{claim}
\begin{proof}
    By Claim~\ref{clm:tps_median_value}, the TPS is at least $s_k$. 

    If $k=1$, then for $t = s_1$, it holds that $t \geq v(j)$ for all $j\in \items$. Thus, $f(s_1) = \frac{1}{n} \sum_{j\in \items} \bar{v}(j)$, and $f$ is weakly monotonically increasing in $t$. Thus, $f(PS(\items,v,p, n)) \geq f(s_1) = PS(\items,v,p, n)$, and so 
    $$PS(\items,v,p, n) \stackrel{\text{Claim~\ref{clm:tps_median_value}}}{\leq} TPS(\items,v,p, n) \stackrel{\text{Lemma~\ref{lem:mms_tps_ps}}}{\leq} PS(\items,v,p, n). $$

    If $k>1$, then for any $k' < k$, $f(s_{k'}) < s_{k'}$. Thus, none of the points can be the TPS (as it does not satisfy the equality condition), but also the interior intervals cannot be the TPS, as within each interval $f$ is linear since for each $j$, $\max \{p(j), \min \{v(j), t\}\}$ remains either equal to $p(j), v(j),$ or $t$, and $f$ is the sum of these linear expressions. If there was an interior point so that $f(t) = t$, and given $f(s_{k'}) < s_{k'}$, that implies that $f(t) - t$ is increasing within the interval, and so it must be that $f(s_{k'+1})\geq s_{k'+1}$, in contradiction. 
\end{proof}

Thus, if $k = 1$ we can directly compute the TPS by Claim~\ref{clm:s_k}, and if $k>1$, we can pinpoint the interval where the TPS resides. Since $f$ takes a linear form within the interval $[s_k, s_{k-1})$, we can solve the one-variale linear equation $f(t) = t$ with the appropriate coefficients, which we can determine by choosing a point $t \in (s_k, s_{k-1})$, and seeing whether $\max\{p(j), \min \{v(j), t\}\}$ equals $p(j), v(j),$ or $t$, for every $j\in \items$. 

\subsection{TPS Approximation}

\begin{algorithm}[htb]
    \SetAlgoLined
\DontPrintSemicolon
\KwIn{Valuation functions $v_1, \ldots, v_n$ and market price function $p$ over goods $\items$, number of agents $n$}
\KwOut{An allocation $A,P$ with at least $\frac{n}{2n-1} \cdot TPS$ of value for each agent}

$AVAIL = \items$

$ACTIVE = [n]$

\While{$\exists j\in AVAIL, i\in ACTIVE,  p(j) \geq \frac{n}{2n-1}TPS_i$} {

    Let $i^* \in \argmin_{i\in ACTIVE} TPS_i$.

    Assign $\frac{n}{2n-1}TPS_{i^*}$ in sale-proceeds from selling $j$ to agent $i^*$.

    $D'_j = p(j) - \frac{n}{2n-1}TPS_{i^*} $ // The remaining sale proceeds from $j$, which we treat as a virtual good.

%\ufc{Is a step missing here saying that other agents can take parts of $D'_j$. Also, my understanding is that there is no need to join money from different $j$. Right?} \ygc{$D'_j$ is added as an ''item'' to AVAIL, so if another agent sees enough value in it, they can take it in the next iteration. We don't mix money from different $j$-s: Either pull money from a single $j$, or treat them as indivisible goods}

    $AVAIL = AVAIL \setminus \{j\} \cup \{D'_j\}, ACTIVE = ACTIVE \setminus \{i^*\}$. 

} 

\While{$\exists j\in AVAIL, i\in ACTIVE,  v_i(j) \geq \frac{n}{2n-1}TPS_i$} {

    Assign $j$ to agent $i$.

    $AVAIL = AVAIL \setminus \{j\}, ACTIVE = ACTIVE \setminus \{i\}.$

}

\While{ $|ACTIVE| \neq 0$} {

    Bag-fill complete goods into a bundle $B$, until some agent $i\in ACTIVE$ values it at least $\frac{n}{2n-1} TPS_i$. Assign $B$ to agent $i$ (if there are several agents, choose one arbitrarily). 

    $AVAIL = AVAIL \setminus B, ACTIVE = ACTIVE \setminus \{i\}$. 
}
%Allocate any remaining goods in AVAIL to an arbitrary agent. 

 \caption{APX-TPS}
 \label{alg:APX-TPS}
\end{algorithm}

Consider a good $j$ that was sold and some of its sale proceeds were allocated. Let $D'_j$ be the remaining sale proceeds. Then, we can treat $D'_j$ as a new \textit{virtual indivisible good}. The good has the value $D'_j$ for all agents. We can ``sell'' it, which means taking some of its sale proceeds, or we can ``allocate'' it as a complete good (during bag-filling).

\halfTPS*

\begin{proof}
Any agent $i$ that is allocated a bundle gets at least $\frac{n}{2n-1}TPS_i$. This is since it receives a bundle for the first time either in the large goods selling loop (line 3), the large good keeping loop (line 9), or in the bag-filling loop (line 13), whereas in each of these cases it has a value of at least $\frac{n}{2n-1}MMS_i$ for the bundle.

We show that every agent gets allocated a bundle during the run of the algorithm. 
\begin{lemma}
    For every active agent $i$, after $k$ agents are allocated a bundle, it values the remaining goods at least $(n-\frac{2nk}{2n-1})TPS_i$. 
\end{lemma}

\begin{proof}
Let $K_1$ be the agents allocated a bundle in the first loop out of goods $g$ so that agent $i$ has $p(g) \geq \frac{n}{2n-1}TPS_i$, $K_2$ agents allocated a bundle in the first loop out of goods agent $i$ has $p(g) < \frac{n}{2n-1}TPS_i$, and $K_3$ agents allocated a bundle in the second loop. We have $k = \sum_{t=1}^3 |K_t|$. 

Notice the following about goods in $K_1$: It holds that $TPS_i \leq \frac{2n}{2n-1}TPS_i = \frac{n}{2n-1}TPS_i + p(g)$. Then,
\begin{equation}
\label{eq:k1}
\max \{p(g), \min \{v_i(g), TPS_i\} \} \leq \max \{p(g), TPS_i\} \leq p(g) + \frac{n}{2n-1}TPS_i + p(g).
\end{equation}

Notice the following about goods in $K_2$: If $p(g) > v_i(g)$, then 

\[
\begin{split}
& \max \{p(g), \min \{v_i(g), TPS_i\} = \max \{p(g), TPS_i\} \\
& \leq \max \{\frac{n}{2n-1}\cdot TPS_i, TPS_i\} \leq TPS_i.
\end{split}
\]
If $p(g) \leq v_i(g)$, then $\max \{p(g), \min \{v_i(g), TPS_i\} = \min \{v_i(g), TPS_i\} \leq TPS_i$. To conclude, in both cases,
\begin{equation}
\label{eq:k2}
\max \{p(g), \min \{v_i(g), TPS_i\} \leq TPS_i.\end{equation}

We can write, by the definition of $TPS_i$:

\begin{align}
\label{eq:TPS}
& n \cdot TPS_i \nonumber \\
& = \sum_{g\in \items} \max \{p(g), \min \{v_i(g), TPS_i\} \} \nonumber\\
& \stackrel{\text{Eq.~\ref{eq:k2}}}{\leq} \sum_{g\in \items \setminus \cup_{j\in K_2} B_j} \max \{p(g), \min \{v_i(g), TPS_i\} \} + \sum_{g\in \cup_{j\in K_2} B_j} TPS_i  \nonumber \\
& \leq   \sum_{g\in \items \setminus \cup_{j\in K_2} B_j} \max \{p(g), \min \{v_i(g), TPS_i\} + |K_2| TPS_i \nonumber\\
& \stackrel{\text{Eq.~\ref{eq:k1}}}{\leq} \sum_{g\in \items \setminus \cup_{j\in (K_1 \cup K_2) } B_j} \max \{p(g), \min \{v_i(g), TPS_i\} + |K_1| \cdot \frac{n}{2n-1}TPS_i \\
& + \sum_{g\in \cup_{j\in K_1} B_j} p(g) + |K_2| TPS_i \nonumber\\
& = \sum_{g\in \items \setminus \cup_{j\in (K_1 \cup K_2) } B_j} \max \{p(g), \min \{v_i(g), TPS_i\} \nonumber\\
& + \sum_{g\in \cup_{j\in K_1} B_j} \left( p(g) - \frac{n}{2n-1}TPS_i \right) + |K_1| \frac{2n}{2n-1}TPS_i + |K_2| TPS_i \nonumber\\
& = \sum_{g\in \items \setminus \cup_{j\in (K_1 \cup K_2 \cup K_3) } B_j} \max \{p(g), \min \{v_i(g), TPS_i\} \\
& + \sum_{j\in K_3} \sum_{g\in B_j} \max \{p(g), \min \{v_i(g), TPS_i\} \nonumber\\
 & \qquad + \sum_{g\in \cup_{j\in K_1} B_j} \left( p(g) - \frac{n}{2n-1}TPS_i \right) + |K_1| \frac{2n}{2n-1}TPS_i + |K_2| TPS_i \nonumber\\
& = \sum_{g\in \items \setminus \cup_{j\in (K_1 \cup K_2 \cup K_3) } B_j} \max \{p(g), \min \{v_i(g), TPS_i\} + \sum_{j\in K_3} \sum_{g\in B_j} \bar{v}_i(g) \nonumber\\
& \qquad + \sum_{g\in \cup_{j\in K_1} B_j} \left( p(g) - \frac{n}{2n-1}TPS_i \right) + |K_1| \frac{2n}{2n-1}TPS_i + |K_2| TPS_i \nonumber\\
& \leq \sum_{g\in \items \setminus \cup_{j\in (K_1 \cup K_2 \cup K_3) } B_j} \max \{p(g), \min \{v_i(g), TPS_i\} + |K_3|\frac{2n}{2n-1}TPS_i \\
& + \sum_{g\in \cup_{j\in K_1} B_j} \left( p(g) - \frac{n}{2n-1}TPS_i \right) + |K_1| \frac{2n}{2n-1}TPS_i + |K_2| TPS_i.
\nonumber
%\end{align*}
\end{align}

Now, we can reduce terms from both sides, and arrive at:

\[
\begin{split}
& \left( n - \frac{2n \cdot k}{2n-1}\right)\cdot TPS_i= \left( n - \frac{2n}{2n-1}(|K_1| + |K_2| + |K_3|)\right)\cdot TPS_i\\
& \leq \left( n - |K_2| - \frac{2n}{2n-1}|K_1| - \frac{2n}{2n-1}|K_3|\right)\cdot TPS_i \\
& \stackrel{\text{Eq.~\ref{eq:TPS}}}{\leq} \sum_{g\in \items \setminus \cup_{j\in (K_1 \cup K_2 \cup K_3) } B_j} \max \{p(g), \min \{v_i(g), TPS_i\} \\
& + \sum_{g\in \cup_{j\in K_1} B_j} \left( p(g) - \frac{n}{2n-1}TPS_i \right)  \\
& \leq \sum_{g\in \items \setminus \cup_{j\in (K_1 \cup K_2 \cup K_3) } B_j} \bar{v}_i(g) + \sum_{g\in \cup_{j\in K_1} B_j} \left( p(g) - \frac{n}{2n-1}TPS_i \right) \\
%& \sum_{g\in \items \setminus \cup_{j\in (K_1 \cup K_2 \cup K_3) } B_j} \bar{v}_i(g) + \sum_{g\in \cup_{j\in K_1} B_j} \left( p(g) - \frac{n}{2n-1}TPS_j \right),
\end{split}
\]

whereas the RHS of this inequality is a lower bound on the remaining value after $k$ agents were allocated a bundle. 
%\ufc{Though I believe that the theorem is true, it seems much easier to try to read it myself than to check the proof.} \ygc{I went over the chain of inequalities, it's accurate, maybe a bit ``over-detailed'' but this way it can be mechanistically verified.}
\end{proof}

Now, to see that the algorithm terminates, consider $k \leq n-1$. Then, the remaining value is at least:
$$\left( n - \frac{2n \cdot k}{2n-1}\right)\cdot TPS_i \geq\frac{n\cdot (2n-1) - 2n \cdot (n-1)}{2n-1} \cdot TPS_i = \frac{n}{2n-1} TPS_i,$$
and so we are able to fill another bundle. We conclude that we fill $n$ bundles, at which point no active agent remains and the algorithm terminates. 

\end{proof}

\subsection{TPS Approximation With SEFX Allocations}

A shrinking operation finds a minimal subset $B' \subseteq B$ of the original bundle $B$ for some agent $i$, that is of greater value than the bundle they currently hold, and so that no other agent $\epsilon$-SEFX-envies $B'$. We assign $B'$ to agent $i$. W.l.o.g., $B'$ requires selling at most one good. If such a good is sold, refer to it as $e_{B'}$. Sale proceeds of $e_{B'}$ that are not assigned to $i$ are kept unassigned. If in some future iteration $B'$ is released (because agent $i$ is envious of another bundle, which is then assigned to agent $i$), these remaining sale proceeds and the sale proceeds that were included in $B'$ are deleted, and the good $e_{B'}$ is returned to the pool of unallocated goods (undoing the sale of $e_{B'}$). This allows us to ensure that each agent is responsible for at most a single forced sale, from the bundle it is currently allocated (as in Lemma~\ref{lem:number_of_split_goods}). %\ufc{Remove the following text, that was not helpful, because it too did not explain that items are bought back: It prevents the case where (i) Agent $i$ forces a sale of a good, (ii) the remaining proceeds of the good are assigned to agent $j$, and then (iii) Agent $i$ steals a bundle, and forces the sale of another good.} 

%\ufc{The explanation of why the leftovers are not assigned is not clear, but this may be okay if the appendix offers a clearer explanation}
%\ygc{I rewrote, not sure if it is much clearer. The proof itself is not in the appendix but rather Theorem 5.14.}
%\ufc{I think that I now understand what was not clear. The paragraph only talked about sale proceeds, and kept on selling items without ever buying them back. However, I think that you intended to buy back items, but just never wrote that this is what happens. Please see if my edits are correct.}

This, overall, allows us to maintain two important properties of Algorithm~\ref{alg:APX-TPS} for the TPS approximation: (i) Every assigned agent has a bundle of value at least $\frac{n}{2n-1}$-TPS. 
This is because the shrinked bundle has at least the value of the existing allocation for a stealing agent, which is of value at least $\frac{n}{2n-1}$-TPS itself. 
(ii) Every unassigned agent either has a value of at most $\frac{2n}{2n-1}$-TPS for an assigned bundle, or that bundle contains at most one good. %, and all other items are available \textit{without a forced sale}. 
With these properties maintained, our TPS approximation argument follows through, and we can focus on proving that the allocation is $\epsilon$-SEFX (by Lemma~\ref{lem:limit}, this allows us to find a SEFX allocation as well), and that the algorithm always terminates. One important element of our shrinking operation is that it forces progress by adding something small to the minimal bundle (either $\epsilon$ in sale proceeds, or rounding up sale proceeds to a full good). We can motivate the forced progress by the following example: 
%\ufc{this is not "necessary" but only for convenience. Instead of doing this epsilon progress, we can presumably directly jump to a good future point. In general, it is not a good policy to claim that parts of a proof are necessary or needed. Often, they are just a choice made for convenience, and other approaches also work. A better way of explaining a proof is by pointing out difficulties that we encounter and how we overcome them, and not say that there is no alternative proof. Also, the example seems to illustrate a point that is obvious and aready explained (in a revised text) just before Lemma~\ref{lem:limit}.}
\begin{example}
\label{ex:epsilon_SEFX_reason}
    Consider $n=2$ agents and two goods $g_1, g_2$, with valuations and market price as given in Table~\ref{tab:epsilon_SEFX_reason}. The good $g_2$ has high market value (more than $\frac{n}{2n-1} = \frac{2}{3}$-TPS for both agents), and so following a moving-knife procedure like in Algorithm~\ref{alg:APX-TPS}, we allocate $\frac{2}{3}$ in sale proceeds to agent $2$ (that has the lower TPS requirement). Still, the remaining $\frac{4}{3}$ in sale proceeds are more than $\frac{2}{3}$ of the TPS of agent $1$, and so our next step is to consider $1$ in sale proceeds to be allocated to agent $2$. If we do the shrinking operation, \textit{without forcing $\epsilon$ progress}, agent $1$ SEFX-envies any amount of sale proceeds greater than $\frac{2}{3}$, and so they would be allocated $\frac{2}{3}$ of the remaining sale proceeds, and release their current bundle, which is also $\frac{2}{3}$ in sale proceeds. This may then continue indefinitely. If, on the other hand, we force $\epsilon$ progress, this process ends up with a partial allocation of $1$ in sale proceeds to both agents, which is both $\frac{2}{3}$-TPS, and SEFX. 
\end{example}

\begin{table}[ht]
    \centering

\renewcommand{\arraystretch}{1.3} % Adds vertical padding for fractions
    \begin{tabular}{lccc}
        \toprule
        & $g_1$ & $g_2$ & TPS \\
        \midrule
        Agent 1 & $1$ & $0$  & $1.5$ \\
        Agent 2 & $0$ & $0$  & $1$ \\
        \midrule
        Market price $p$ & $0$ & $2$ & -- \\
        \bottomrule
    \end{tabular}
    \caption{Subjective valuations and market prices for Example~\ref{ex:epsilon_SEFX_reason}}
    \label{tab:epsilon_SEFX_reason}
\end{table}

A few notes on terminology:

For two bundles $X = (A, P), Y = (A', P')$, we say $X \sqsubseteq Y$ is a \textit{sub-bundle} of Y if $A \subseteq A'$, and for some $g \in A' \setminus A$, $P \leq P' + p(g)$. Notice that because our notion of a generalized bundle contains both full goods and sale proceeds, our notion essentially considers removing proceeds, then selling a good, removing its proceeds, and so on.

In Algorithm~\ref{alg:APX-TPS-EFX} we use the term $Unshrink(B')$ to mean the following: Every full good in $B'$ gets added to AVAIL. For a partial good, if there is a ``virtual'' good in AVAIL that consists of leftovers from the same good, it is joined with it to form a new ``virtual good'' (or, if it now contains the full good, then the original full good), and if there is no such virtual good, it is added as a virtual good to AVAIL. Notice that due to the operation of Algorithm~\ref{alg:APX-TPS-EFX}, that leaves all leftovers of a good that was shrinked during bag-filling unassigned, the effect is the following: Any good that was shrinked during bag-filling reappears as the full good after Unshrink. For a good that was shrinked during the large good selling/assignment loop, it depends on whether after Unshrink there are still agents who are allocated parts of that good: If there are, then all leftovers are joined into one virtual good, and if there are no such agents, then the full good is in AVAIL. This maintains the property that a good is divided only if there is an agent that holds part of its proceeds. Notice also that in Shrink (Algorithm~\ref{alg:Shrink}), we shrink so that only at most a single good is divided. All in all, this guarantees that there are no more goods sold at any step than agents that are currently allocated a bundle, which is a property we need for the TPS approximation.

\begin{algorithm}[htb]
    \SetAlgoLined
\DontPrintSemicolon
\KwIn{Bundle $X$, minimal progress $\epsilon$, valuations $v_1, \ldots, v_n$, market price $p$, Current bundle assignments $B_1, \ldots, B_n$}
\KwOut{A shrinked bundle $X_{shrink}$, corresponding agent $i'$}

Let $ACTIVE$ be the agents $i$ so that $B_i = \emptyset$, and $INACTIVE$ the remaining agents. 

Let $X_{shrink} \sqsubseteq X$ be a sub-bundle of $X$ %\ufc{if $X$ contains money, what does "subset" mean?} \ygc{I mean by it a smaller amount of proceeds, or iteratively selling a good and cutting into its proceeds...} 
so that there is $i\in ACTIVE$ so that $\bar{v}_i(X_{shrink}) \geq \frac{n}{2n-1}TPS_i$. Let $i'$ be the corresponding agent.

\If{$\exists i \in INACTIVE$ so that $i$ $\epsilon$-SEFX-envies $X_{shrink}$} {

Let $X_{shrink} \sqsubseteq X$ be a maximal bundle so that no agent SEFX-envies any strict sub-bundle of it.  // We can always find $X_{shrink}$ by repetitively following the SEFX-envy violation and either selling a good or removing proceeds accordingly. 

Let $\delta = \max_{i\in INACTIVE} \bar{v}_i(X_{shrink}) - \bar{v}_i(B_i)$, and let agent $i$ be a maximizing agent. 

\If{$\delta < \epsilon$} {

If a good $g$ has leftover sale proceeds of at least $\epsilon - \delta$, add the sale proceeds to $X_{shrink}$. If there is no such good, round up (add all the remaining sale proceeds) from some sold good $g$. }

} 

Return $X_{shrink}, i$.

 \caption{Shrink}
 \label{alg:Shrink}
\end{algorithm}

\begin{algorithm}[H]
    \SetAlgoLined
\DontPrintSemicolon
\KwIn{Valuations $v_1, \ldots, v_n$, market price function $p$ over goods $\items$, number of agents $n$}
\KwOut{A partial allocation $\mathbf{B}$ with at least $\frac{n}{2n-1} \cdot TPS_i$ utility for each agent, no $\epsilon$-SEFX envy}

$AVAIL = \items, ACTIVE = [n], \forall i, A_i = \emptyset$. 

\While{ $|ACTIVE| \neq 0$} {

\While{$\exists j\in AVAIL, i\in ACTIVE, p(j) > \frac{n}{2n-1}TPS_i$} {

    Let $i^* \in \argmin_{i\in ACTIVE} TPS_i$.

    Let $\alpha$ be $\frac{n}{2n-1}\cdot TPS_{i^*}$ in sale proceeds from $j$. 

    $X_{shrink}, i' =$ Shrink($\{\alpha\}, \epsilon, \mathbf{v}, p, \mathbf{B}$) 

    $AVAIL = AVAIL \setminus \{X_{shrink}\}$

    \If{$i' \not \in ACTIVE$} {
        $AVAIL = AVAIL \cup \{B_{i'}\}$ // add $i'$ current bundle back to the pool
    } \Else{
        $ACTIVE = ACTIVE \setminus \{i'\}$
    }

   $B_{i'} = X_{shrink}$; // $i'$ gets the shrinked bundle
} 

\While{$\exists j\in AVAIL, i\in ACTIVE, v_i(j) > \frac{n}{2n-1}TPS_i$} {

    $X_{shrink}, i' =$ Shrink($\{j\}, \epsilon, \mathbf{v}, p, \mathbf{B}$) 

    $AVAIL = AVAIL \setminus \{X_{shrink}\}$

    \If{$i' \not \in ACTIVE$} {
        $AVAIL = AVAIL \cup \{B_{i'}\}$ // add $i'$ current bundle back to the pool
    } \Else{
        $ACTIVE = ACTIVE \setminus \{i'\}$
    }

   $B_{i'} = X_{shrink}$; // $i'$ gets the shrinked bundle
} 

    Bag-fill complete goods into a bundle $X$, until some agent $i\in ACTIVE$ values it more than $\frac{n}{2n-1} TPS_i$. 

    %\ufc{Just to confirm: is leftover money also used for bag filling?} \ygc{Yes, treated as a virtual indivisible good}

    $X_{shrink}, i' =$ Shrink($X, \epsilon, \mathbf{v}, p, \mathbf{B}$) 

    $AVAIL = AVAIL \setminus X$. 

    \If{$i' \not \in ACTIVE$} {
        $AVAIL = AVAIL \cup \{Unshrink(B_{i'})\}$
    } \Else{
        $ACTIVE = ACTIVE \setminus \{i'\}$
    }

    $B_{i'} = X_{shrink}$;
}

%Let $D' = \cup_{1\leq i' \leq n} D'_{i'}$. 

%$Envy-Graph-Waterfill(D', B_1, \ldots, B_n)$

 \caption{APX-TPS-$\epsilon$-SEFX}
 \label{alg:APX-TPS-EFX}
\end{algorithm}

\TPSAlloc*
\begin{proof}
Any agent $i$ that is assigned a bundle during the run of the algorithm gets at least $\frac{n}{2n-1}TPS_i$. This is since it can be assigned a bundle for the first time either in the large goods assignment loop, or in the bag-filling loop, whereas in each of these cases it has a value of at least $\frac{n}{2n-1}MMS_i$ for the bundle. Then in any subsequent reallocation for the agent, its value increases. 

We show that every agent gets allocated a bundle. First, we note we cannot ``get stuck'' in the large good selling/assignment loop. The reason is that at each iteration, either another agent gets allocated a bundle, or that \textit{Shrink} allocates to an agent $i'$ that is already allocated a bundle, so that it receives an additional $\epsilon \min_i TPS_i$ in value, or a full good. We thus always make a Pareto improvement w.r.t. the agents that are already allocated a bundle, with either at least $\epsilon$ of additional value, or along a finite set (of the possible item-agent matchings). Similarly, we cannot be stuck in the general allocation loop. 
In the general allocation loop, the bag-filling step is always completed successfuly. This is because the Shrink operation only makes bags less valuable in the perspective of agents that are not assigned them, compared with if the bag was not shrinked. Thus, at any iteration of the bag-filling loop, the agents that are not assigned a bundle satisfy the conditions of Theorem~\ref{lem:halfTPS}, and so we are guaranteed that we can continue filling bags. Given that, since we make a Pareto improvement w.r.t. the agents that are already allocated a bundle, with either at least $\epsilon$ of additional value, or along a finite set (of the possible item-agent matchings), the loop terminates with all agents assigned a bundle. 

The $\epsilon$-SEFX property is maintained inductively among the agents that are allocated a bundle at any step of the algorithm. Initially, no agent is allocated a bundle, so it holds in a vacuous sense. At each step when an agent is allocated a bundle (or reallocated a bundle), the shrinking operation ensures the property is maintained. 

Then, Lemma~\ref{lem:limit} shows that since we can guarantee a partial allocation with $\frac{n}{2n-1}$-TPS together with $\epsilon$-SEFX for any $\epsilon$, we can also guarantee a partial allocation with $\frac{n}{2n-1}$-TPS together with SEFX. Then, we can use the process of \cite{bei2021fair} (that we also use in Lemma~\ref{lem:extending_indivisible_EFX}) to extend it to a full allocation with the same properties. 
\end{proof}

\section{Pseudo-Polynomial Algorithm and PTAS for MMS Computation}
\label{sec:computation}

It is well-known that computing the MMS in the setting of additive indivisible goods, with a fixed number of agents, even with $n=2$, is weakly NP-hard, by a reduction from PARTITION \cite{bouveretLemaitre}. 
Weakly NP-hardness holds in our setting as well. Complementing this hardness result, we show how to compute the MMS in pseudo-polynomial time (a result previously known for the setting without sellable goods).

\subsection{Pseudo-Polynomial Algorithm for \texorpdfstring{$n=2$}{2 agents}}

We start by considering the $n=2$ (two agents) case. 
Formally, we assume that $v, p$ take integer values in $\{0, \ldots, K\}$ for some $K$, and we require that the algorithm is polynomial in the number of goods $m$ and $K$. Our main building block is the well-known dynamic programming algorithm for MMS computation with indivisible goods (and no selling), which we include here for completeness:

\begin{algorithm}
\SetAlgoLined
\DontPrintSemicolon
\KwIn{Valuation $v$ in range $[0,K]$ over goods $\items$}
\KwOut{The MMS with Indivisible Goods}

Let $\sigma = \sum_{g\in \items} v(g)$

Let $S = \emptyset$;

Let $T:\{0,\ldots, \lfloor \frac{\sigma}{2} \rfloor\} \times \{0, \ldots, m+1\} \rightarrow \{0,1\}$, and initialize $\forall j, T(0,j) = 1$. 

% If $T$ is passed parameters that are out of bounds, it returns $0$

\For{$j = 1$ to $m$} {
\For{$k=1$ to $\lfloor \frac{\sigma}{2} \rfloor$} {
    $T(k,j) = T(k, j-1) \lor T(k - v(g_j), j-1)$;
}
}
return $\argmax_c  T(c,m) = 1$.

 \caption{INDIVISIBLE-MMS-COMP}
 \label{alg:max_sum_achievable}
\end{algorithm}

The algorithm works as follows. It looks for the highest number smaller than $\frac{\sigma}{2}$, where $\sigma$ is the total sum of values, that can be attained by a subset of $\items$. This number $T$ will be the MMS, as its complementary will be $\sigma - T \geq \frac{\sigma}{2} \geq T$. The search is done by adding a new good for consideration at each iteration, and expanding the set of valid sums by considering each good in the set (which corresponds to a sum achieved by a subset) and either adding the good (and $v(g_j)$ to the sum) or not adding the good (and maintaining the existing sum of the subset as an option). The following lemma, which we state without proof, summarizes the guarantee of INDIVISIBLE-MMS-COMP.

\begin{lemma}
\label{lem:max-sum-achievable}
    INDIVISIBLE-MMS-COMP runs in time $O(m^2\cdot K)$ and returns the MMS for indivisible goods without selling. 
\end{lemma}

\begin{algorithm}
    \SetAlgoLined
\DontPrintSemicolon
\KwIn{Valuation $v$ and market price $p$ in range $[0,K]$ over goods $\items$}
\KwOut{MMS value for $v,p$ with $n=2$}

Let $\forall g, \bar{v}(g) = \max \{v(g), p(g)\}$.

% First, we try without any split items

$currentMMS = $ INDIVISIBLE-MMS-COMP($\bar{v}$);

\For{$j=1$ to $m$} {

    $v^{\neg j}(g) = \begin{cases} \bar{v}(g) & g\neq j \\ 0 & g = j\end{cases}$;

    $base = $ INDIVISIBLE-MMS-COMP($v^{\neg j}$);

    \If{$base \geq \frac{1}{2} \left(\sum_{g \in \items \setminus \{j\}} \bar{v}(g) - p(j)\right) $} {
        $currentMMS = \max \{ currentMMS, \frac{1}{2} \left(\sum_{g \in \items \setminus \{j\}} \bar{v}(g) + p(j)\right) \}$
    }
}

return $currentMMS$;

 \caption{COMP-MMS-n=2}
 \label{alg:MMS2}
\end{algorithm}

\begin{lemma}
    COMP-MMS-n=2 runs in time $O(m^3\cdot K)$ and returns the MMS value given $v, p$. 
\end{lemma}
\begin{proof}
    The running time holds since we perform a loop with $m$ steps where the main component is INDIVISIBLE-MMS-COMP, which runs in $O(m^2\cdot K)$ by Lemma~\ref{lem:max-sum-achievable}.

    Regarding correctness, recall that by Lemma~\ref{lem:number_of_split_goods}, with $n=2$ goods, it is w.l.o.g. to assume that there are at most one split good. We can therefore compute the MMS value conditional on each certain good being a split good, or no split goods at all, and take the maximum over these options. Moreover, by Lemma~\ref{lem:split_equal_mms}, the MMS value with split good $g_j$ equals exactly $MMS^j = \frac{1}{2}\left(\sum_{g\in \items \setminus \{j\}} \bar{v}(g) + p(j) \right)$. However, we must verify that this is indeed attainable. For this value to be attainable, we must be able to partition the goods other than $j$ into two bundles that are at most $p(j)$ apart. Then, we can use $p(j)$ to ensure that each bundle has the same value once we take the split good into consideration. Thus, the MMS value for the split good $g_j$ can be attained if and only if INDIVISIBLE-MMS-COMP over all goods other than $g_j$ is at least $MMS^j - p(j)$. This is the exact logic followed by COMP-MMS-n=2. 
\end{proof}

\subsection{Pseudo-Polynomial Algorithm for General Fixed \texorpdfstring{$n$}{n}}

We can further develop the argument above to compute the MMS in pseudo-polynomial time for any fixed $n$. Notice that since $3$-PARTITION is strongly NP-hard, we cannot  expect a pseudo-polynomial algorithm when $n$ is part of the input, unless $P=NP$. 

For this purpose, we must first significantly generalize INDIVISIBLE-MMS-COMP in the following way: While with $n=2$, the sum of goods smaller than $\frac{K}{2}$ uniquely determines its complement, with general $n$, our dynamic programming table is $n+1$-dimensional (one dimension for the number of goods used, corresponding to the columns in the $n=2$ case, and $n$ dimensions for the sum-structure of a partition to $n$ bundles). Moreover, in this case, we incorporate the divisible sale proceeds $p$ into Algorithm~\ref{alg:max_sum_achievable_general_n} as an additional input. After computing the dynamic programming table, for every achievable sum-structure, we translate it into a single number by a waterfall procedure using $p$. In the waterfall procedure, we are given bundle sum values, and ``water'' that can be distributed among them, and we gradually fill the bundles from the lowest to highest value to achieve the highest possible MMS value given the input. 

\begin{algorithm}
\SetAlgoLined
\DontPrintSemicolon
\KwIn{$S_1, \ldots, S_n$ integer sums, integer $p$ in range $[0, m\cdot K]$}
\KwOut{The maximin value of dividing sale proceeds $p$ to complement sums $S_1, \ldots, S_n$}

Sort $S_1, \ldots, S_n$ in ascending order;

$r = p$ // 

\For{$j=1$ to $n-1$} {
    \If{$r\geq j \cdot (S_{j+1} - S_j)$} {
        \For{$j' = 1$ to $j$} {
            $S_{j'} = S_{j+1}$
        }
        $r = r - j \cdot (S_{j+1} - S_j)$
    }
    \Else {
        return $S_1 + \frac{r}{j}$
    }
}

return $S_1 + \frac{r}{n}$

 \caption{WATERFALL-FILL}
 \label{alg:waterfall_fill}
\end{algorithm}

Recall that the $\tilde{O}$ notation means we disregard terms with logarithmic dependence on the problem parameters. 
The following lemma shows that waterfall-fill optimally distributes the ``water'' $p$. 

\begin{lemma}
    WATERFALL-FILL runs in $\tilde{O}(n)$ (as a problem with parameters $n,m,k$)
 and has the highest MMS value for an instance with bundle sums $S_1, \ldots, S_n$ and ``water'' $p$. 
\end{lemma}
\begin{proof}
    The running time is due to the main loop which runs at most $n-1$ times and performs simple arithmetic operations (which account for a logarithmic dependence on $m, K$, subsumed by the $\tilde{O}$ notation). 

    Correctness holds by induction over the step where the loop terminates. The inductive claim is that if we denote $S_1^{(j)}, \ldots, S_n^{(j)}$ for the sums after step $j$ (if no result is returned), then $S_1^{(j)} = \ldots S_j^{(j)}$, as well as $\forall j' > j, S_{j'}^{(j)} = S_{j'}$,
    and that the best attainable MMS given $S_1^{(j-1)}, \ldots, S_n^{(j-1)}$ and sale proceeds $p-r$ is $S_1^{(j)}$ (if the loop continues) or the return value, otherwise. 
    
    For the base case, if the loop terminates after one step, then $p < S_2 - S_1$, and thus
    $$S_1 + p < S_2 \leq S_3 \leq \ldots \leq S_n,$$
    and so if we add $p$ sale proceeds to bundle $1$ we attain an MMS of $S_1 + p$. On the other hand, 
     for any $p_1, \ldots, p_n \geq 0$ so that $\sum_{i=1}^n p_i = p$, it holds that 
$$\min \{S_1 + p_1, \ldots, S_n + p_n\} \leq S_1 + p_1 \leq S_1 + p,$$
and so this is the best attainable MMS given $S_1, \ldots, S_n, p$. 
If the loop continues after the first step, then $p - r = S_2 - S_1, S_1^{(1)} = S_2$, and we similarly conclude that this is the best attainable MMS.

For the inductive step, assume that $S_1^{(j)} = \ldots = S_j^{(j)}$. If the loop terminates at step $j+1$, then $r < j \cdot (S_{j+1} - S_j)$, and thus
$$S_1^{(j)} + \frac{r}{j} = S_2^{(j)} + \frac{r}{j} = \ldots S_j^{(j)} + \frac{r}{j} < S_{j+1}^{(j)} \leq \ldots \leq S_n^{(j)}, $$
and so if we add $\frac{r}{j}$ sale proceeds to each of the first $j$ bundles we attain an MMS of $S_1^{(j)} + \frac{r}{j}$. On the other hand, 
     for any $p_1, \ldots, p_n \geq 0$ so that $\sum_{i=1}^n p_i = r$, it holds that 
$$\min \{S_1 + p_1, \ldots, S_n + p_n\} \leq \min_{1\leq j' \leq j} S_{j'}^{(j)} + p_{j'} = S_1^{(j)} + \min_{1\leq j' \leq j} p_{j'} \leq S_1^{(j)} + \frac{r}{j},$$
and so this is the best attainable MMS given $S_1^{(j)}, \ldots, S_n^{(j)}, r$. If the loop continues after step $j$, then $r^{(j)} - r^{(j+1)} = j \cdot (S_{j+1}^{(j)} - S_j^{(j)}), S_1^{(j+1)} = \ldots = S_{j+1}^{(j+1})$, and we similarly conclude that this is the best attainable MMS given $S_1^{(j)}, \ldots, S_n^{(j)}, r^{(j)} - r^{(j+1)}$. 
    
\end{proof}

WATERFALL-FILL receives both the divisible resource $p$ and specific bundle sums. We next abstract one step further, and not include the bundle sums as input, but rather good values, together with a divisible resource $p$. Then, as we noted above, we combine a dynamic programming approach (as done in INDIVISIBLE-MMS-COMP for $n=2$) for general $n$, with an application of WATERFALL-FILL to handle the divisible resource. This results in Algorithm~\ref{alg:max_sum_achievable_general_n}. 

\begin{algorithm}
\SetAlgoLined
\DontPrintSemicolon
\KwIn{Number of agents $n$, valuation $v$ in range $[0,K]$ over goods $\items$, integer sale proceeds $p$}
\KwOut{The MMS value given $n, v, p$}

Let $S = \emptyset$;

Let $T:\{0,\ldots, m \cdot K\}^n \times \{0, \ldots, m+1\} \rightarrow \{0,1\}$, and initialize $\forall j, T(0,\ldots, 0) = 1$. 

% If $T$ is passed parameters that are out of bounds, it returns $0$

\For{$j = 1$ to $m$} {
\For{$s_1, \ldots, s_n \in \{0, \ldots, m \cdot K\}^n$} {
    $T(s_1, \ldots, s_n, j) = T(s_1 - v(g_j), s_2, \ldots, s_n, j-1) \lor T(s_1, s_2 - v(g_j), \ldots, s_n, j-1) \lor \ldots \lor T(s_1, s_2, \ldots, s_n - v(g_j), j-1)$
}
}

$currentMMS = 0$

\For{$s_1, \ldots, s_n \in \{0, \ldots, m \cdot K\}^n$} {
    \If{$T(s_1, \ldots, s_n, m) = 1$} {
        $currentMMS = \max \{currentMMS, WATERFALL-FILL(s_1, \ldots, s_n, p)\}$
    }

}
    return $currentMMS$

 \caption{MAX-SUM-ACHIEVABLE-GENERAL-$n$}
 \label{alg:max_sum_achievable_general_n}
\end{algorithm}

\begin{lemma}
\label{lem:max-sum-achievable-general}
    For a given valuation function $v$ over goods $\items$, and divisible sale proceeds $p$, MAX-SUM-ACHIEVABLE-GENERAL-$n$ returns the conditional 
    
    \noindent
    MMS value in time $\tilde{O}((m \cdot K)^n \cdot n \cdot m)$. 
\end{lemma}
\begin{proof}[Proof sketch.]
    The running time is due to the algorithm running over an array of size $(m \cdot K)^n \cdot m$ and performing a $\tilde{O}(n)$ operation at each step. 

    The correctness holds since the algorithm takes the maximum over the 
    
    \noindent
    WATERFALL-FILL value, which guarantees the best attainable MMS given sums $S_1, \ldots, S_n$ and divisible sale proceeds $p$, over all feasible sums, given by the dynamic program. 
\end{proof}

MAX-SUM-ACHIEVABLE-GENERAL-$n$ finds the MMS given a particular selling decision, which translates an instance to goods with value $v$, and a divisible resource $p$. Finally, we make the final abstraction to consider the original instance with a valuation $v$ and market price $p$ over $m$ goods in COMP-MMS-FIXED-$n$. We make use of Lemma~\ref{lem:number_of_split_goods} in that we only need to consider subsets of size up to $n-1$ of goods that are sold in the MMS partition, and use $\bar{v}$ for the remaining goods. 

\begin{algorithm}
    \SetAlgoLined
\DontPrintSemicolon
\KwIn{Number of agents $n$, valuation $v$ and market price $p$ in range $[0,K]$ over goods $\items$}
\KwOut{MMS value for $n,v,p$}

Let $\forall g, \bar{v}(g) = \max \{v(g), p(g)\}$.

$currentMMS = 0$

\For{$s = 0$ to $n-1$} {
\For{$S = \{j_1, \ldots, j_s\} \subseteq \items,  s.t. |S| = s$} {

    $currentMMS = \max \{ currentMMS, $ MAX-SUM-ACHIEVABLE-GENERAL-n$(n, \bar{v}, \items \setminus S, \sum_{g\in S} p(g))\}$
}
}

return currentMMS;

 \caption{COMP-MMS-FIXED-$n$}
 \label{alg:MMS_fixed_n}
\end{algorithm}

\begin{lemma}
    COMP-MMS-FIXED-$n$ computes the MMS for an instance of goods with selling in pseudo-polynomial time $\tilde{O}((\frac{e \cdot m}{n} \cdot m \cdot K)^{n-1} \cdot n \cdot m)$. 
\end{lemma}
\begin{proof}
    The running time follows by iterating over all choices of up to $n-1$ goods out of $m$, where we can upper bound the binomial sum \cite{alon2016probabilistic}:

$$\sum_{s=0}^{n-1} \binom{m}{s} \leq (\frac{e \cdot m}{n-1})^{n-1}.$$

This is multiplied by the running time of 

\noindent MAX-SUM-ACHIEVABLE-GENERAL-$n$ in each iteration, which is given by Lemma~\ref{lem:max-sum-achievable-general}. 

Regarding correctness, we use the fact that by Lemma~\ref{lem:number_of_split_goods}, an MMS output for $n$ agents has, w.l.o.g., at most $n-1$ split goods. Then, for any fixed choice of split goods $\sellable_C$, we sell them, and MAX-SUM-ACHIEVABLE-GENERAL-$n$ guarantees (through Lemma~\ref{lem:max-sum-achievable-general}) the best conditional MMS guarantee given the indivisible goods $\items \setminus \sellable_C$ and the divisible sale proceeds $\sum_{g\in \sellable_C} p(g)$. Taking the maximum over all these options results in the overall MMS value. 
\end{proof}

\subsection{A Polynomial-Time Approximation Scheme (PTAS) for MMS Computation for Goods With Selling}

We show how to reduce an instance of goods with selling to an instance of indivisible goods (without selling) for the purpose of a PTAS. We can then rely on the known PTAS for goods without selling \cite{woeginger}. 

\begin{theorem}
For the purpose of computing the MMS value and an associated MMS partition,
    given a PTAS for indivisible goods without selling, there is a PTAS for goods with selling.
\end{theorem}
\begin{proof}[Proof sketch]
Let ALG denote a PTAS for goods without selling.

In the instance with selling, 
    we sell all goods with $p(g) \ge v(g)$, and iteratively sell all goods with $p(g)$ higher than the TPS (recall that the TPS is computable in polynomial time). None of these sales can decrease the MMS. By multiplicative scaling, we may assume that the TPS of the resulting instance is~1. This also implies that the MMS is at least $\frac{1}{2}$. (This relation between MMS and TPS is known when there are no sellable goods. Its proof, by a simple bucket-filling algorithm, applies also when there are sellable goods.)
    
    Choose $k \simeq \frac{1}{\epsilon} \log \frac{1}{\epsilon}$, so that $\epsilon (1 + \epsilon)^k = 1$ (if needed, tweak the value of $\epsilon$ so that $k$ with this property exist). For every unsold good, we place it in one of $k + 2$ buckets. Bucket $B_0$ contains goods with $v(g) \leq \epsilon$. For $1 \le i < k$, bucket $B_i$ contains goods with $\epsilon \cdot (1+\epsilon)^{i-1} < v(g) \leq \epsilon \cdot (1+\epsilon)^i$. Bucket $B_k$ contains those goods with $v(g) > 1$.
    %Let us assume w.l.o.g. the maximum value of an item is normalized to $1$ \ygc{I'm actually not sure why it's ok to do that}, then it suffices to choose $k$ so $\epsilon \cdot (1 + \epsilon)^k = 1$, which yields $k = \frac{\log(\frac{1}{\epsilon})}{\log(1+\epsilon)} \in O(\frac{1}{\epsilon} \log(\frac{1}{\epsilon}))$. 

    For the purpose of obtaining an MMS partition, the benefit of selling goods for $B_0$ is negligible, as they never cause an imbalance larger than an additive term of $\epsilon$ among the bundles of the MMS partition.

    By Lemma~\ref{lem:number_of_split_goods}, we need to sell at most $n-1$ unsold goods. A {\em configuration} specifies how many goods from each of the $k+1$ buckets $B_1, \ldots, B_{k+1}$ need to be sold. The total number of possible configurations is at most $n^{k+1}$. Given a configuration $s = s_1, \ldots, s_{k+1}$, the $s_j$ goods that we sell from each bucket $B_j$ are those that have the highest price, because in terms of $v$ value, there is negligible difference among goods in the same bucket. This applies also to goods in $B_{k+1}$, because any unsold good in $B_{k+1}$ has value larger than the MMS, and for the purpose of achieving an MMS partition, they are interchangeable. 

    For a given configuration, we group the money obtained from selling goods into small auxiliary indivisible goods of value $\epsilon$ each, and add them to the instance. There are at most $O(\frac{m}{\epsilon})$ auxiliary goods. 
    
    Then, for each configuration we run ALG. Our output approximation for the MMS (and the MMS-partition) is the one given by the configuration for which ALG gave the highest MMS value. 

    %We account for sources of error in the MMS estimation. One is due to the $\epsilon$ granularity of the auxiliary items (rather than being continuous), which may cause an error up to $\epsilon$. Another error cause may be the $\epsilon$ of the PTAS of goods without selling. Lastly, our choice of items to sell from each bucket may have an error of up to $\epsilon$ each. Let $\eta$ be the amount of items sold. Each item sold has at most $\epsilon$ difference in value from any other item in their bucket. Thus, overall the difference is bounded by $\eta \cdot \epsilon$. By Lemma~\ref{lem:number_of_split_goods} we can assume at most $\eta + 1$ bundles in the resulting MMS partition contain sale proceeds, and by Lemma~\ref{lem:split_equal_mms}, all the bundles containing sale proceeds have equal values (up to $\epsilon$). Thus adding the $\eta \cdot \epsilon$ term must be split equally between all these bundles, and results in $\epsilon$ change in MMS overall. 

It is not hard to see that the algorithm sketched above gives a $(1 - O(\epsilon))$ approximation to the MMS. In terms of its running time, it calls ALG with each of at most $n^{k+1} \le n^{f(\epsilon})$ configurations, each time with a number of goods that is $O(\frac{m}{\epsilon})$. Here, $f(\epsilon) = {\frac{1}{\epsilon} \log(\frac{1}{\epsilon})}$. As ALG is a PTAS, so is our algorithm. 

    %We conclude that the error is bounded by $3\epsilon$, which is a constant expression of $\epsilon$, and so our approximation scheme achieves good approximation. In terms of run-time, it has polynomial dependence in $m, n$: If the goods without selling has run-time $f(\frac{1}{\epsilon}, m, n)$ (and, per our assumption, this is polynomial in $n,m$), then the run-time of our PTAS is: $O(n^{\frac{1}{\epsilon} \log(\frac{1}{\epsilon})} \cdot f(\frac{1}{\epsilon}, m, n))$. 

    %\ygc{Uri: I tried to translate your emails about the PTAS into this proof sketch, in particular I felt there was a point to be made about how to ensure the error resulting from bucketing is not dependent on $n$ or $m$, what do you think..? Maybe it's simpler than that, and we can just round all items down because by this we lose a multiplicative (rather than additive) factor of at most $(1-\epsilon)$.}
\end{proof}

\section{Extension: Chores with Outsourcing}

It is natural to translate our model to the chores setting, which we call ``Chores with Outsourcing''. In the chores setting, agents have subjective costs of performing different tasks, but there is also an option to pay a known cost for the task to be outsourced. The cost can be divisibly shared between the agents. We do not restate the entire model, but rather use $c$ (cost) functions instead of $v$ (valuation) functions notation, and $t$ (task) instead of $g$ (good) for the goods. We let $$\bar{c}_i(t) = \min \{c_i(t), p(t)\}.$$ Notice that $c$ and $p$ receive positive values, representing the \textit{dis-utility} of the allocation (which we aim to minimize). 
%\ufc{$\bar{c}_i(t) = \min \{c_i(t), p(t)\}$?}

How do our results translate to this setting? The $n=2$ case goes through well, and allows us to guarantee the MMS, using an analogous procedure to that of the goods with selling setting. 

\subsection{\texorpdfstring{$2$}{2}-MMS Approximation for Chores}

We need to consider throughout the algorithm that outsourced chores may be very costly relative to the MMS. This is potentially a problem, if we consider, as a starting point, the following natural algorithm for indivisible chores (without outsourcing). The algorithm is an adaptation of the bag-filling algorithm in \cite{apxMMS2017}, which was given for the goods setting, and gives a $2$-approximation. 

%Though it gives a $2$-approximation for indivisible chores (without outsourcing), it is different from the $2$-approximation round-robin algorithm of \cite{apxMMSChores}. \ufc{The logic of the last sentence is not clear. It is as if you assume that the reader assumes that all factor 2 approximations that are known are from \cite{apxMMSChores}. Just say that it gives a factor 2 approximation.. No reason to mention \cite{apxMMSChores}.}

\begin{algorithm}
    \SetAlgoLined
\DontPrintSemicolon
\KwIn{Cost functions $c_1, \ldots, c_n$ over chores $\items$, number of agents $n$}
\KwOut{An allocation $B$ with at most $2\cdot MMS$ cost for each agent}

\For{$j=1$ to $n$} {
    Let $B_j = \emptyset$.

    Do bag-filling for $B_j$ until all agents agree $c_i(B_j) \geq MMS_i$.

    Allocate $B_j$ to an agent that has $c_i(B_j) \leq 2\cdot MMS_i$. 
}

 \caption{APX-MMS-Indivisible-Chores-2}
 \label{alg:APX-MMS-Indivisible-Chores-2}
\end{algorithm}

The algorithm achieves the $2$-approximation because of an important property for indivisible chores MMS: No chore costs more than the MMS cost. If a chore were to cost more, then in any partition, the bundle that contains it would cost more than the MMS, including in the MMS partition, which is a contradiction. Thanks to this property, we are guaranteed that if for some agent a bundle's dis-utility is at most the MMS, then the dis-utility after adding any chore is at most $2 \cdot MMS$. This allows us to always be able to fill a bag that is worth at least the MMS cost for all agents, but still have an agent that values it at most at $2 \cdot MMS$%(the value before the addition of the last item, which was less than the MMS or otherwise we would stop earlier, in addition to at most the MMS cost from the last added item)
. However, as noted before, this property is no longer true in the chores with outsourcing setting, if we consider allocating complete outsourced tasks (i.e., adding the full costs of an outsourced tasks within a bundle). 

Still, with a careful adjustment and analysis, it is possible to achieve the $2$-approximation. The idea is as follows. At each round of bag-filling, we put aside chores that cause issues in the bag-filling process, i.e., before adding them some agents consider the bundle has below $MMS_i$ cost, and after adding them all agents think the bundle has above $2 MMS_i$ cost. If we end up putting all chores aside (except for those already in the bag) and are unable to complete the bag-filling as we wish, then we allocate the bundle as it is, and continue to the next round. This violates the bag-filling principle, in that some agents may consider the cost of the allocated bundle is too low (below $MMS_i$). However, the gist of the proof is that this violation only happens at most once in the perspective of any agent $i$, and this still allows for the argument to follow through. 

\begin{algorithm}
    \SetAlgoLined
\DontPrintSemicolon
\KwIn{Costs $c_1, \ldots, c_n$ and outsourcing price $p$ over chores $\items$, number of agents $n$}
\KwOut{An allocation $\mathbf{B}$ with at most $2\cdot MMS$ cost for each agent}

Let $AVAIL = \items$ be all the available goods, $ACTIVE = [n]$ all the active agents.

\For{$j=1$ to $n-1$} {
    Let $X = \emptyset, \Psi = \emptyset$. // $X$ is the bundle we bag-fill. $\Psi$ are all the chores we take out of circulation for the current round of bag-filing. 

    \While{$\Psi \neq AVAIL \setminus X$} {
    Add chores to $X$ from $AVAIL \setminus (X \cup \Psi)$, until all active agents $i$ agree that $\bar{c}_i(X) \geq MMS_i$. Let $t$ be the last chore added to $X$. 
    
    \If{ There is an agent $i$ so that $\bar{c}_i(X) \leq 2\cdot MMS_i$} {

        %Add all chores of zero cost to $i$ in $AVAIL \setminus X$ to $X$. 

        $B_i = X, ACTIVE = ACTIVE \setminus \{i\}, AVAIL = AVAIL \setminus B_j$. 
        
    } \Else {
        $\Psi = \Psi \cup \{t\}, X = X \setminus \{t\}$.
    }

    }

    \If{$X \neq \emptyset$} {
    Allocate $X$ to an agent $i$ that has $c_i(X) \leq 2\cdot MMS_i$, and let $ACTIVE = ACTIVE \setminus \{i\}, AVAIL = AVAIL \setminus X$. 
    } \Else{

        Outsource all chores in AVAIL, share equally among all ACTIVE agents, and return. 
    }

}

Allocate all chores in AVAIL to the last remaining agent in ACTIVE, and return.

 \caption{APX-MMS-Chores-With-Outsourcing-2}
 \label{alg:APX-MMS-Chores-With-Outsourcing-2}
\end{algorithm}

\begin{lemma}
    APX-MMS-Chores-With-Outsourcing-2 guarantees a 
    
    \noindent
    $2$-approximation for all agents. 
\end{lemma}
\begin{proof}
    
    All chores are allocated by the end of the algorithm, because either a bundle is allocated at each round of the for loop, or there is a round where all remaining chores are outsourced and the sale proceeds are split among the active agents. If the latter happens, then clearly all chores are allocated. If the former happens, then the remaining agent in ACTIVE takes all remaining chores. 

    Let us consider the approximation guarantee. It is clear by the IF conditions that any agent allocated a bundle within the bag-filling loop (rather than shares the outsourcing cost) has a cost of at most $2 \cdot MMS$. We thus focus on the last agent that is allocated a bundle outside the loop, or on the agents that share the outsourcing cost (the two possible return conditions, of line 17 and 20). 

    Let us first consider the last agent $i$ that is allocated a bundle in line 20. If all bundles of other agents have a cost of at least $\bar{c}_i(X) \geq MMS_i$, then overall, $\bar{c}_i(\cup_{j\neq i} B_j) = \sum_{j\neq i} \bar{c}_i(B_j) \geq (n-1) \cdot MMS_i$. Since it also holds that $n \cdot MMS_i \geq \bar{c}_i(\items)$, we conclude that $\bar{c}_i(\items \setminus (\cup_{j\neq i} B_j)) \leq MMS_i$. 

    If for some other agent $j$ it holds that $\bar{c}_i(B_j) < MMS_i$, then it must be allocated a bundle within the second IF condition in line 14 (as it is prohibited by the pre-condition of the first IF condition). Take the first such bundle that is allocated within the second IF condition.
    This means that when agent $j$ is allocated bundle $B_j$, it holds that $B_j \cup \Psi = AVAIL$, and for any $t \in \Psi$, $\bar{c}_i(B_j \cup \{t\}) > 2 MMS_i$. We thus conclude that $\bar{c}_i(t) > MMS_i$ for any $g \in \Psi$. Since $\Psi$ is all the available chores from this round on-wards, $\bar{c}_i(t) > MMS_i$ subsequently holds in all future rounds for any chore in AVAIL. This means that any other bundle $X$ that is allocated, since it contains some chore $t\in AVAIL$, has $\bar{c}_i(X) > MMS_i$. Overall, we conclude that for any other agent $\ell \neq j$, $\bar{c}_i(B_{\ell}) \geq MMS_i$, and we can repeat our previous argument to show that $\bar{c}_i(\items \setminus (\cup_{\ell \neq i} B_{\ell})) \leq \bar{c}_i(\items \setminus (\cup_{\ell \not \in  \{j,i\}} B_{\ell}))\leq n \cdot MMS_i - (n-2) MMS_i = 2MMS_i$. 

    Now consider if the outsourcing step of line 17 takes place. This step only happens if $X = \emptyset$. Since the condition happens after the While loop is terminated, it must be that $\Psi = AVAIL \setminus X = AVAIL$. Thus, all remaining agents consider all remaining chores $t\in AVAIL$ to have $\bar{c}_i(t) > 2MMS_i$ (otherwise, either the IF condition of line 6 would be invoked, or $X$ would not be empty). Since a chore that has $\bar{c}_i(t) > MMS_i$ must be outsourced in the MMS partition, we conclude that all agents outsource the remaining chores in their MMS partition. Following our previous arguments, for any agent $i \in ACTIVE$, it holds that if $k$ rounds of bag-filling have been previously completed, $p(AVAIL) + \bar{c}_i(\cup_{j=1}^k B_j) = p(AVAIL) + \sum_{j=1}^k \bar{c}_i(B_j) \leq n \cdot MMS_i$, and $\sum_{j=1}^k \bar{c}_i(B_j) \geq (k-1) MMS_i$, and so $p(AVAIL) \leq (n-k+1) MMS_i$. 
$k$ rounds previously completed also implies there are also $n-k$ agents in ACTIVE. Overall, each agent is assigned outsourcing costs of $\frac{n-k+1}{n-k} MMS_i \leq \frac{3}{2} MMS_i$, where the last inequality holds because $1\leq k \leq n-2$. 
\end{proof}

\subsection{Envy Notion Definitions For Chores}
%Defining the envy notions for chores with outsourcing has some subtlety, and for this reason we develop them explicitly. Let us focus on the SEFX notion in order to develop the correct formula (SEF1 can be similarly deduced). In the setting with chores (without outsourcing), for EFX, we require envy to disappear if we remove any of the chores allocated to the agent. But in the case with outsourcing, how does the outsourcing cost factor in? It would seem, in an analogue from goods with selling, that we should add the outsourcing cost, after removing the chore. However, this is either the same, or worse, than $\bar{c}_i$ of the chore, and thus this would yield a notion equivalent to EF. 

%\ufc{Very long introduction, but the definition is fine. I would start with the definition and then explain that the rational is that if $i$ can outsource a chore and share the cost with $j$ in a way that she prefers each of the new bundles over her original bundle, then she should not have received that chore in full.}

%Finally, we conclude with the formal definition for SEFX and SEF1 that follow this approach:

We provide a definition of envy notions for chores.

\begin{definition}
\label{def:chores_EF}
    A chores allocation $(A_1, P_1), \ldots, (A_n, P_n)$ is {\em envy free} (EF) if for every agent $i$ and other agent $j$, {$\bar{c}_i(A_i) + P_i \le \bar{c}_i(A_j) + P_j$.} %\ufc{what is $u_i$?} 
    {An allocation satisfies the following relaxations of EF if the above envy free condition holds whenever $P_i > 0$, and either the EF condition or the following conditions hold if $P_i = 0$.}
    \begin{enumerate}
        \item SEF1 (strong envy free up to one chore). $p(t) - c_i(t) \ge \bar{c}_i(A_i) - \bar{c}_i(A_j) - P_j$ for some chore $t \in A_i$.
         \item SEFX (strong/sellable envy free up to any chore). $p(t) - c_i(t) \ge \bar{c}_i(A_i) - \bar{c}_i(A_j) - P_j$ for \textit{every} chore $t \in A_i$.
    \end{enumerate}
\end{definition}

At a surface level, this definition looks very unlike the EF1 definition for chores (where you remove one chore from the agent's bundle, and compare with the other agent's bundle). We explain why it makes intuitive sense. The immediate analogue in the settings of chores with outsourcing, would be to remove a chore, and replace it with its outsourcing costs. However, this would be the same, or worse, than $\bar{c}_i$ of that chore, and so this would yield a notion equivalent to EF. 

A better way to go about it is to consider the local pairwise maximin principle that guided us in the goods with selling case. There, an agent was willing to accept having a worse-off bundle, if they know that a localized change (selling the good) makes the other bundle worse-off. Otherwise, it could be argued that a localized change can improve the maximin of the two bundles: We sell a good of the other bundle, recover the same value as $B_i$, and are left with some more proceeds that can be shared equally among $B_i$ and $B_j$. 

Similarly, in the chores with outsourcing setting, for SEFX to hold, we should forbid the following scenario: We outsource a chore $t$ from $B_i$, use water-filling to levy the costs on $B_j$ and $B_i \setminus \{t\}$, and end-up with a lower cost than $B_i$, since this would constitute a localized maximin improvement. As an inequality, this takes the form:

$$\frac{\bar{c}_i(B_i \setminus \{t\}) + \bar{c}_i(B_j) + p(t)}{2} \geq \bar{c}_i(B_i).$$

By changing wings, we can translate it to the condition:

$$p(t) - \bar{c}_i(t) \geq \bar{c}_i(B_i) - \bar{c}_i(B_j).$$

Now, if $\bar{c}_i(t) = p(t)$, this is just the envy-freeness condition. Otherwise, it must be that $\bar{c}_i(t) = c_i(t)$, and so:

$$p(t) - c_i(t) \geq \bar{c}_i(B_i) - \bar{c}_i(B_j).$$

This yields the nice ``delta'' form of Definition~\ref{def:chores_EF}: The difference between outsourcing and in-house costs should be greater than the envy towards the other agent's bundle.

\section{Extension: Equal Proceeds Case}
\label{sec:equal_proceeds}

Our goal in this section is to consider the impact of imposing an additional restriction over the allocation, where we require that sale proceeds are equally distributed among all agents. This may be due to a regulatory or legal constraint, that attempts to impose fairness on the allocation. As the sale proceeds have global objective value, it stands to reason that they should be equally distributed, but as we see, this may significantly worsen the overall fairness of the allocation. 

Given the additional restriction, we adjust the MMS value an agent may expect, as it can only come from allocations with equal sale proceeds. This may significantly reduce the MMS the agent expects. 

\begin{example} \label{ex:mms_equal_proceeds_gap}

    Consider $m=n$ goods, in the perspective of some fixed agent with the valuations and market price given in Table~\ref{tab:mms_equal_proceeds_gap}. Then, the equal-proceeds MMS value is $\frac{1}{n}$, while the unconstrained MMS value is $1$. 
\end{example}

\begin{table}[ht]
    \centering
    \renewcommand{\arraystretch}{1.3} % Adds vertical padding for fractions
    \begin{tabular}{lcccc}
        \toprule
        & $g_1$ & $\ldots$ & $g_{n-1}$ & $g_n$ \\
        \midrule
        Subjective Valuation $v$ & $1$ & $\ldots$ & $1$ & $0$  \\
        \midrule
        Market price $p$ & $0$ & $\ldots$ & $0$ & $1$ \\ 
        \bottomrule
    \end{tabular}
    \caption{Subjective valuation and market price for Example~\ref{ex:mms_equal_proceeds_gap}}
    \label{tab:mms_equal_proceeds_gap}
\end{table}

The following example shows that no more than a $\frac{1}{n}$ approximation of the equal-proceeds MMS can be guaranteed. 

\begin{example}
\label{ex:equal_proceeds}
Consider $m=n$ goods, with the valuations and market price given in Table~\ref{tab:equal_proceeds}. For agent $1$, the optimal outcome is selling all goods, resulting in $MMS_1 = 1$. For all other agents, the optimal outcome is keeping all goods, each in a singleton bundle, resulting in $MMS_i = L$. 

Now, consider the possible outcomes. If at most one good is sold, then agent $1$ gets at most $1 \cdot \frac{1}{n} + (n-1) \cdot \frac{1}{L} \approx \frac{1}{n} \cdot MMS_1$, when $L$ is sufficiently large. If at least two goods are sold, then some agent $i > 1$ does not get to keep any good, by the pigeonhole principle. They get $\frac{2}{n} = \frac{2}{n \cdot L} MMS_i$, which is arbitrarily low with large $L$. 
\end{example}

\begin{table}[ht]
    \centering
    \renewcommand{\arraystretch}{1.3} % Adds vertical padding for fractions
    \begin{tabular}{lcccc}
        \toprule
        & $g_1$ & $\ldots$  & $g_n$ & MMS \\
        \midrule
        Agent $1$ & $\frac{1}{L}$ & $\ldots$ & $\frac{1}{L}$ & $1$  \\
        Agent $2$ & $L$ & $\ldots$ & $L$ & $L$  \\
        $\ldots$ \\
        Agent $n$ & $L$ & $\ldots$ & $L$ & $L$  \\
        \midrule
        Market price $p$ & $1$ & $\ldots$ & $1$ & -- \\ 
        \bottomrule
    \end{tabular}
    \caption{Subjective valuations and market prices for Example~\ref{ex:equal_proceeds}. We consider $L\geq 1$. }
    \label{tab:equal_proceeds}
\end{table}

We show that this is tight using a bag-filling algorithm that guarantees $\frac{1}{n}$ of the TPS, as defined in Definition~\ref{def:tps_selling}. Notice that this is a definition of the TPS with unconstrained proceeds, which is always larger than the MMS with unconstrained proceeds, by Lemma~\ref{lem:mms_tps_ps}. Thus, this shows that there is no ``double discounting'': Example~\ref{ex:mms_equal_proceeds_gap} shows the equal-proceeds MMS can be $\frac{1}{n}$ of the unconstained MMS, and Example~\ref{ex:equal_proceeds} shows we can only guarantee $\frac{1}{n}$ of the equal-proceeds MMS. However, these are two different examples, and so our result shows they cannot be combined to create a $\frac{1}{n^2}$ gap. 

Let us describe the bag-filling algorithm. We use the notation $$v'_i(e) = \max[\min[TPS_i,v_i(e)], p(e)].$$ Let $S_i = \{e \text{ | } p(e) > \min[TPS_i,v_i(e)]\}, K_i = \items \setminus S_i$, i.e., the set of items the agent prefers to sell according to $v'_i$ and the set of items the agent prefers to keep. Notice that for the purpose of $v'_i, S_i, K_i$, we ignore the fact that upon sale, the agent only keeps $\frac{1}{n}$ of the sale proceeds. We also introduce the notation $\tilde{v}_i(e) = \begin{cases} \frac{1}{n}v'_i(e) & e \in S_i \\ v'_i(e) & e \in K_i.\end{cases}$.

%\ufc{For the sake of TPS, the valuation of agent $k$ is $v'_k(e) = \max[\min[TPS_k,v_k(e)], p(e)]$. This definition also partitions the items into those that $k$ intends to sell, and those that $k$ intends to keep (breaking ties arbitrarily). Write the algorithm below with $v'$ instead of $v$. This  will make the description less ambiguous and shorter. In particular, step (3) becomes.

%Else, start filling a bag $B$ with items until some agent $i$  declares that $v'_i(B) \ge \frac{1}{n} TPS_i - s_r$.   Here, every item $e \in B$ that $i$ intends to sell contributes $\frac{p(e)}{n}$ to $v'_i(B)$. If there is more than one such agent, then the winner of $B$ is the agent $i$ that generates the highest amount of sale proceeds from $B$. Let $s_i$ denote the amount of these sales proceeds. $B$ is given to $i$, who keeps those items that she intended to keep and sells the others. Agent $i$ becomes inactive, and $s_r, B_i$ are updated accordingly. } 

%\ygc{Define v', $\tilde{v}$ or whatever to mark the different valuation functions we use in different parts of the algorithm. }

%\ygc{If $i$ generates less proceeds than $k$, then this implies $k$ didn't value the whole thing in terms of $\tilde{v}$ as more than $\frac{1}{n} TPS$, so not more than $TPS$ in terms of v'. And the other case is the one we know where both agents find the bundle acceptable. }

Let $s_r$ denote the sales proceeds that each agent holds at the beginning of round $r$ (by equal-proceeds, all agents hold the same amount). Initially, agents are not assigned any items, which we denote by $B_i = \emptyset$ for every agent $i$. 
In each round, we do the following. 
We define an agent $i$ as active if it holds a value of less than $\frac{1}{n} \cdot TPS_i$ in goods assigned to it and sale proceeds $s_r$. I.e., an agent is active if $v_i(B_i) + s_r < \frac{1}{n} \cdot TPS_i$. 

Initially, all agents are active. In each round, we ensure at least one agent becomes inactive. We do so by the following sequence of rules: %Note also that several agents may become passive at the same round (as sale proceeds are distributed among all agents), but for simplicity of the presentation, we present the algorithm as if this does not happen.

\begin{enumerate}
    \item If there is an item $e$ and an active agent $i$ with $p(e) \ge TPS_i$, sell $e$. Update $s_r = s_r + \frac{1}{n} \cdot p(e)$.
    \item Else, if there is an item $e$ and active agent $i$ with $v_i(e) \ge \frac{1}{n}TPS_i - s_r$, give $e$ to agent $i$, i.e., set $B_i = B_i \cup \{e\}$. 
    \item Else, start filling a bag $B$ with items until some agent $i$  declares that $\tilde{v}_i(B) \ge \frac{1}{n} TPS_i - s_r$. %  Here, every item $e \in B$ that $i$ intends to sell contributes $\frac{p(e)}{n}$ to $v'_i(B)$. 
    If there is more than one such agent, then the winner of $B$ is the agent $i$ that generates the highest amount of sale proceeds from $B$. Let $s_i$ denote the amount of these sales proceeds. $B$ is given to $i$, who keeps those items that she intended to keep and sells the others. Agent $i$ becomes inactive, and $s_r, B_i$ are updated accordingly.
    %Else, start filling a bag $B$ with items until some agent $i$  declares that $v_i(B) \ge \frac{1}{n} TPS_i - s_r$.   $v_i(B)$ is computed by letting $i$ decide which items to keep and which to sell,  where from sold items she keeps only $\frac{1}{n}$ of the sales proceeds. If there are several ways to decide where $v_i(B) \ge \frac{1}{n} TPS_i - s_r$, \ufc{I did not understand the beginning of the sentence. Is there a typo there?} we assume that the agent maximizes the amount of total sale proceeds. 
    %If there is more than one such agent, then the winner of $B$ is the agent $i$ that generates the highest amount of sale proceeds from $B$. Let $s_i$ denote the amount of these sales proceeds. $B$ is given to $i$, who keeps those items that she intended to keep and sells the others. \ufc{What does the last sentence mean?} $i$ becomes inactive, and $s_r, B_i$ are updated accordingly. 
    \end{enumerate} 

\begin{theorem}
The above algorithm guarantees $\frac{1}{n} \cdot TPS_i$ to every agent $i$. 
\end{theorem}

%Let $t_{i} = TPS_{i}$. We rely on the effective value $v'_{i}(e) = \max\{p(e), \min\{v_{i}(e), t_{i}\}\}$. Let $K_{i}$ be the items agent $i$ prefers to keep ($p(e) \le \min\{v_{i}(e), t_{i}\}$), and $S_{i}$ the items they prefer to sell. Under equal proceeds, the net value agent $i$ receives from an item is $\tilde{v}_{i}(e) = v'_{i}(e)$ for $e \in K_{i}$, and $\frac{p(e)}{n}$ for $e \in S_{i}$. 
%Note that $v'_{i}(B) = v'_{i}(B \cap K_{i}) + p(B \cap S_{i})$, so pointwise $v'_{i}(e) \le n \cdot \tilde{v}_{i}(e)$.
\begin{proof}
Let $R$ be the set of available goods, $s$ be the money each agent currently holds, and $T$ be the set of active agents. We track the algorithm's progress using the following invariant for every $k \in T$:
\[ v'_{k}(R) + n \cdot s \ge |T| \cdot TPS_{k} \]

Initially, $R = \mathcal{M}$, $s = 0$, and $|T| = n$. The invariant holds with equality because $\sum_{e \in \mathcal{M}} v'_{k}(e) = n \cdot TPS_{k}$. In each round, we ensure at least one agent becomes inactive.

\begin{enumerate}
    \item Rule 1: If $p(e) \ge TPS_{i}$ for some active $i$ and $e \in R$: Sell $e$ and update $s \leftarrow s + \frac{p(e)}{n}$. Agent $i$ reaches the $\frac{TPS_{i}}{n}$ threshold and becomes inactive. For any other active $k$, the effective loss minus returned money is $v'_{k}(e) - p(e) \le \min\{v_{k}(e), TPS_{k}\} \le TPS_{k}$.
    
    \item Rule 2: Else, if $v_{i}(e) \ge \frac{TPS_{i}}{n} - s$ for some active $i$ and $e \in R$: Give $e$ to $i$. Since Rule 1 failed, $p(e) < TPS_{k}$, so the effective loss to any other active $k$ is $v'_{k}(e) = \max\{p(e), \min\{v_{k}(e), TPS_{k}\}\} \le TPS_{k}$.
    
    \item Rule 3 (Bag-filling): Else, begin bag-filling. As $v'_{k} \le n \cdot \tilde{v}_{k}$, the invariant ensures $n \cdot \tilde{v}_{k}(R) + n \cdot s \ge |T| \cdot TPS_{k}$, meaning $\tilde{v}_{k}(R) + s \ge \frac{TPS_{k}}{n}$. Thus, a bag filled with all of $R$ must trigger a claim. We can thus find a \emph{minimal} bag $B \subseteq R$ where a set of active agents $W$ declares $\tilde{v}_{i}(B) \ge \frac{TPS_{i}}{n} - s$. Let winner $i \in W$ be the one who maximizes the generated proceeds $d = p(B \cap S_{i})$. Agent $i$ takes $B$, sells $B \cap S_{i}$, keeps $B \cap K_{i}$, and becomes inactive. We update $s \leftarrow s + \frac{d}{n}$.
\end{enumerate}

For any remaining active agent $k \neq i$, let $a = v'_{k}(B \cap K_{k})$ and $b = p(B \cap S_{k})$. The effective value lost from the pool is $v'_{k}(B) = a + b$. $k$'s net loss, accounting for the system-wide proceeds $d$ returned to the pool, is $v'_{k}(B) - d$.

\begin{itemize}
    \item If $b \le d$ (winner sells at least as much as $k$ would): The net loss is $a + b - d \le a \le a + \frac{b}{n} = \tilde{v}_{k}(B)$. By the minimality of $B$ and the failure of prior rules, $\tilde{v}_{k}(B) < \frac{2TPS_{k}}{n} \le TPS_{k}$ (as $n \ge 2$).
    
    \item If $b > d$ (winner sells less than $k$ would): By the tie-breaking rule, $k \notin W$, so $\tilde{v}_{k}(B) < \frac{TPS_{k}}{n} - s$. Expanding this gives $a + \frac{b}{n} < \frac{TPS_{k}}{n} - s$. Multiplying by $n$ yields $n \cdot a + b < TPS_{k} - n \cdot s$. The net loss is $a + b - d \le a + b \le n \cdot a + b < TPS_{k} - n \cdot s \le TPS_{k}$.
\end{itemize}

In all cases, the net loss to any active agent $k$ is bounded by $TPS_{k}$. Updating the invariant for the remaining items $R'$ and new shared money $s'$ yields: \[ v'_{k}(R') + n \cdot s' \ge (v'_{k}(R) + n \cdot s) - TPS_{k} \ge |T| \cdot TPS_{k} - TPS_{k} = (|T| - 1) \cdot TPS_{k} \]. Since $|T|$ decreases by at least one each round, the algorithm terminates in at most $n$ rounds, successfully guaranteeing every agent at least $\frac{t_{i}}{n}$.
\end{proof}

\end{document}